\documentclass[11pt, a4paper]{article}
\usepackage{amsmath}
\usepackage{amsfonts}
\usepackage[parfill]{parskip}
\usepackage{fullpage}
\usepackage{cite}
\usepackage[normalem]{ulem}
\usepackage{enumitem}
\usepackage{amsthm}
\usepackage{breqn}
\usepackage{relsize}
\usepackage{graphicx}
\usepackage{tikz}
\usetikzlibrary{decorations.markings}

\tikzset{->-/.style={decoration={
  markings,
  mark=at position #1 with {\arrow{>}}},postaction={decorate}}}

\tikzset{
  big arrow/.style={
    decoration={markings,mark=at position 1 with {\arrow[scale=1.6,#1]{>}}},
    postaction={decorate},
    shorten >=0.4pt}}

\newcommand{\beq}{\begin{equation}}
\newcommand{\eeq}{\end{equation}}
\newcommand{\ba}{\begin{array}}
\newcommand{\ea}{\end{array}}
\newcommand{\sech}{\mathrm{sech}}
\newcommand{\sn}{\mathrm{sn}}
\newcommand{\cn}{\mathrm{cn}}
\newcommand{\dn}{\mathrm{dn}}

\newtheorem{thm}{Theorem}[section]
\newtheorem{lemma}[thm]{Lemma}

\newcommand\UW{Department of Applied Mathematics,\\ University of Washington,\\ Seattle, WA 98195-3925, USA}

\definecolor{UWpurple}{rgb}{.35,.01,.55}
\definecolor{BoxCol}{rgb}{.35,.01,.55}

\title{The Stability Spectrum for Elliptic Solutions to the Focusing NLS Equation}
\date{\today}
\author{Bernard Deconinck and Benjamin L. Segal\\
\\
\UW}

\begin{document}
\maketitle

\begin{abstract}
We present an analysis of the stability spectrum of all stationary elliptic-type solutions to the focusing Nonlinear Schr\"{o}dinger equation (NLS). An analytical expression for the spectrum is given. From this expression, various quantitative and qualitative results about the spectrum are derived. Specifically, the solution parameter space is shown to be split into four regions of distinct qualitative behavior of the spectrum. Additional results on the stability of solutions with respect to perturbations of an integer multiple of the period are given, as well as a procedure for approximating the greatest real part of the spectrum.
\end{abstract}

\section{Introduction}

The focusing, one-dimensional, cubic  Schr\"{o}dinger equation (NLS) is given by
\beq\label{fNLS} i \Psi_t +\frac{1}{2}\Psi_{xx}+\Psi |\Psi|^2=0. \eeq
In the context of water waves, nonlinear optics, and plasma physics, $\Psi(x,t)$ represents a complex-valued function describing the envelope of a slowly modulated carrier wave in a dispersive medium \cite{S07,Z68,Kivshar03,C84}. The equation also arises in the description of Bose-Einstein condensates \cite{P61,G61}, where $\Psi$ represents a mean-field wave function.

We begin by looking at stationary solutions to (\ref{fNLS}) in the form
\beq \label{statsolns} \Psi = e^{-i \omega t} \phi(x). \eeq
Then $\phi(x)$ satisfies
\beq \label{stationaryeqn} \omega \phi=-\frac{1}{2}\phi_{xx}-\phi |\phi|^2. \eeq
The stationary solutions we study in this paper are the elliptic solutions to this equation and their limits. These solutions are periodic or quasi-periodic in $x$ and limit to the well-known soliton solution as their period goes to infinity. More details on the periodic and quasi-periodic solutions relevant to our investigation are presented in Section \ref{solutions}.

Rowlands \cite{R74} was the first to study the stability of elliptic solutions. Using regular perturbation theory, treating the Floquet parameter as a small parameter, he conjectured that the stationary periodic solutions to focusing NLS are unstable. Since he expanded in a neighborhood of the origin of the spectral plane, his calculations suggest modulational instability for elliptic solutions of focusing NLS.  More recently, Gallay and H{\u{a}}r{\u{a}}gu{\c{s}} \cite{GH} examined the stability of small-amplitude elliptic solutions, with respect to arbitrary periodic and quasiperiodic perturbations. In a second paper \cite{GH2}, using the methods of \cite{GSS1, GSS2}, they proved that periodic and quasiperiodic solutions are orbitally stable with respect to disturbances having the same period. Also, they showed that the cnoidal wave solutions (see below) are stable with respect to perturbations of twice the period. H{\u{a}}r{\u{a}}gu{\c{s}} and Kapitula \cite{HK} put some of these results in the more general framework of determining the spectrum for the linearization of an infinite-dimensional Hamiltonian system about a spatially periodic traveling wave. For the quasi-periodic solutions of sufficiently small amplitude, they establish spectral instability.
Following this, Ivey and Lafortune \cite{IL} examine the spectral stability of cnoidal wave solutions to focusing NLS with respect to periodic perturbations, using the algebro-geometric framework of hyperelliptic Riemann surfaces and Riemann theta functions \cite{belokolos}. Their calculations make use of the squared eigenfunction connection as do ours below. Additionally, they use a periodic generalization of the Evans function. This gives an analytical description of the spectrum for the cnoidal wave solutions, which we replicate in this paper using elliptic functions. Lastly, we mention a recent paper by Gustafson, Coz, and Tsai \cite{GCT}. In this paper, the authors give a rigorous version of the formal asymptotic calculation of Rowlands to establish the linear instability of a class of real-valued periodic waves against perturbations with period a large multiple of their fundamental period. They achieve this by directly constructing the branch of eigenvalues using a formal expansion and the contraction mapping theorem. In terms of elliptic function solutions, their results are limited to the cnoidal and dnoidal solutions. Using entirely different methods, we confirm their results and extend their findings to nontrivial-phase solutions, in effect making the results of Rowlands rigorous for all elliptic solutions of NLS.

In Sections \ref{setup}, \ref{lax-section} and \ref{sqeig-section}, using the same methods as \cite{BD,BDN,DN}, we exploit the integrability of (\ref{fNLS}) to associate the spectrum of the linear stability problem with the Lax spectrum using the squared eigenfunction connection \cite{AKNS}. This allows us to obtain an analytical expression for the spectrum of the operator associated with the linearization of (\ref{fNLS}) in the form of a condition on the real part of an integral over one period of some integrand. However, unlike in \cite{BD,BDN,DN} the linear operator associated with the focusing NLS equation is not self adjoint. The self adjointness of the linear operator was directly exploited in these papers and that is not available here. Instead, we proceed by integrating the integrand explicitly. This is done in Section \ref{Weierstrass}. Next, using the expressions obtained, we prove results concerning the location of the stability spectrum on the imaginary axis in Section \ref{imagaxis}. In Section \ref{regions}, we present analytical results about the spectrum, and we make use of the integral condition to split parameter space into different regions where the spectrum shows qualitatively different behavior. In Section \ref{subharmonic} we examine the spectral stability of solutions against perturbations of an integer multiple of their fundamental period confirming and extending results of \cite{GCT,GH,GH2}. Finally, in Section \ref{approximations} we discuss approximations to the spectral curves in $\mathbb{C}$ found by expanding around known spectral elements. We use those approximations to give estimates for the maximal real part of the spectrum.

\section{Elliptic solutions to focusing NLS}\label{solutions}
The results of this section are presented in more detail in \cite{C00foc}. We limit our analysis to just what is necessary for the following sections.
We split $\phi$ into its amplitude and phase
\beq \label{phi} \phi(x)=R(x)e^{i \theta(x)}, \eeq
where $R(x)$ and $\theta(x)$ are real-valued, bounded functions of $x\in\mathbb{R}$.
Substituting (\ref{phi}) into (\ref{stationaryeqn}), we find the standard Jacobi elliptic function solutions given by
\begin{align}
R^2(x) & = b - k^2 \sn^2(x,k), \\
\omega & = \frac{1}{2}(1+k^2)-\frac{3}{2}b,\\
\label{thetaeqn} \theta(x)&=\int_0^x \frac{c}{R^2(y)}dy, \\
c^2 & =b(1-b)(b-k^2).
\end{align}
Here $\sn(x,k)$ is the Jacobi elliptic $\sn$ function with elliptic modulus $k$ \cite{BF,L13,WW,O}. Besides $k$, the only other parameter present is $b$, which is an offset parameter for the solutions. We are not specifying the full class of parameters allowed by the four Lie point symmetries of (\ref{fNLS}) \cite{O00}. Specifically, we are neglecting to include a scaling and a horizontal shift in $x$. The use of a Galilean shift allows for the application of our results to traveling wave solutions as well stationary solutions. The different symmetries are not included here as they do not produce qualitatively different results to what is covered here, but the methods presented apply equally well. 

In order for our solutions to be valid, we require that both $R^2(x)$ and $c^2$ are real, positive, and bounded. These conditions result in constraints on our parameters:
\begin{align}
0\le  k< 1, \\
k^2\le b\le1.
\end{align}

Of special importance is the boundary of this region, as many of the well-studied solutions to (\ref{fNLS}) lie on the boundary. Specifically, when $b=k^2$ or when $b=1$ we have that $c=0$ and $\phi(x) = \cn(x,k)$ or $\phi(x)=\dn(x,k)$ respectively. Here $\cn(x,k)$ and $\dn(x,k)$ are the Jacobi elliptic $\cn$ and $\dn$ functions respectively.
These solutions are called trivial-phase solutions, as there is no phase in $\phi(x)$, \textit{i.e.,} $c=0$. These solutions are periodic in $x$, of period $4K(k)$ and $2K(k)$ respectively. Here
\beq K(k) = \int_0^{\pi/2} \frac{1}{\sqrt{1-k^2 \sin^2 y}} dy, \eeq
is the complete elliptic integral of the first kind. As $k\rightarrow 1$ these solutions approach the well-studied stationary soliton solution of (\ref{fNLS}): $\phi(x) = \sech(x)$.

The other part of the boundary of parameter space occurs when $k=0$. Here the amplitude of $\phi(x)$ is constant and thus the analysis of the solutions simplifies greatly. These solutions are called Stokes wave solutions. They have the form $\phi(x)=\sqrt{b}\exp\left(i x \sqrt{1-b}\right)$. The stability of all these boundary cases has been examined to some extent in the literature. See \cite{GH,IL,K03}, among others. Figure \ref{SolSpace} depicts a plot of parameter space with labels for the boundary cases.

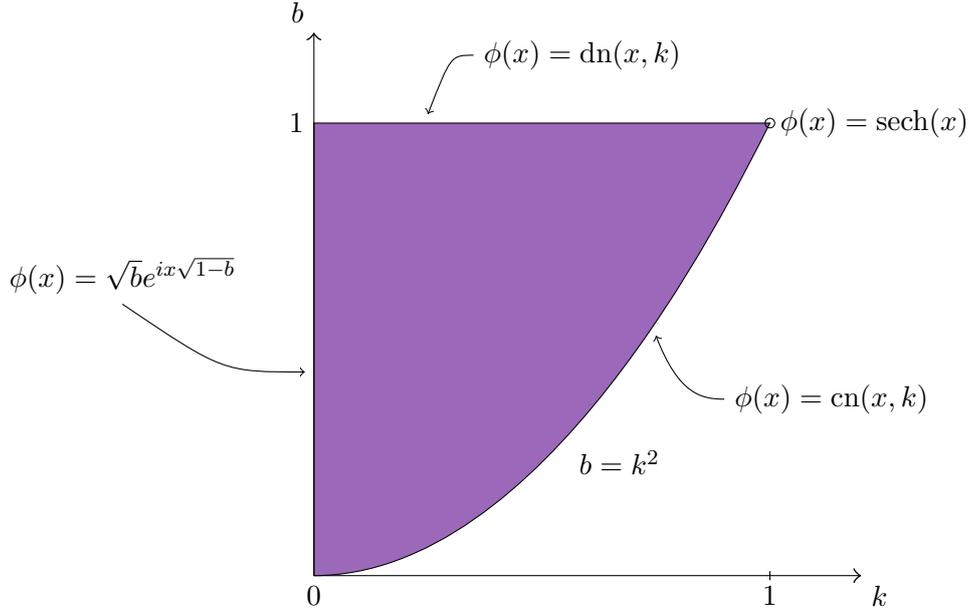
\begin{figure}
\centering
\begin{tikzpicture}[scale=6]
\filldraw[fill=BoxCol!60!white]
(0,0)  parabola (1,1)  -- (0,1) node[left]{$1$} -- (0,0);
\draw (1,1)  node[circle,scale=0.3]{} circle (0.3pt);
\draw (1,1) node[right] {$\phi(x)=\sech(x)$};
\draw[->] (0.35,1.15) node[right] {$\phi(x)=\dn(x,k)$} .. controls (0.3,1.15) .. (0.25,1.02) ;
\draw[->] (0.9,0.39) node[right] {$\phi(x)=\cn(x,k)$} .. controls (0.85,0.39) and (0.8,0.40)  .. (0.75,0.53) ;
\draw[big arrow] (0,0) node[below]  {$0$}  -- (1.2,0) node[below right] {$k$};
\draw[big arrow] (0,0) -- (0,1.2) node[above left] {$b$};
\draw  (1,0) node[below] {$1$};
\draw[->] (-0.42,0.6) node[above] {$\phi(x)=\sqrt{b}e^{ix\sqrt{1-b}}$} .. controls (-0.2,0.45) .. (-0.02,0.45) ;
\draw (.67,.3) node[below] {$b=k^2$};
\draw (1,-0.01) -- (1,0.01);
\end{tikzpicture}
\caption{Visualization of parameter space in terms of $b$ and $k$ with special solutions on the boundaries labeled.}
\label{SolSpace}
\end{figure}

We reformulate our elliptic solutions to (\ref{fNLS}) using Weierstrass elliptic functions \cite{O} rather than Jacobi elliptic functions. This will simplify working with the integral condition (\ref{intcond}) in Section \ref{lax-section}, as formulas for integrating Weierstrass elliptic functions are well documented \cite{BF,GR}. It is important to note that nothing is lost by switching to Weierstrass elliptic functions, as we can map any Weierstrass elliptic function to a Jacobi elliptic function, and vice versa \cite{O}. Let
\beq \wp(z+\omega_3,g_2,g_3)-e_3 =\left(\frac{K(k)k}{\omega_1}\right)^2 \sn^2\left(\frac{K(k)z}{\omega_1},k\right), \eeq
with $g_2$ and $g_3$ the lattice invariants of the Weierstrass $\wp$ function, $e_1$, $e_2$, and $e_3$ the zeros of the polynomial $4t^3-g_2 t-g_3 ,$ and $\omega_1$ and $\omega_3$ the half-periods of the Weierstrass lattice given by
\beq \omega_1 = \int _{e_1} ^\infty \frac{\textrm{d}z}{\sqrt{4z^3-g_2 z-g_3}}, \eeq
\beq \omega_3 = i \int _{-e_3} ^\infty \frac{\textrm{d}z}{\sqrt{4z^3-g_2 z+g_3}}.\eeq
We look for stationary solutions to (\ref{fNLS}) of the form (\ref{statsolns}). We split $\phi$ into its amplitude and phase, letting
\beq \phi_w(x)=R_w(x)e^{i \theta_w(x)}, \eeq
where $R_w$ and $\theta_w$ are expressed in terms of Weierstrass elliptic functions. Substituting this ansatz into (\ref{stationaryeqn}) gives solutions in Weierstrass form \cite{CM}:
\beq \theta_w(x)=\pm \frac{i}{2} \left(\log\left(\frac{\sigma_w(x+x_w+a_w,g_2,g_3)}{\sigma_w(x+x_w-a_w,g_2,g_3)}\right)+2 (x+x_w) \zeta_w(a)\right), \eeq
\beq \label{R_weqn} R_w^2(x)=e_0-\wp(x+x_w,g_2,g_3), \eeq
\beq \label{g2eqn} g_2=12 e_0^2+K_2, \eeq
\beq \label{g3eqn} g_3= 4 K_1-8 e_0^3-e_0 K_2, \eeq
\beq \label{e0def} e_0 = -\frac{2 \omega}{3} = \wp(a_w,g_2,g_3). \eeq
Here $\sigma_w$ and $\zeta_w$ are the Weierstrass $\sigma$ and Weierstrass $\zeta$ functions respectively \cite{O}, and $\omega$, $K_1$, $K_2$, and $x_w$ are free parameters. The constant $a_w$ from (\ref{e0def}) is given explicitly as
\beq a_w=\wp^{-1}(e_0,g_2,g_3). \eeq
We can recover the Jacobi elliptic solutions from these solutions by fixing the free parameters
\beq \label{omegaJacobi} \omega = \frac{1}{2} \left(1+k^2-3b\right), \eeq
\beq K_1^2 = c^2 = b(1-b)(b-k^2), \eeq
\beq K_2 = -4 \left(k^2-2bk^2 +3 b^2 -2b\right), \eeq
\beq \label{xwJacobi} x_w = i K'(k), \eeq
where $K'(k)$ is the complement to $K(k)$ given by $K'(k) =K\left(1-k^2\right)$.
Under this mapping we have
\beq g_2 = \frac{4}{3}\left(1-k^2+k^4\right), \eeq
\beq g_3 = \frac{4}{27}\left(2-3k^2-3k^4+2k^6\right), \eeq
\beq e_1 = \frac{1}{3}\left(2-k^2\right),\;\; e_2 = \frac{1}{3}\left(-1+2k^2\right),\;\; e_3 = \frac{1}{3}(-1-k^2), \eeq
\beq \label{omega-nums} \omega_1 = K(k),\;\; \omega_3 = i K'(k).\eeq
The homogeneity property of the Weierstrass $\wp$ function \cite{O},
\beq \wp(x,g_2,g_3) = g_2^{\frac{1}{2}} \wp\left(g_2^{\frac{1}{4}}x,1,g_3 g_2^{-\frac{3}{2}}\right). \eeq
allows us to rewrite (\ref{R_weqn}) as
\beq \label{Rweqnnew} R_w^2(x)=-\left(g_2^{\frac{1}{2}}\wp\left(g_2^{\frac{1}{4}}(x+x_w),1,g_3 g_2^{-\frac{3}{2}}\right)-e_0\right), \eeq
which has only one varying lattice invariant $g_3 g_2^{-\frac{3}{2}}$. This comes at the cost of rescaling $x$ and the magnitude of the Weierstrass $\wp$ function. The formulation (\ref{Rweqnnew}) allows for a display of parameter space as in Figure \ref{SolSpace}, but using the Weierstrass parameters, see Figure \ref{SolSpaceW}. In this figure, we see where the $\cn$, $\dn$, and Stokes wave solutions map to in the Weierstrass domain.

\begin{figure}
\centering
\includegraphics[width=12cm]{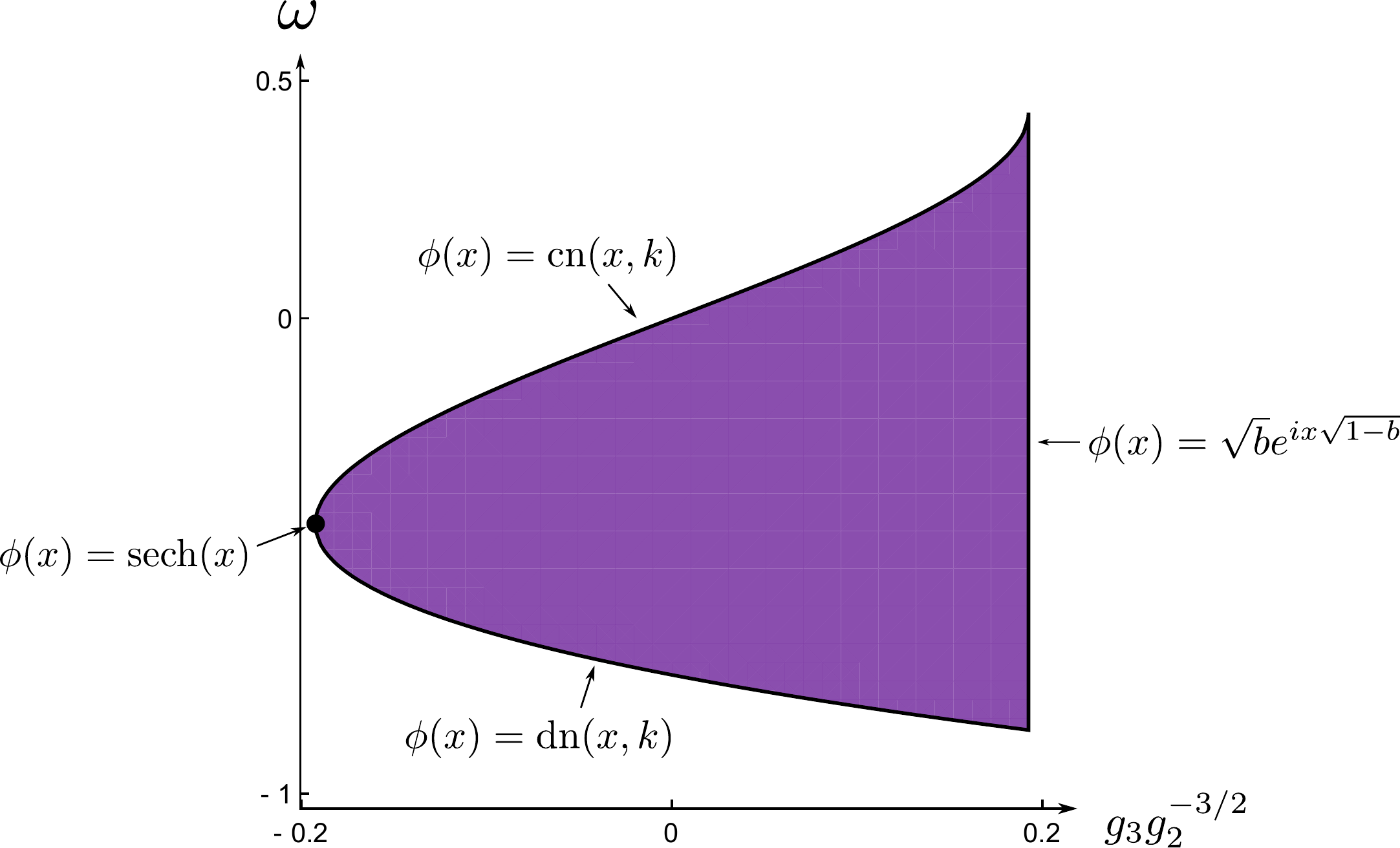}
\caption{The parameter space for elliptic solutions in Weierstrass form (\ref{Rweqnnew}). The cn, dn and Stokes wave solutions are found on the boundaries of this space, with the soliton solution occurring at the limiting point where the dn and cn curves meet.}
\label{SolSpaceW}
\end{figure}

\section{The linear stability problem}\label{setup}
To examine the linear stability of our solutions we consider
\beq \label{linpert} \Psi(x,t)= e^{-i \omega t}e^{i \theta(x)}\left(R(x)+\epsilon u(x,t)+\epsilon i v(x,t)+\mathcal{O}(\epsilon^2)\right), \eeq
where $\epsilon$ is a small parameter and $u(x,t)$ and $v(x,t)$ are the real and imaginary parts of our perturbation, which depends on both $x$ and $t$. Substituting (\ref{linpert}) into (\ref{fNLS}), isolating $O(\epsilon)$ terms, and splitting into real and imaginary parts, we obtain a system of equations
\beq \label{linprob}
\frac{\partial}{\partial t}\left(\begin{array}{c} u \\ v \end{array} \right)=\mathcal{L} \left(\begin{array}{c} u \\ v \end{array} \right) = \left(\begin{array}{cc} -S & \mathcal{L}_- \\ -\mathcal{L}_+ & -S \end{array} \right)\left(\begin{array}{c} u \\ v \end{array} \right) = J \left(\begin{array}{cc} \mathcal{L}_+ & S \\ -S & \mathcal{L}_- \end{array} \right)\left(\begin{array}{c} u \\ v \end{array} \right) ,
\eeq
with
\beq J = \left( \ba{cc} 0 & 1 \\ -1 & 0 \ea \right). \eeq
The linear operators $\mathcal{L}_+$, $\mathcal{L}_-$, and $S$ are given by
\begin{align}
\mathcal{L}_- & = -\frac{1}{2}\partial_x^2-R^2(x)-\omega+\frac{c^2}{2R^4(x)}, \\ \mathcal{L}_+ & = -\frac{1}{2}\partial_x^2-3R^2(x)-\omega+\frac{c^2}{2R^4(x)}, \\ S&=\frac{c}{R^2(x)}\partial_x-\frac{c R'(x)}{R^3(x)} = \frac{c}{R(x)}\partial_x \frac{1}{R(x)}.
\end{align}
An elliptic solution $\phi(x)=R(x)e^{i \theta(x)}$ is by definition linearly stable if for all $\epsilon >0$ there exists a $\delta>0$ such that if $||u(x,0)+i v(x,0)||<\delta$ then $||u(x,t)+i v(x,t)||<\epsilon$ for all $t>0$. This definition depends on the choice of norm $||\cdot||$ which is specified in the definition of the spectrum in (\ref{spectral-norm}) below.

Since (\ref{linprob}) is autonomous in $t$, we separate variables to look at solutions of the form
\beq \label{uvseparation}\left( \ba{c} u(x,t) \\ v(x,t) \ea \right) = e^{\lambda t} \left( \ba{c} U(x) \\ V(x) \ea \right), \eeq
resulting in the spectral problem
\beq \label{spectralprob}
\lambda \left(\begin{array}{c} U \\ V \end{array} \right)= \mathcal{L} \left(\begin{array}{c} U \\ V \end{array} \right)=\left(\begin{array}{cc} -S & \mathcal{L}_- \\ -\mathcal{L}_+ & -S \end{array} \right)\left(\begin{array}{c} U \\ V \end{array} \right)= J \left(\begin{array}{cc} \mathcal{L}_+ & S \\ -S & \mathcal{L}_- \end{array} \right)\left(\begin{array}{c} U \\ V \end{array} \right) .
\eeq
Here
\beq \label{spectral-norm} \sigma_{\mathcal{L}}=\{\lambda\in\mathbb{C}: \max_{x\in \mathbb{R}} \left( |U(x)|,|V(x)| \right) < \infty\}, \eeq
or
\beq U,V\in C_b^0(\mathbb{R}). \eeq
In order to have spectral stability, we need to demonstrate that the spectrum $\sigma_{\mathcal{L}}$ does not enter into the right half of the complex $\lambda$ plane.
Since (\ref{fNLS}) is Hamiltonian \cite{AS}, the spectrum of its linearization is symmetric with respect to both the real and imaginary axis \cite{W}. In other words, proving spectral stability for elliptic solutions to (\ref{fNLS}) amounts to proving that the stability spectrum lies strictly on the imaginary axis. In our case, we show that none of the elliptic solutions are spectrally stable, as we demonstrate spectral elements in the right-half plane near the origin for any choice of the parameters $b$ and $k$.

\section{The Lax pair}\label{lax-section} We wish to obtain an analytical representation for the spectrum $\sigma_{\mathcal{L}}$. As mentioned in the introduction, this analytical representation comes from the squared eigenfunction connection between the linear stability problem (\ref{linprob}) and its Lax pair. We begin by formulating (\ref{fNLS}) in a traveling frame, by defining
\beq \label{Psieqn} \Psi(x,t) = e^{-i \omega t} \psi(x,t), \eeq
so that $\psi$ satisfies
\beq \label{statsolneqn} i \psi_t=-\omega \psi -\frac{1}{2}\psi_{xx}-\psi |\psi|^2. \eeq
This equation is equivalent to the compatibility condition $\chi_{xt}=\chi_{tx}$ of the following Lax pair \cite{ZS}:
\beq \label{chix} \chi_x= \left( \ba{cc} -i \zeta & \psi \\ -\psi^* & i \zeta \ea \right)\chi, \eeq
\beq \label{chit} \chi_t = \left( \ba{cc}
-i \zeta^2+\frac{i}{2} |\psi|^2+\frac{i}{2} \omega & \zeta \psi+\frac{i}{2}\psi_x \\
-\zeta \psi^*+\frac{i}{2} \psi^*_x & i \zeta^2-\frac{1}{2} |\psi|^2-\frac{i}{2}\omega
\ea \right) \chi, \eeq
where $^*$ represents the complex conjugate \cite{AKNS,BDN}. Regarding (\ref{chix}) as a spectral problem with $\zeta$ as the spectral parameter:
\beq \left( \ba{cc}
i \partial_x & -i \psi \\
-i \psi^* & -i \partial_x \ea \right) \chi = \zeta \chi,\eeq
we see that it is not self adjoint \cite{K12}. This means that the spectral parameter $\zeta$ is not necessarily confined to the real axis as it was for defocusing NLS \cite{BDN} which makes our analysis more difficult. Since the elliptic solutions are given by $\psi(x,t)=\phi(x)$, we restrict the Lax pair to elliptic solutions by writing
\beq \label{chixphi} \chi_x= \left( \ba{cc} -i \zeta & \phi \\ -\phi^* & i \zeta \ea \right)\chi, \eeq
\beq \label{chitphi} \chi_t = \left( \ba{cc}
-i \zeta^2+\frac{i}{2} |\phi|^2+\frac{i}{2} \omega & \zeta \phi+\frac{1}{2}\phi_x \\
-\zeta \phi^*+\frac{i}{2} \phi^*_x & i \zeta^2-\frac{1}{2} |\phi|^2-\frac{i}{2}\omega
\ea \right)\chi. \eeq
Henceforth we refer to the spectrum of (\ref{chixphi}) as $\sigma_L$ or informally as the Lax spectrum. Specifically, $\sigma_L$ consists of all $\zeta$ for which (\ref{chixphi}) has a bounded (in $x$) eigenfunction solution.
To determine $\sigma_L$ we start by rewriting (\ref{chitphi}) in the short-hand form
\beq \label{chitphisimp}
\chi_t = \left( \ba{cc}
A & B \\ C & -A
\ea \right) \chi, \eeq
where
\begin{align}
A &=-i \zeta^2+\frac{i}{2} |\phi|^2+\frac{i}{2} \omega, \label{Aeqn}\\
B &= \zeta \phi+\frac{1}{2}\phi_x,\label{Beqn}\\
C &= -\zeta \phi^*+\frac{i}{2} \phi^*_x.
\end{align}
Since $A$, $B$, and $C$ are independent of $t$, we separate variables. Let
\beq \label{Omega} \chi(x,t)=e^{\Omega t} \varphi(x), \eeq
with $\Omega$ being independent of $t$ but possibly depending on $x$.
Substituting (\ref{Omega}) into (\ref{chitphisimp}) and canceling the exponential, we find
\beq \label{Omegamat} \left( \ba{cc} A-\Omega & B \\ C & -A-\Omega \ea \right) \varphi =0. \eeq
In order to have nontrivial solutions we require the determinant of (\ref{Omegamat}) to be zero. Using the definitions of $A,B$ and $C$, we get
\beq \label{Omegacond} \Omega^2 =A^2+B C = -\zeta^4+\omega \zeta^2+c \zeta+\frac{1}{16}\left(-4 \omega b -3b^2-{k'}^4\right), \eeq
where $k' = \sqrt{1-k^2}$. We notice that $\Omega$ is not only independent of $t$ but also of $x$. Thus $\Omega$ is strictly a function of $\zeta$ and the solution parameters.

To satisfy (\ref{Omegamat}), we let
\beq \label{varphidef} \varphi(x) = \gamma(x)\left( \ba{c} -B(x) \\ A(x)-\Omega \ea \right), \eeq
where $\gamma(x)$ is a scalar function. By construction of $\varphi(x)$, $\chi(x,t)$ satisfies (\ref{chitphi}). Since (\ref{chixphi}) and (\ref{chitphi}) commute, it is possible to choose $\gamma(x)$ such that $\chi$ also satisfies (\ref{chixphi}). Indeed, $\gamma(x)$ satisfies a first-order linear equation, whose solution is given by
\beq \label{gammacon} \gamma(x)=\gamma_0 \exp \left(-\int \frac{(A-\Omega) \phi+B_x+i \zeta B}{B}\textrm{d}x \right). \eeq
For almost every $\zeta\in \mathbb{C}$, we have explicitly determined the two linearly independent solutions of (\ref{chixphi}), \textit{i.e.}, those corresponding to the positive and negative signs of $\Omega$ in (\ref{Omegacond}). Assuming $\Omega \ne 0$ these two solutions are by construction linearly independent. In the case where $\zeta$ corresponds to $\Omega =0$ the second solution to (\ref{chixphi}) can be determined via the reduction-of-order method.

Since (\ref{chixphi}) and (\ref{chitphi}) share their eigenfunctions, $\sigma_L$ is the set of all $\zeta \in \mathbb{C}$ such that (\ref{varphidef}) is bounded for all $x\in\mathbb{R}$. Indeed, the vector part of $\varphi$ is bounded for all $x$, so we only need that the scalar function $\gamma(x)$ is bounded as $x\rightarrow\pm\infty$. A necessary and sufficient condition for this is
\beq \label{intcond} \left\langle \textrm{Re} \left(\frac{(A-\Omega) \phi+B_x+i \zeta B}{B}\right)\right\rangle = 0,\eeq
where $\langle \cdot \rangle$ is the average over one period $2 K(k)$ of the integrand, and $\textrm{Re}$ denotes the real part. At this point, the integral condition (\ref{intcond}) completely determines the Lax spectrum $\sigma_L$.

\section{The squared eigenfunction connection}\label{sqeig-section}
A connection between the eigenfunctions of the Lax pair (\ref{chixphi}) and (\ref{chitphi}) and the eigenfunctions of the linear stability problem (\ref{linprob}) using a squared eigenfunctions is well known \cite{AKNS}.
We prove the following theorem.
\vspace{2mm}
\begin{thm}
The vector
\beq \label{connection} \left( \ba{c} u \\ v \ea \right)=\left( \ba{c} e^{-i \theta(x)} \chi_1^2-e^{i \theta(x)}\chi_2^2 \\ -i e^{-i \theta(x)} \chi_1^2-i e^{i \theta(x)} \chi_2^2 \ea \right) \eeq
 satisfies the linear stability problem (\ref{linprob}). Here $\chi = \left(  \chi_1, \chi_2 \right)^T$ is any solution of the Lax pair (\ref{chix}-\ref{chit}) corresponding by direct calculation to the elliptic solution $\phi(x)=R(x)e^{i \theta(x)}.$
\end{thm}
\begin{proof}
The proof is done by direct calculation. For the left-hand side of (\ref{linprob}), evaluate $(u_t,v_t)$ using the product rule and (\ref{chit}). Eliminate $x$-derivatives of $u$ and $v$ (up to order 2) using (\ref{chix}). Upon substitution and using (\ref{chixphi}) and (\ref{chitphi}), the left-hand side and right-hand side of (\ref{linprob}) are equal, finishing the proof.
\end{proof}

To establish the connection between the $\sigma_{\mathcal{L}}$ spectrum and the $\sigma_L$ spectrum we examine the right- and left-hand sides of (\ref{uvseparation}).
Substituting in (\ref{connection}) and (\ref{Omega}) to the left-hand side of (\ref{uvseparation}) we find
\beq \label{lambdaOmegacon} e^{2 \Omega t} \left( \ba{c} e^{-i \theta(x)} \varphi_1^2-e^{i \theta(x)}\varphi_2^2 \\ -i e^{-i \theta(x)} \varphi_1^2-i e^{i\theta(x)} \varphi_2^2 \ea \right) = e^{\lambda t} \left( \ba{c} U \\ V \ea \right), \eeq
and we conclude that
\beq \label{lambdacond} \lambda = 2 \Omega(\zeta), \eeq
with eigenfunctions  given by
\beq \label{eigenfunctions} \left( \ba{c} U \\ V \ea \right) = \left( \ba{c} e^{-i \theta(x)} \varphi_1^2-e^{i \theta(x)}\varphi_2^2 \\ -i e^{-i \theta(x)} \varphi_1^2-i e^{i\theta(x)} \varphi_2^2 \ea \right).
\eeq
This gives the connection between the $\sigma_L$ spectrum and the $\sigma_{\mathcal{L}}$ spectrum. It is also necessary to check that indeed all solutions of (\ref{spectralprob}) are obtained through (\ref{lambdaOmegacon}). This is not shown explicitly here, but is done analogous to the work in \cite{BDN}.

Although in principle the above construction determines $\sigma_{\mathcal{L}}$, it remains to be seen how practical this determination is.
In the following section we discuss a technique for explicitly integrating (\ref{intcond}) using Weierstrass elliptic functions, leading to a more explicit characterization of $\sigma_{\mathcal{L}}$.

\section{The Lax spectrum in terms of elliptic functions} \label{Weierstrass}
In terms of Weierstrass elliptic functions,
\beq \label{OmegacondW} \Omega^2 = -\zeta^4 +\omega \zeta^2 +K_1 \zeta+\frac{1}{4}\left(-\omega^2-\frac{K_2}{4}\right), \eeq
while (\ref{intcond}) becomes
\beq \label{intcondW} \textrm{Re} \int_0^{2\omega_1} \frac{(A-\Omega) \phi_w+B_x+i \zeta B}{B} \textrm{d}x=0,\eeq
with $A$ and $B$ given in (\ref{Aeqn}) and (\ref{Beqn}).
Substituting for $\phi_w$ we find that (\ref{intcondW}) is of the form
\beq \label{intcondCs} \textrm{Re} \int_0^{2\omega_1} \frac{C_1+C_2 \wp(x) +C_3 \wp'(x)}{C_4+C_5 \wp(x)} \textrm{d}x=0, \eeq
here $\wp(x)=\wp(x+x_w,g_2,g_3)$ with the dependence on $x_w,$ $g_2$, and $g_3$ suppressed. The $C_j$'s depend on $\zeta$ but are independent of $x$. They are given by
\begin{align}\begin{split} C_1 &= -\frac{2\omega \zeta}{3}-\frac{K_1}{2}+\zeta^3-\frac{\omega \zeta}{6}-i \zeta \Omega(\zeta), \\
C_2 &= -\frac{\zeta}{2}, \hspace{3cm} C_3 = \frac{i}{4},\\
C_4 &= -\Omega(\zeta)-i \zeta^2 +i \frac{\omega}{6}, \hspace{6mm} C_5 = -\frac{i}{2}.  \end{split}\end{align}
We can evaluate the integral in (\ref{intcondCs}) explicitly \cite{GR}. We find
\beq\textrm{Re} \left[ \frac{2 \omega_1 C_2}{C_5}+\frac{4(C_1 C_5-C_2 C_4)}{\wp'(\alpha) C_5^2}\left(\zeta_w(\alpha)\omega_1-\zeta_w(\omega_1)\alpha\right)\right] =0,\eeq
with
\beq \label{intcondWfull} \alpha=\alpha(\zeta)=\wp^{-1}\left(-\frac{C_4(\zeta)}{C_5(\zeta)},g_2,g_3\right). \eeq
Here $\wp^{-1}$ is a multivalued function, but for the sake of our analysis $\alpha$ is chosen as any value such that $\wp(\alpha) = -C_4(\zeta)/C_5(\zeta)$.
Substituting for the $C_j$'s, (\ref{intcondCs}) becomes
\beq\label{intcond2} \textrm{Re} \left[-2 i \zeta \omega_1+\frac{4i\left(-K_1 +4\zeta^3-2\zeta\omega-4i\zeta\Omega(\zeta)\right)}{\wp'(\alpha)}\left(\zeta_w(\alpha)\omega_1-\zeta_w(\omega_1)\alpha\right)\right] =0.\eeq
We simplify this further by recognizing that
\beq \label{wpeqn} {\wp'}^2(\alpha)=4\wp^3(\alpha)-g_2 \wp(\alpha)-g_3=4{\left(-\frac{C_4(\zeta)}{C_5(\zeta)}\right)}^3-g_2\left(-\frac{C_4(\zeta)}{C_5(\zeta)}\right)-g_3.\eeq
Substituting for $C_4(\zeta)$ and $C_5(\zeta)$, changing $g_2$ and $g_3$ to $K_1$ and $K_2$ via (\ref{g2eqn}) and (\ref{g3eqn}) respectively, and substituting in (\ref{OmegacondW}) for higher powers of $\Omega(\zeta)$ gives
\beq \label{wpprime} {\wp'}^2(\alpha)=-4\left(-K_1 +4\zeta^3-2\zeta\omega-4i\zeta\Omega(\zeta)\right)^2.\eeq
Thus (\ref{intcond2}) simplifies to
\beq\label{intcond3} \textrm{Re} \left(-2 i \zeta \omega_1+ 2\tau\left(\zeta_w(\alpha)\omega_1-\zeta_w(\omega_1)\alpha\right)\right) =0,\eeq
where $\tau = \textrm{sgn}(\textrm{Re}\left(-K_1 +4\zeta^3-2\zeta\omega-4i\zeta\Omega(\zeta)\right).$

Under the mapping (\ref{omega-nums}), and applying the formula for the Weierstrass $\zeta$ function evaluated at a half period \cite{BF}, $ \zeta_w(\omega_1) = \sqrt{e_1-e_3}\left(E(k)-\frac{e_1}{e_1-e_3}K(k)\right),$ (\ref{intcond3}) becomes
\beq \label{intcond4} \textrm{Re} \left[ -2 i \zeta K(k) +2\tau \left(\zeta_w(\alpha)K(k)-\left(E(k)-\frac{1}{3}\left(2-k^2\right) K(k)\right)\alpha \right)\right] = 0.\eeq
Here
\beq E(k) = \int_0^{\pi/2} \sqrt{1-k^2 \sin^2 y} dy, \eeq
is the complete elliptic integral of the second kind.
At this point, we have simplified the integral condition (\ref{intcondW}) as much as possible. Thus $\zeta\in\sigma_L$ if and only if (\ref{intcond4}) is satisfied. To simply notation, let
\beq I(\zeta)= -2 i \zeta \omega_1+ 2\tau\left(\zeta_w(\alpha)\omega_1-\zeta_w(\omega_1)\alpha\right), \eeq
so that (\ref{intcond4}) is
\beq \label{intcondI} \textrm{Re} \left[ I(\zeta)\right] = 0.\eeq

Next, we wish to examine the level sets of the left-hand side of (\ref{intcondI}). To this end, we differentiate $I(\zeta)$ with respect to $\zeta$. To evaluate this derivative we use the chain rule and note that
\beq \frac{\partial}{\partial \zeta} \zeta_w(\alpha) = -\wp(\alpha)\frac{\partial \alpha}{\partial \zeta} = \frac{C_4(\zeta)}{C_5(\zeta)}\frac{\textrm{d} \wp^{-1}}{\textrm{d} \zeta} \left(-\frac{C_4(\zeta)}{C_5(\zeta)},g_2,g_3\right) \left(-\frac{C_4(\zeta)}{C_5(\zeta)}\right)'.\eeq
Since
\beq \frac{\textrm{d}}{\textrm{d}z} \wp^{-1}\left(-\frac{C_4(\zeta)}{C_5(\zeta)},g_2,g_3\right)=\frac{1}{\wp'\left (\wp^{-1}\left(-\frac{C_4(\zeta)}{C_5(\zeta)},g_2,g_3\right)\right)}=\frac{1}{\wp'(\alpha)}, \eeq
we can use (\ref{wpprime}) to obtain
\beq \label{derivintcond} \frac{\textrm{d} I(\zeta)}{\textrm{d}\zeta} =\frac{2 E(k)-\left(1+b-k^2+4 \zeta^2\right)K(k)}{2 \Omega(\zeta)} . \eeq
Simply taking the real part of (\ref{derivintcond}) does not give another characterization of the spectrum.  Instead, if we think of (\ref{intcond4}) as restricting ourselves to the zero level set of the left-hand side. Then we use (\ref{derivintcond}) to determine a tangent vector field which allows us to map out level curves originating from any point. This is explained in more detail in Section \ref{regions}. Additionally, there we see that (\ref{derivintcond}) is useful in determining the boundary regions in parameter space corresponding to qualitatively different parts of the spectrum.

\section{The $\sigma_{\mathcal{L}}$ spectrum on the imaginary axis} \label{imagaxis}
In this section we discuss $\sigma_{\mathcal{L}}\cap i \mathbb{R}$. As we demonstrate, this corresponds to the part of $\sigma_L$ lying on the real axis. Using (\ref{intcond4}) we obtain analytic expressions for $\sigma_L\cap \mathbb{R},$ and thus for $\sigma_{\mathcal{L}}\cap i \mathbb{R}$.

First, we consider $\zeta\in\mathbb{R}$. As we demonstrate below, (\ref{intcond4}) is satisfied for any real $\zeta$. Using (\ref{OmegacondW}) and (\ref{lambdacond}), we determine the corresponding parts of $\sigma_{\mathcal{L}}$.
\vspace{2mm}
\begin{thm}
The condition (\ref{intcond4}) is satisfied for all $\zeta\in\mathbb{R}$.
\end{thm}
\begin{proof}
Since $k,$ $K(k),$ and $E(k)$ are real, it suffices to show that $\alpha\in i \mathbb{R}$ and $\zeta_w(\alpha) \in i\mathbb{R}$. Since $\zeta_w$ with $g_2,g_3\in\mathbb{R}$ takes real values to real values and purely imaginary values to purely imaginary values \cite{O}, it suffices to show that $\alpha=\wp^{-1}\left(-\frac{C_4(\zeta)}{C_5(\zeta)},g_2,g_3\right)\in i \mathbb{R}$.
For $g_2,g_3\in \mathbb{R}$, $\wp(\mathbb{R},g_2,g_3)$ maps to $[e_1,\infty),$ and since $\wp(i x,g_2,g_3)=-\wp(x,g_2,-g_3)$ we have that $\wp(i \mathbb{R},g_2,g_3)$ maps to $(-\infty,e_3].$ Thus we need to show that $\zeta,$ $-\frac{C_4(\zeta)}{C_5(\zeta)}\le e_3.$
Substituting for $C_4(\zeta)$ and $C_5(\zeta),$ we want to show
\beq \frac{1}{6}\left(2\omega-12\zeta^2-3\sqrt{1-3b^2+k^4-16c\zeta-8\zeta^2+16\zeta^4-2k^2\left(1+4\zeta^2\right)+2b\left(1+k^2+12\zeta^2\right)}\right)\le e_3. \eeq
Simplifying the left- and right- hand sides of this expression yields
\beq \label{wts2} 4\zeta^2+\sqrt{\left(1+k^2-b-4\zeta^2\right)^2+4\left(2\sqrt{b}\zeta-\sqrt{(1-b)(b-k^2)}\right)^2}\ge 1+k^2-b. \eeq
There are two cases. If $4\zeta^2\ge 1+k^2-b$ we are done, as the square root term is nonnegative. If $4\zeta^2< 1+k^2-b$, we have
\beq 4\zeta^2+\sqrt{\left(1+k^2-b-4\zeta^2\right)^2+4\left(2\zeta\sqrt{b}-\sqrt{(1-b)(b-k^2)}\right)^2}\ge4\zeta^2+\sqrt{\left(1+k^2-b-4\zeta^2\right)^2}. \eeq
Since $1+k^2-b-4\zeta^2>0,$ this gives (\ref{wts2}) as we wished to prove.
\end{proof}

At this point, we know that $\mathbb{R}\subset \sigma_L.$ We wish to see what this corresponds to for $\sigma_{\mathcal{L}}$. Looking at (\ref{OmegacondW}), we notice that
\beq \label{Omegaw}\Omega^2 = -\frac{1}{16}\left(\left(1+k^2-b-4\zeta^2\right)^2+4\left(2\zeta\sqrt{b}-\sqrt{(1-b)(b-k^2)}\right)^2\right). \eeq
For convenience define
\beq S_\Omega = \left\{ \Omega : \Omega^2 = -\frac{1}{16}\left(\left(1+k^2-b-4\zeta^2\right)^2+4\left(2\zeta\sqrt{b}-\sqrt{(1-b)(b-k^2)}\right)^2\right) \textrm{ and } \zeta\in\sigma_L \right\}. \eeq
Thus when $\zeta\in\mathbb{R},$ $\Omega(\zeta)\in i \mathbb{R}$ necessarily, since $\Omega^2(\zeta)<0$.
Applying (\ref{lambdacond}), we see that $\zeta\in \mathbb{R}$ corresponds to imaginary spectral elements of $\sigma_{\mathcal{L}}$.
Representative plots of $\Omega^2$ are shown in Figure \ref{Omegaplots}.
The subset of $S_\Omega$ corresponding to $\zeta\in\mathbb{R}$ consists of $\left(-\infty,-i |\Omega_m|\right] \cup \left[i |\Omega_m|,\infty\right),$ where $\Omega_m$ is the maximum value of $\Omega$. The set $S_\Omega$ is in general at least double covered as for almost every value of $\Omega$ there are at least two values of $\zeta$ which map to it. The spectrum on the imaginary axis is quadruple covered if the quartic (\ref{OmegacondW}) has four distinct real roots $\zeta$, as is the case in Figure~\ref{Omegaplots}(d) for $\Omega^2\in(-0.0639,-0.0243)$.

\begin{figure}
\centering
\begin{tabular}{cc}
  \includegraphics[width=55mm]{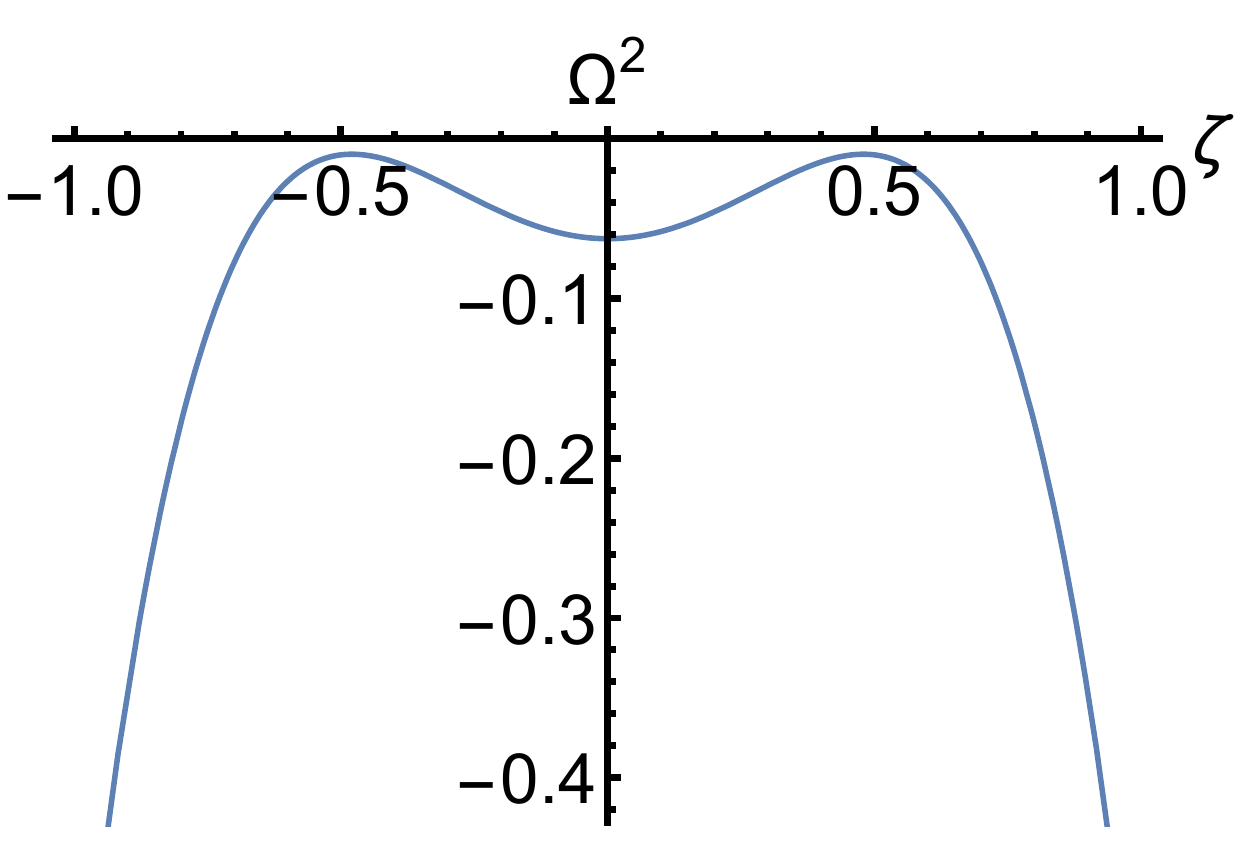} &   \includegraphics[width=55mm]{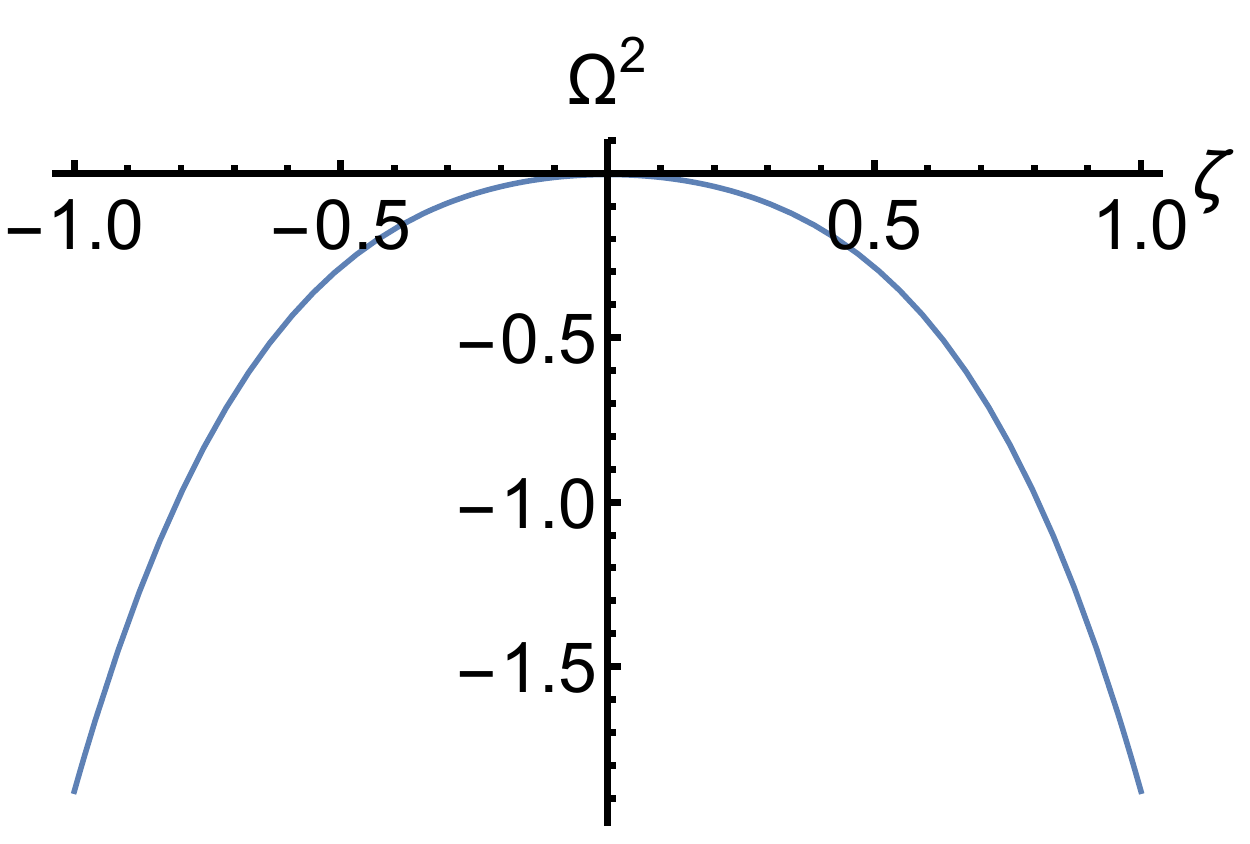} \\
(a) & (b)  \\[4pt]
 \includegraphics[width=55mm]{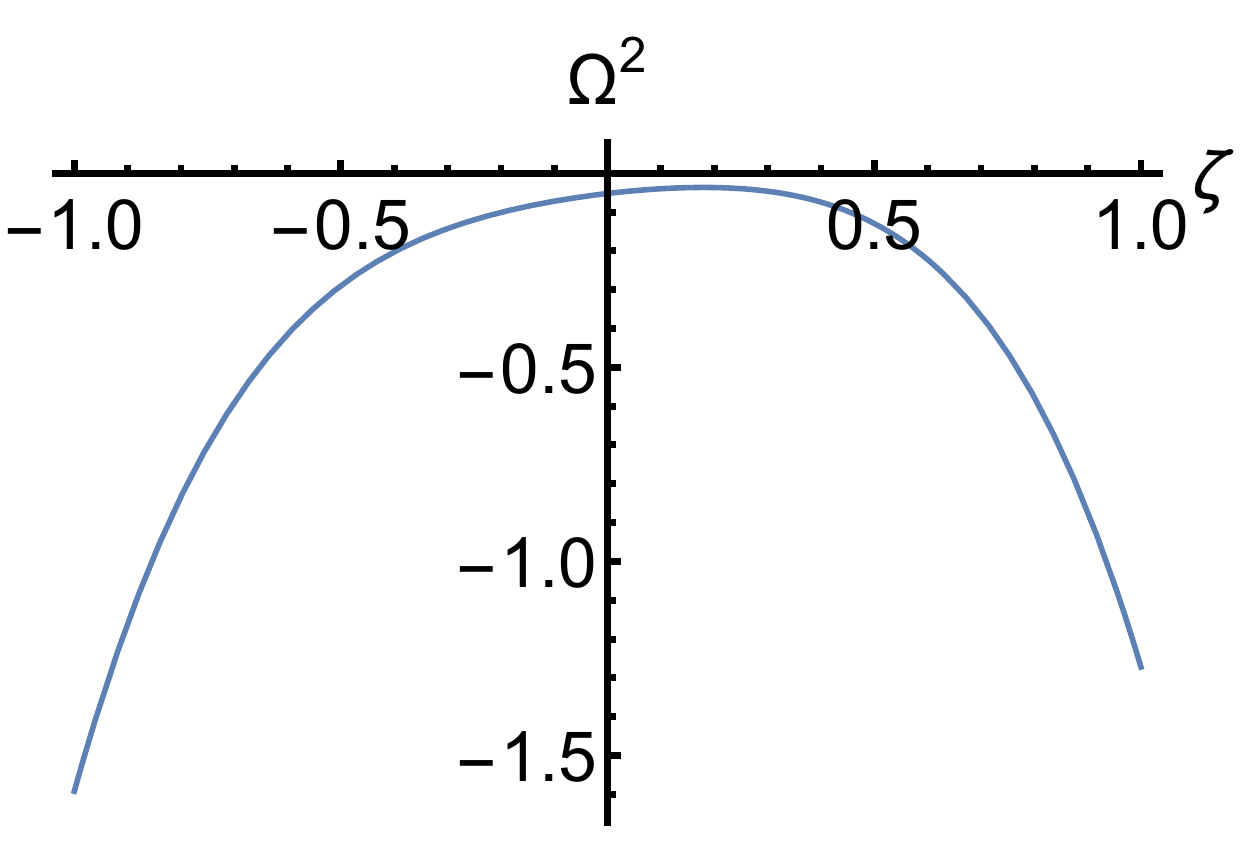} &   \includegraphics[width=55mm]{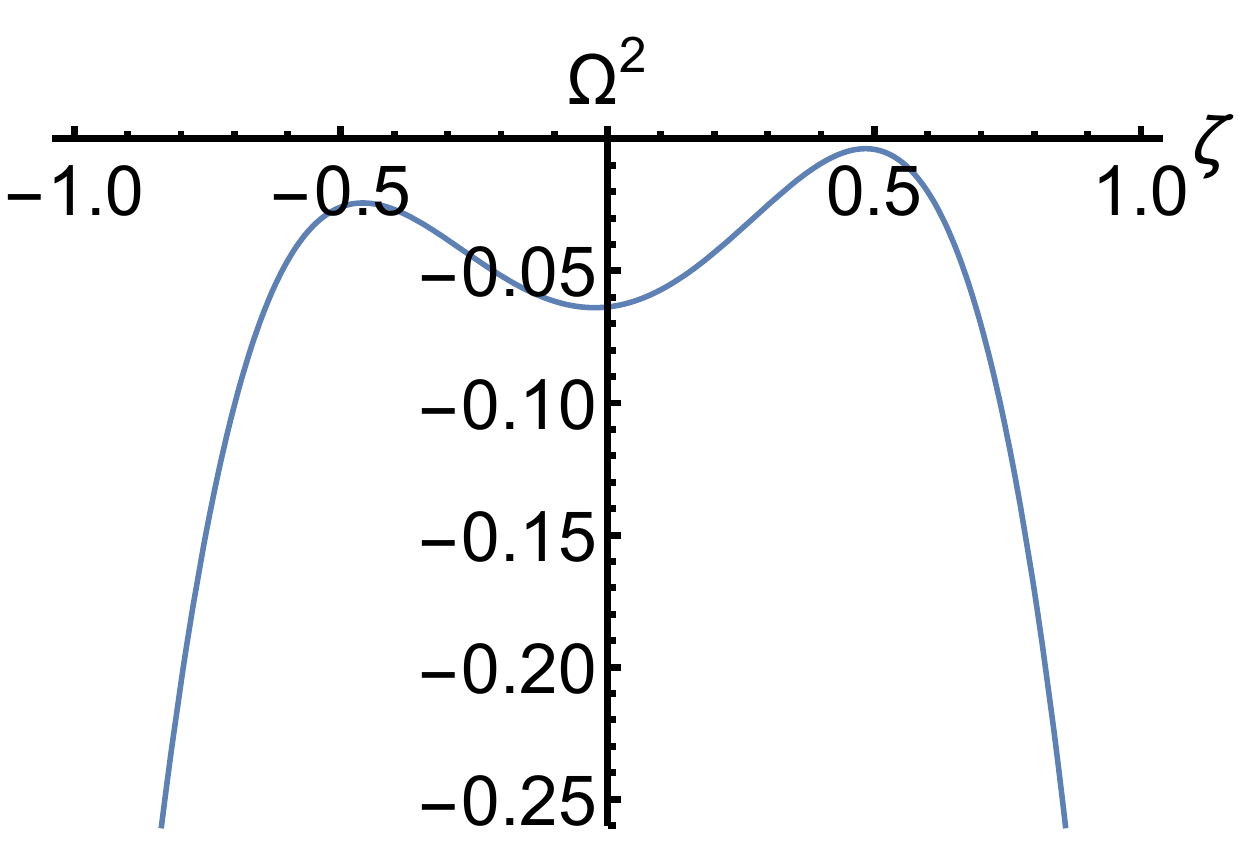} \\
(c) & (d)
\end{tabular}
\caption{$\Omega^2$ as a function of real $\zeta$ for various values of $b$ and $k$: (a) $\cn$ case with $(k,b)=(0.2,0.04);$ (b) $\dn$ case with $(k,b)=(0.5,1);$ (c) general nontrivial-phase case with one maximum with $(k,b)=(0.8,0.8);$ (d) general nontrivial-phase case with two maxima with $(k,b)=(0.2,0.05)$. }
\label{Omegaplots}
\end{figure}

The condition for a subset of the spectrum to have a quadruple covering is readily determined. We require that the quartic $\Omega^2(\zeta)$ has three critical values, \textit{i.e.}, that its derivative has three distinct roots. Examining the discriminant of (\ref{OmegacondW}) with respect to $\zeta,$ we see that if
\beq \label{quadcond} k^2<b<\frac{1+3k^2+3k^4+k^6}{9(1-k^2+k^4)}, \eeq
then there is a region of the imaginary axis which is quadruple covered. We show a plot of parameter space separated into two distinct regions by this condition in Figure \ref{imagaxisregions}. In the upper region, the subset of $\sigma_{\mathcal{L}}$ on the imaginary axis has no quadruple covering. In the lower region there is a quadruple covering.

To explicitly determine the location of the covering on the imaginary axis, we need the local extrema of $\Omega^2$. In the  case when (\ref{quadcond}) is satisfied, the three extrema $\Omega_c^2$ of $\Omega^2$ satisfy the cubic in $\Omega_c^2$
\begin{align}\label{Omegaceqn}
\begin{split}
- & 16k^4(-1+k^2)^2- 32\left(-4 k^2+32k^4-4k^6+27 b g_3\right)\Omega_c^2+ \\
& 256(-1-18b+27b^2+10k^2-18b k^2-k^4) \Omega_c^4-4096\Omega_c^6=0.
\end{split}
\end{align}
Labeling the real roots as $\Omega_{c1}^2,\,\Omega_{c2}^2,\,\Omega_{c3}^2,$ with $\Omega_{c1}^2 < \Omega_{c2}^2 < \Omega_{c3}^2,$ we have that the $\sigma_{\mathcal{L}}$ spectrum is double covered on the region $ \left(-i \infty ,-2\sqrt{\Omega^2_{c3}}\right)\cup \left(-2\sqrt{\Omega^2_{c2}},-2\sqrt{\Omega^2_{c1}}\right) \cup \left(2 \sqrt{\Omega^2_{c1}},2\sqrt{\Omega^2_{c2}}\right) \cup \left(2\sqrt{\Omega^2_{c3}},i \infty \right), $
and quadruple covered on the region
$  \left(-2\sqrt{\Omega^2_{c3}},-2\sqrt{\Omega^2_{c2}}\right) \cup \left(2 \sqrt{\Omega^2_{c2}},2\sqrt{\Omega^2_{c3}}\right). $
If (\ref{quadcond}) is not satisfied, the $\sigma_{\mathcal{L}}$ spectrum has no quadruple covering, and is double covered on the region $ \left(-i \infty ,2\sqrt{\Omega_{c^*}^2}\right) \cup \left(2 \sqrt{\Omega_{c^*}^2},i \infty \right), $
where $\Omega_{c^*}^2$ is the only real root of (\ref{Omegaceqn}).

\begin{figure}
\centering
\includegraphics[height=8cm]{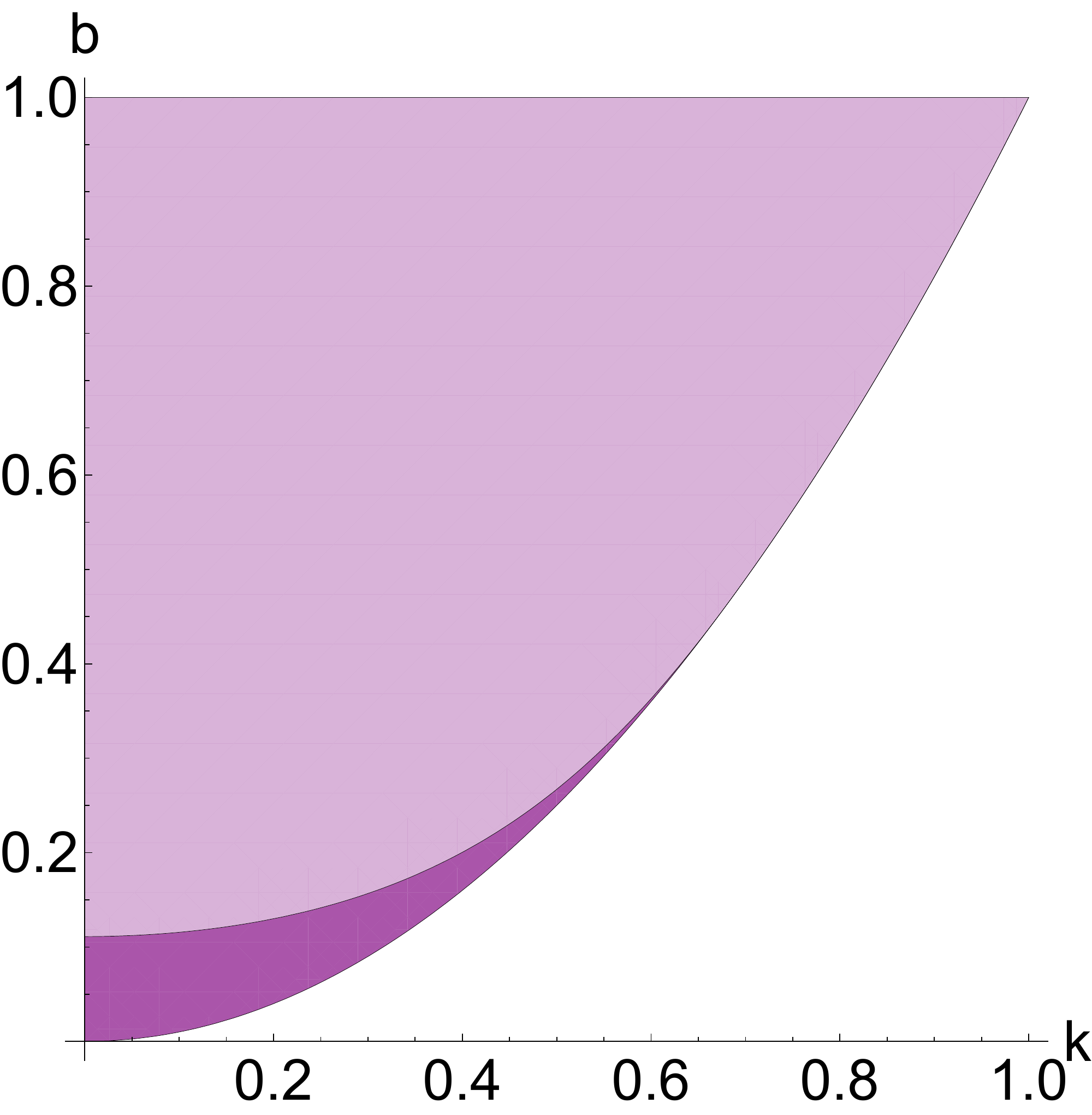}
\caption{Parameter space split using (\ref{quadcond}) in the region for which a subset of $\sigma_{\mathcal{L}}\cup i \mathbb{R}$ is quadruple covered given by (\ref{quadcond}) (dark lower region), and only double covered (light upper region). The lower region region comes to a point at $(k,b)=(\sqrt{2}/2,1/2)$.}
\label{imagaxisregions}
\end{figure}

The extent of the spectrum $\sigma_{\mathcal{L}}$ on the imaginary axis vastly simplifies for the cnoidal wave, the dnoidal wave, and the Stokes wave solutions because (\ref{OmegacondW}) is biquadratic in the former two cases, and because $k=0$ in the latter case.  We detail these boundary cases below giving the $\sigma_{\mathcal{L}}$ spectrum.
For $\dn$ solutions, the imaginary axis is double covered on the region $\left(-i \infty , -i k^2/2 \right)\cup \left( i k^2/2, i \infty \right).$
This confirms results in \cite{De7,K03}.
For $\cn$ solutions, if $k<\sqrt{2}/2$, the imaginary axis is double covered from
$ \left(-i\infty, -i/2 \right) \cup \left(i/2,i \infty \right), $
and quadruple covered from
$ \left(-i/2,- i k\sqrt{1-k^2} \right) \cup \left(i k \sqrt{1-k^2},i/2\right). $
Finally, for the Stokes wave solutions, if $b>1/9$, then $i \mathbb{R} \subset \sigma_{\mathcal{L}}$ and is double covered. If $b<1/9$, then the imaginary axis is still fully double covered except from
$\left(-S_{+},-S_{-}\right) \cup \left(S_{-},S_{+}\right), $
where it is quadruple covered, here
\beq S_{\pm} = \frac{\sqrt{-1\mp \sqrt{(1-9b)(1-b)}+9b\left(-2 \pm \sqrt{(1-9b)(1-b)}+3b\right)}}{2\sqrt{2}}. \eeq

\section{Qualitatively different parts of the spectrum} \label{regions}
Up to this point we have discussed only the subset of $\sigma_{\mathcal{L}}$ that is on the imaginary axis. In this section we discuss the rest of the spectrum. In general, for all choices of the parameters $b$ and $k$, a part of the spectrum $\sigma_{\mathcal{L}}$ is in the right-half plane (corresponding to unstable modes). We split parameter space into five regions where $\overline{\sigma_{\mathcal{L}}\setminus i \mathbb{R}}$ is qualitatively different. Here $\overline{\sigma_{\mathcal{L}}\setminus i \mathbb{R}}$ refers to the closure of $\sigma_{\mathcal{L}}$ not on the imaginary axis.

We refer to Figure \ref{allcases}, which shows $(k,b)$ parameter space with curves that split it into regions where $\overline{\sigma_{\mathcal{L}}\setminus i \mathbb{R}}$ spectrum is qualitatively different. The exact curves splitting up the regions, as well as their derivations, are given below. In Figure \ref{boundaryspectrum}(1) we show representative plots of $\sigma_{\mathcal{L}}$ for the trivial-phase solutions on the boundary of parameter space, and in Figure \ref{boundaryspectrum}(2) we show the corresponding $\sigma_L$ spectrum. Additionally, we plot the $\zeta$ choices for which $\Omega(\zeta)\in i \mathbb{R}$. These curves are used to split up parameter space. The stability of trivial-phase solutions has been well studied in the literature \cite{De7,GH,IL,K03}. The Stokes wave solutions have constant magnitude and their stability problem has constant coefficients. Thus it is significantly easier to analyze.

\begin{figure}
\centering
\includegraphics[height=8cm]{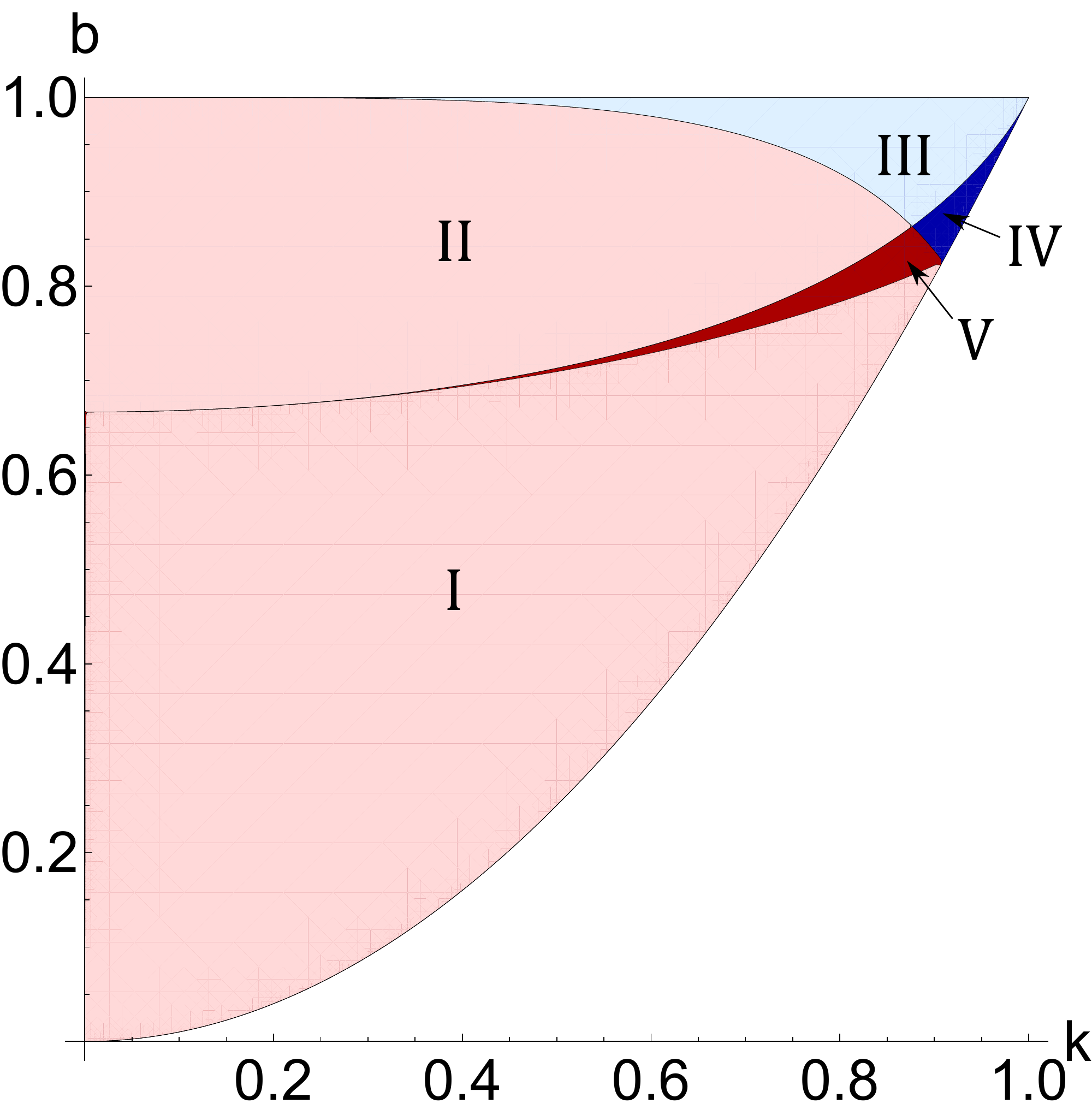}
\caption{A colored plot of parameter space with regions corresponding to different qualitative behavior in the linear stability spectrum. Regions I and II: two nested figure 8s; region III: non-self-intersecting butterflies; region IV:  self-intersecting butterflies; region V: one triple-figure 8 inside of a figure 8. }
\label{allcases}
\end{figure}

For the $\dn$ solutions, $\overline{\sigma_{\mathcal{L}}\setminus i \mathbb{R}}$ consists of a quadruple covered finite interval on the real axis. For Stokes wave solutions $\overline{\sigma_{\mathcal{L}}\setminus i \mathbb{R}}$ consists of a single-covered figure 8, and $\overline{\sigma_{\mathcal{L}}\setminus i \mathbb{R}}$ for $\cn$ solutions consists of a double covered figure 8. There are two cases for the $\cn$ solutions. Either $\sigma_{\mathcal{L}}\cap i \mathbb{R}$ pierces the figure 8 (see Figure \ref{boundaryspectrum}(1c)), or it does not (see Figure \ref{boundaryspectrum}(1d)). The exact value of $k$ separating the closure of the regions is given below.

For these trivial-phase cases, much can be proven and quantified explicitly, \textit{i.e.}, not in terms of special functions. Specifically, for the spectrum in the Stokes wave case we give a parametric description for the figure 8 curve.
For the spectrum for the $\dn$ case we calculate the extent of the covering of $\sigma_{\mathcal{L}}\cap i \mathbb{R}$.
For the spectrum in the piercing $\cn$ case, we give an explicit expression for where the top (or bottom) of the figure 8 crosses the imaginary axis.
Additionally, we have an explicit expression for the tangents to $\sigma_{\mathcal{L}}$ leaving the origin in both $\cn$ cases.
In fact, we are able to approximate the spectrum at the origin using a Taylor series to arbitrary order.
These series give a good approximation to the greatest real part of the figure 8 using only a few terms.

In the interior of parameter space we examine the nontrivial-phase solutions. Four cases appear and plots of the $\sigma_{\mathcal{L}}$ spectrum for representative choices of $k$ and $b$ are seen in Figure \ref{ntpspectrum}. The cases are as follows
\begin{enumerate}[label={(\ref{ntpspectrum}-1\alph*)}]
\item $\overline{\sigma_{\mathcal{L}}\setminus i \mathbb{R}}$ consists of two single-covered figure 8s, resulting in the degenerate double covered case of $\overline{\sigma_{\mathcal{L}}\setminus i \mathbb{R}}$ for $\cn$ solutions.
\item We have a single-covered non-self-intersecting butterfly. As $b\rightarrow 1$ the wings of this butterfly collapse to the real axis and the spectrum for $\dn$ solutions is seen with a quadruple covering on the real axis.
\item $\overline{\sigma_{\mathcal{L}}\setminus i \mathbb{R}}$ is a single-covered triple-figure 8 inside of a single-covered figure 8.
\item $\overline{\sigma_{\mathcal{L}}\setminus i \mathbb{R}}$ consists of a single-covered self-intersecting butterfly, which is seen as a perturbation of $\overline{\sigma_{\mathcal{L}}\setminus i \mathbb{R}}$ for the $\cn$ solutions as the double covered figure 8 splits apart horizontally.
\end{enumerate}

\begin{figure}
\centering
\begin{tabular}{cccc}
  \includegraphics[width=36mm]{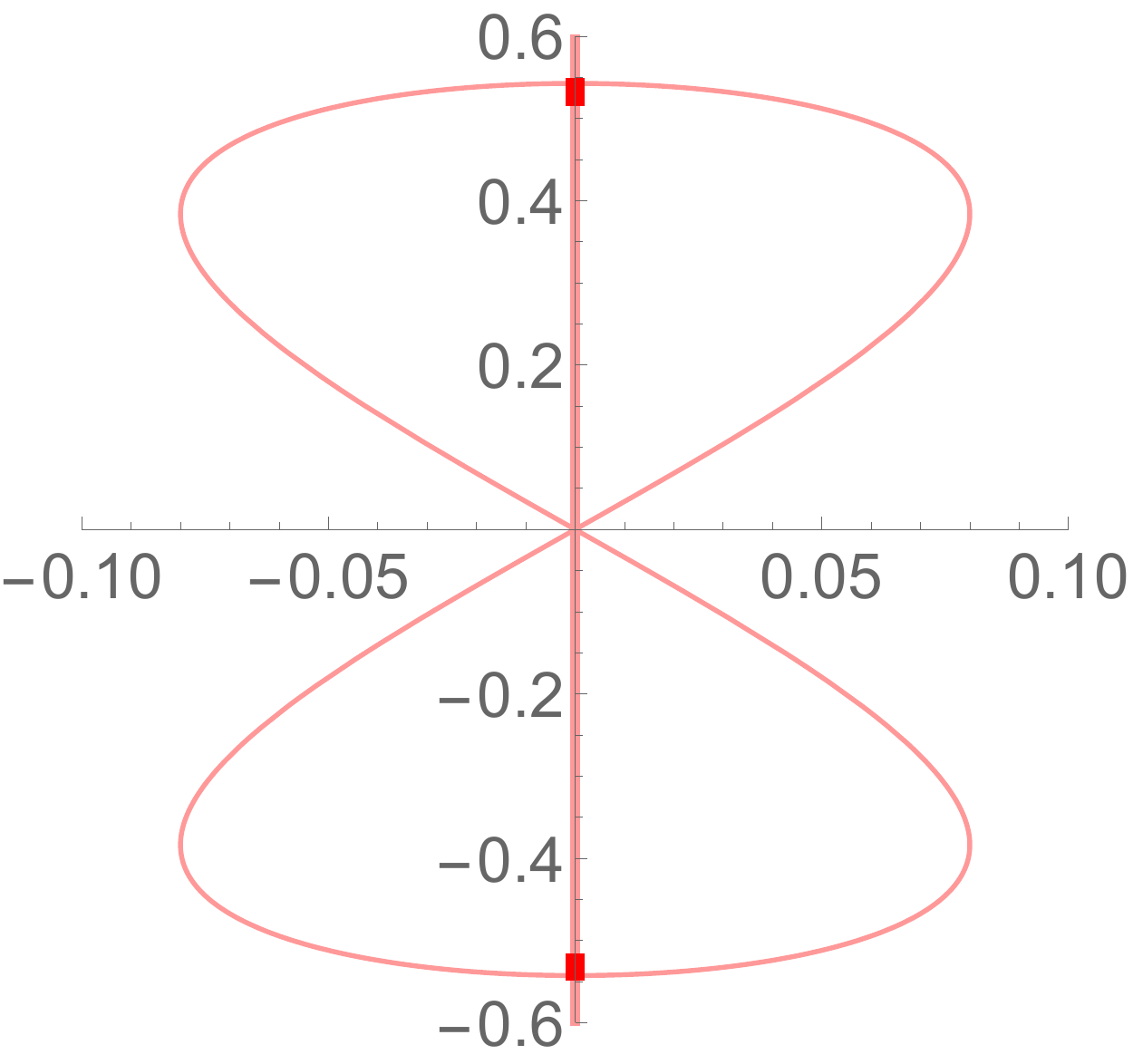} & \includegraphics[width=36mm]{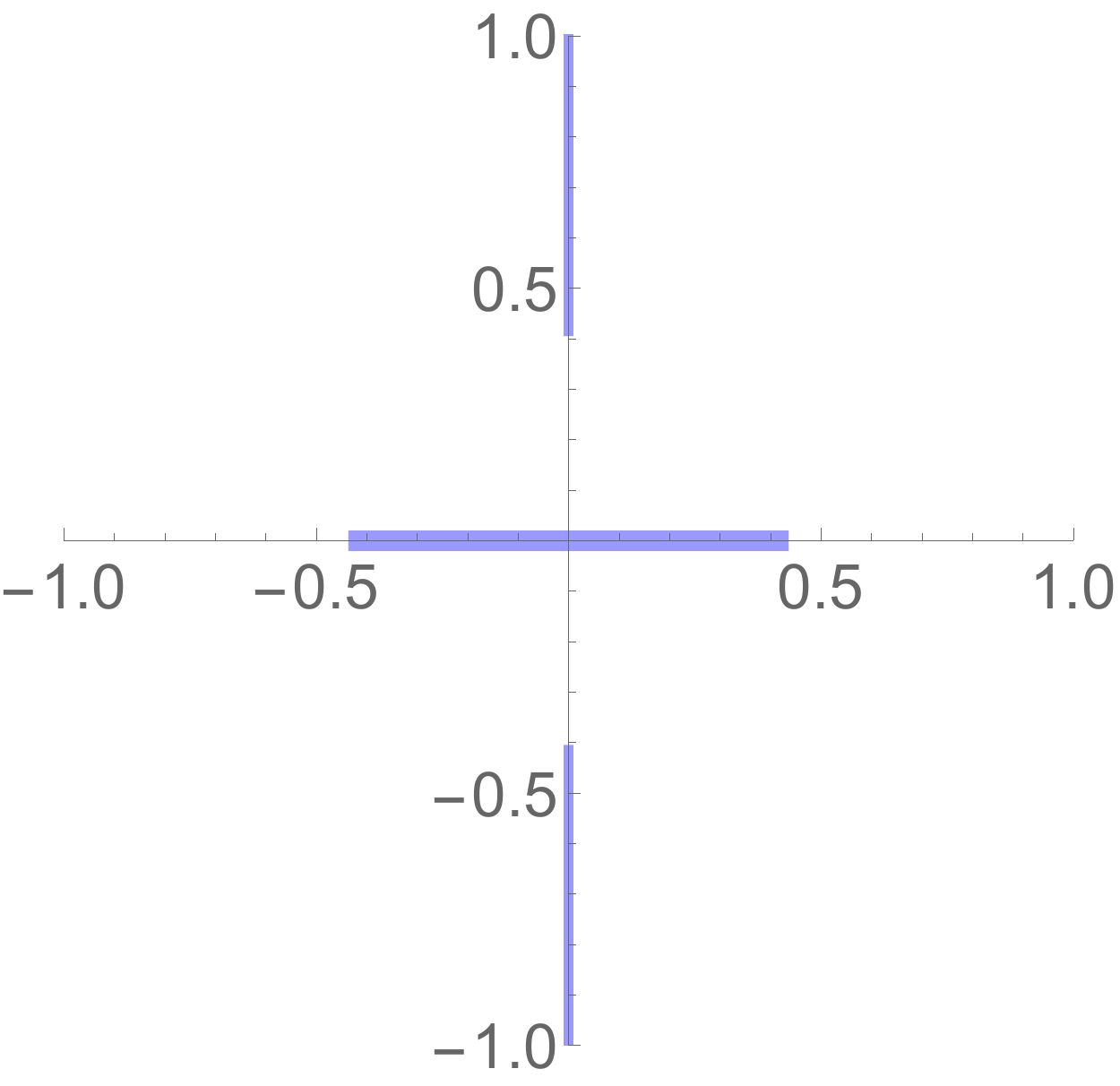} & \includegraphics[width=36mm]{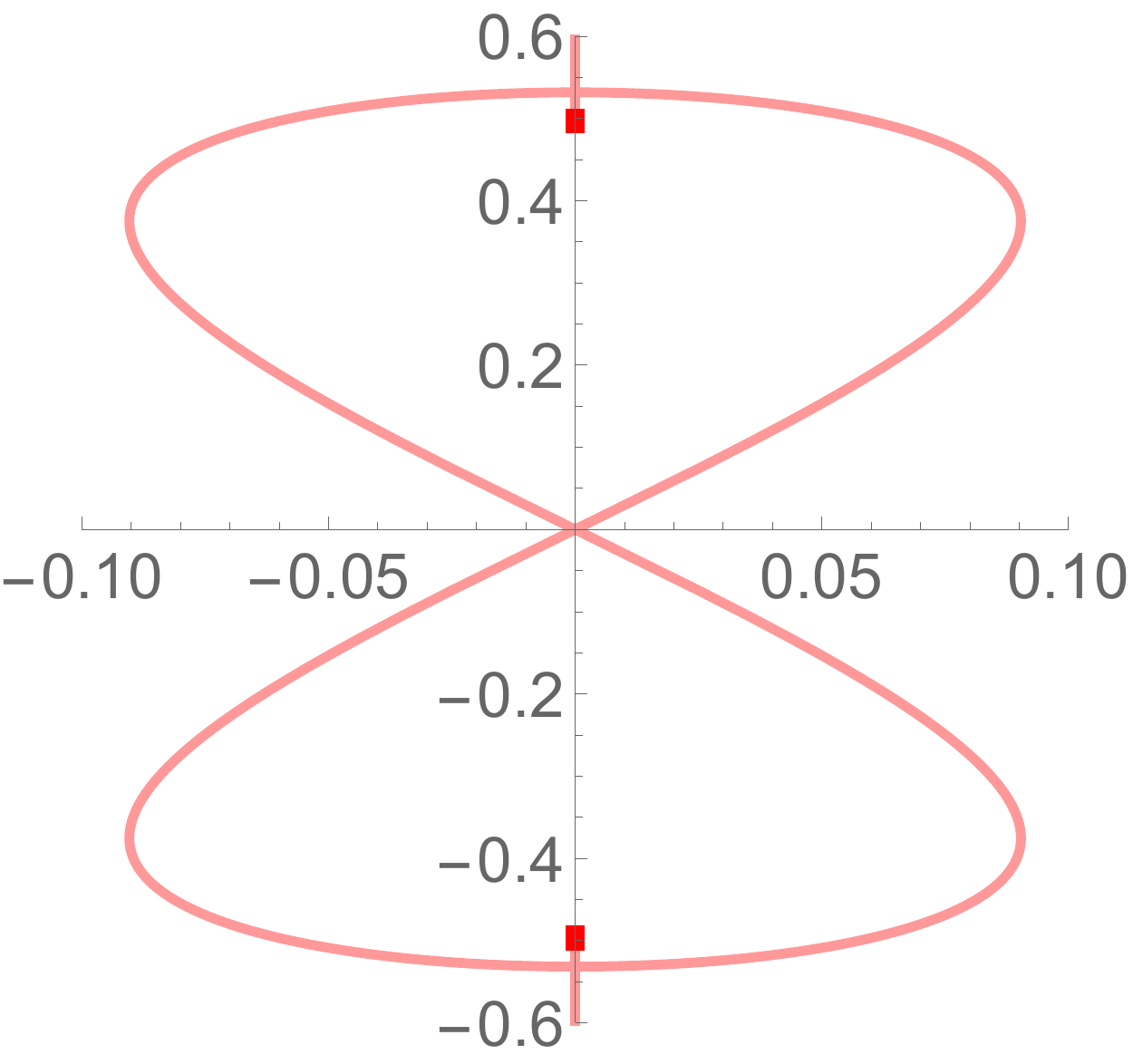} & \includegraphics[width=36mm]{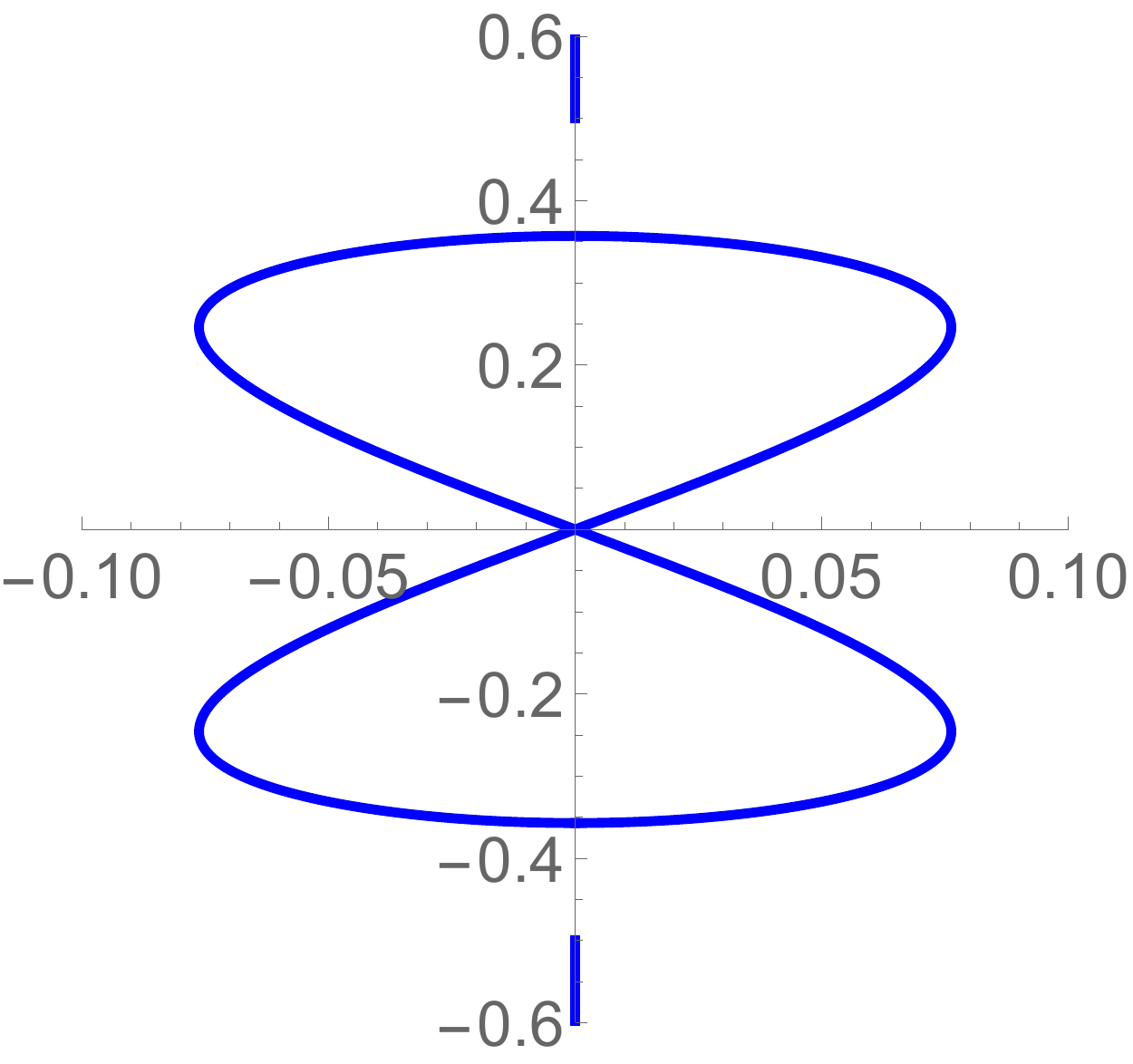} \\
(1a) & (1b) & (1c) & (1d)  \\[4pt]
 \includegraphics[width=36mm]{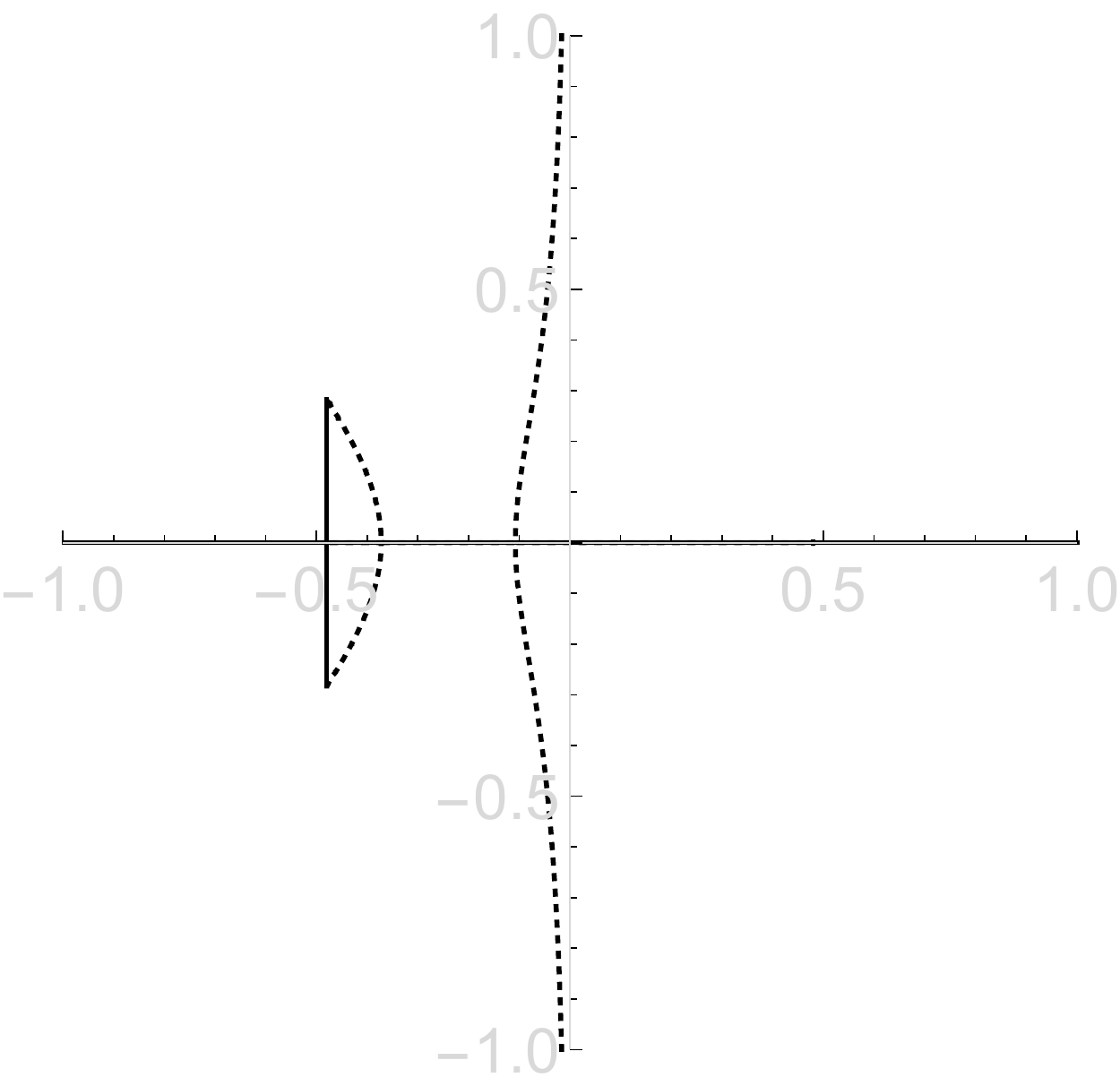} & \includegraphics[width=36mm]{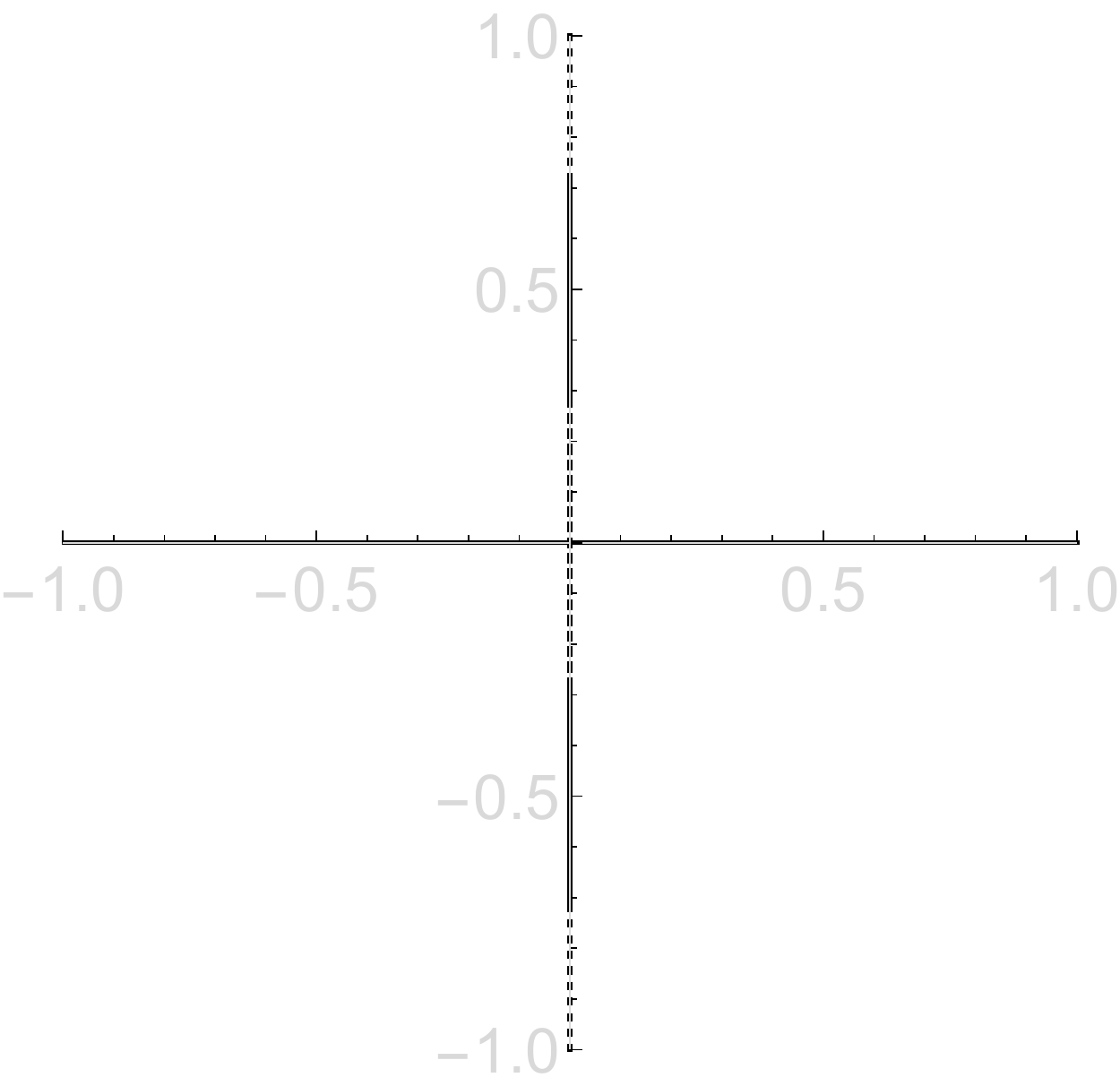} & \includegraphics[width=36mm]{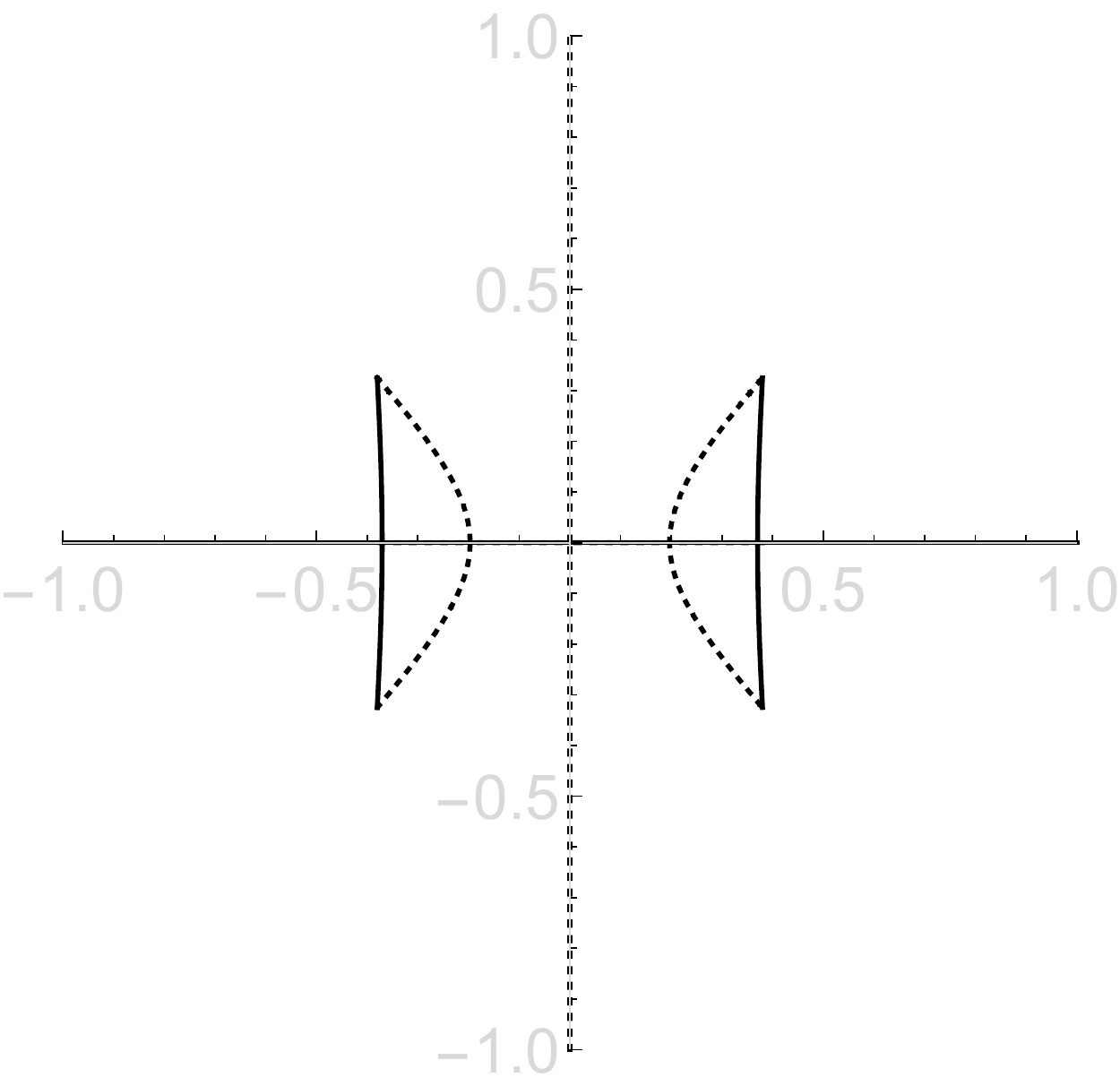} &  \includegraphics[width=36mm]{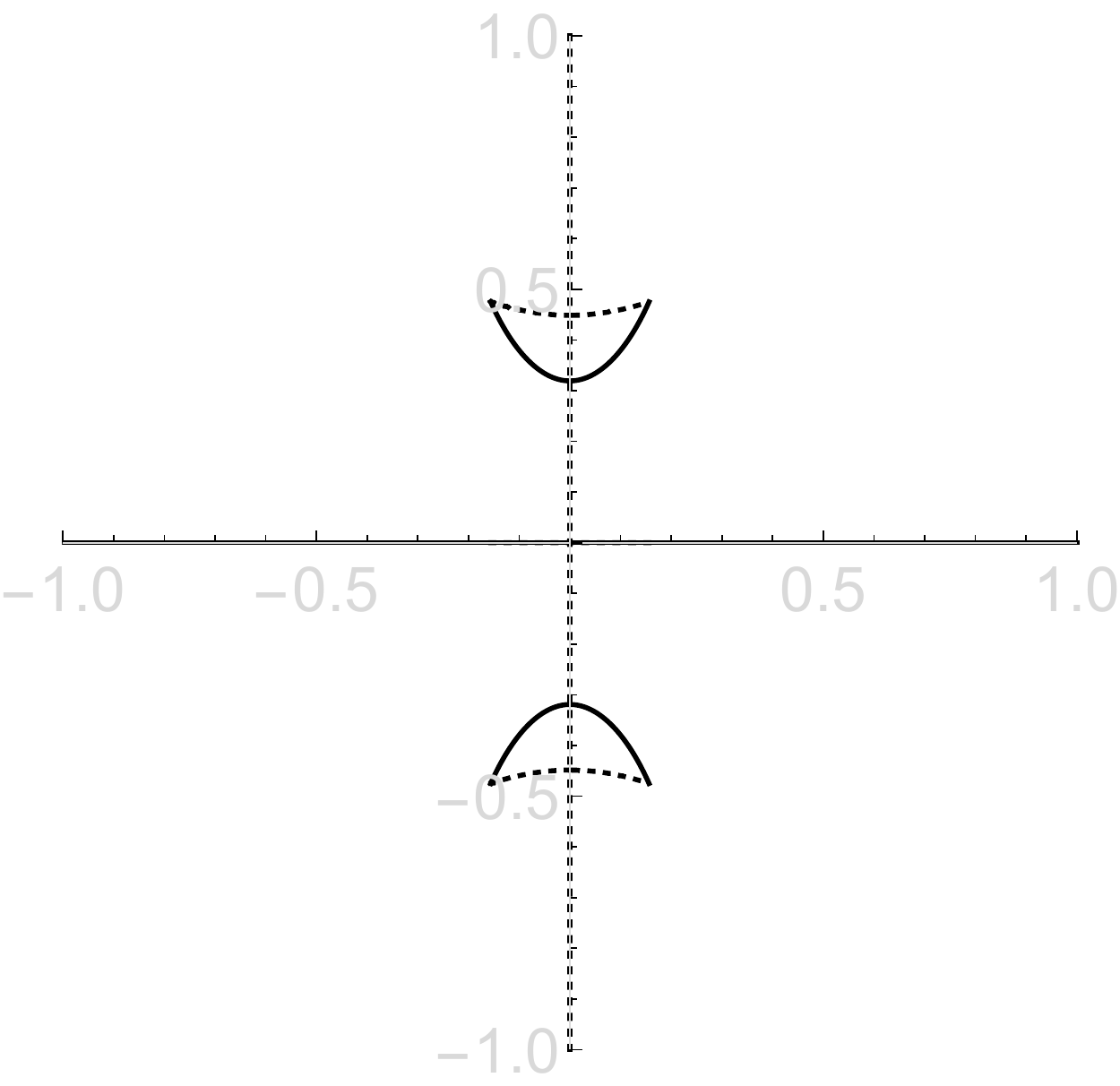} \\
(2a) & (2b) & (2c) & (2d)
\end{tabular}
\caption{(1) $\sigma_{\mathcal{L}}$ for the trivial-phase cases and (2) the corresponding $\sigma_L$ spectra (solid lines), values for which $\textrm{Re}\left(\Omega(\zeta)\right)=0$ (dotted). In (1), color corresponds to location in Figure \ref{allcases} and thickness of curves corresponds to single, double, or quadruple covering going from thinnest to thickest. (a) Stokes wave solution, $(k,b)=(0,0.08);$ (b) $\dn$ solution, $(k,b)=(0.9,1);$ (c) $\cn$ solution with piercing, $(k,b)=(0.65,0.4225);$ (d) $\cn$ solution without piercing, $(k,b)=(0.95,0.9025)$.}
\label{boundaryspectrum}
\end{figure}

In fact, there are two non-connected regions in parameter space for which we have two single-covered figure 8s, but qualitatively the spectrum is the same so we do not show samples from both regions.

For the nontrivial-phase case less can be determined explicitly. That said, we present an explicit expression for the slope of the spectrum for any nontrivial-phase solution as it leaves the origin. Since at least some of these slopes are finite, this settles the conjecture of Rowlands \cite{R74} that {\em all} stationary solutions of (\ref{fNLS}) are unstable. Moreover, a Taylor series expansion around the origin can be obtained for all cases and it well approximates the largest real part with a small number of terms. Additionally, explicit expressions for the tops (or bottoms) of the figure 8s in both cases with figure 8s are given.

A starting point for solving (\ref{intcond4}) for $\zeta$ is to recognize that if $\zeta$ satisfies $\Omega^2(\zeta)=0,$ then $\zeta$ must satisfy (\ref{intcond4}). This is due to the fact that the origin is always included in $\sigma_{\mathcal{L}}$ and hence in $S_\Omega$. In fact, the four roots of the quartic $\Omega^2=0$ corresponds to the fact that $0\in\sigma_{\mathcal{L}}$ with multiplicity four. This is seen from the symmetries of (\ref{fNLS}) and by applying Noether's Theorem \cite{KP, S07}.

It may be instructive to see this explicitly. In the general case, the roots of $\Omega^2(\zeta)$ are
\beq \label{zeta-roots} \zeta_c=\left\{\frac{\sqrt{1-b}}{2}\pm i \frac{\sqrt{b}-\sqrt{b-k^2}}{2} ,\, - \frac{\sqrt{1-b}}{2}\pm i \frac{\sqrt{b}+\sqrt{b-k^2}}{2} \right\}.\eeq
These roots are seen in Figures \ref{boundaryspectrum}-\ref{ntpspecialspectrum} (bottom) as the intersections between the solid and dotted lines lying off of the real axis. Indeed, as long as $b,k>0,$ these points have nonzero imaginary part, and other $\zeta\in \overline{\sigma_L\setminus \mathbb{R}}$ can be found by following the level curves of (\ref{intcond4}) originating from these points. For convenience we label these four roots $\zeta_1,\zeta_2,\zeta_3,\zeta_4,$ where the subscript corresponds to the quadrant on the real and imaginary plane the root is in.

\begin{figure}
\centering
\begin{tabular}{cccc}
  \includegraphics[width=36mm]{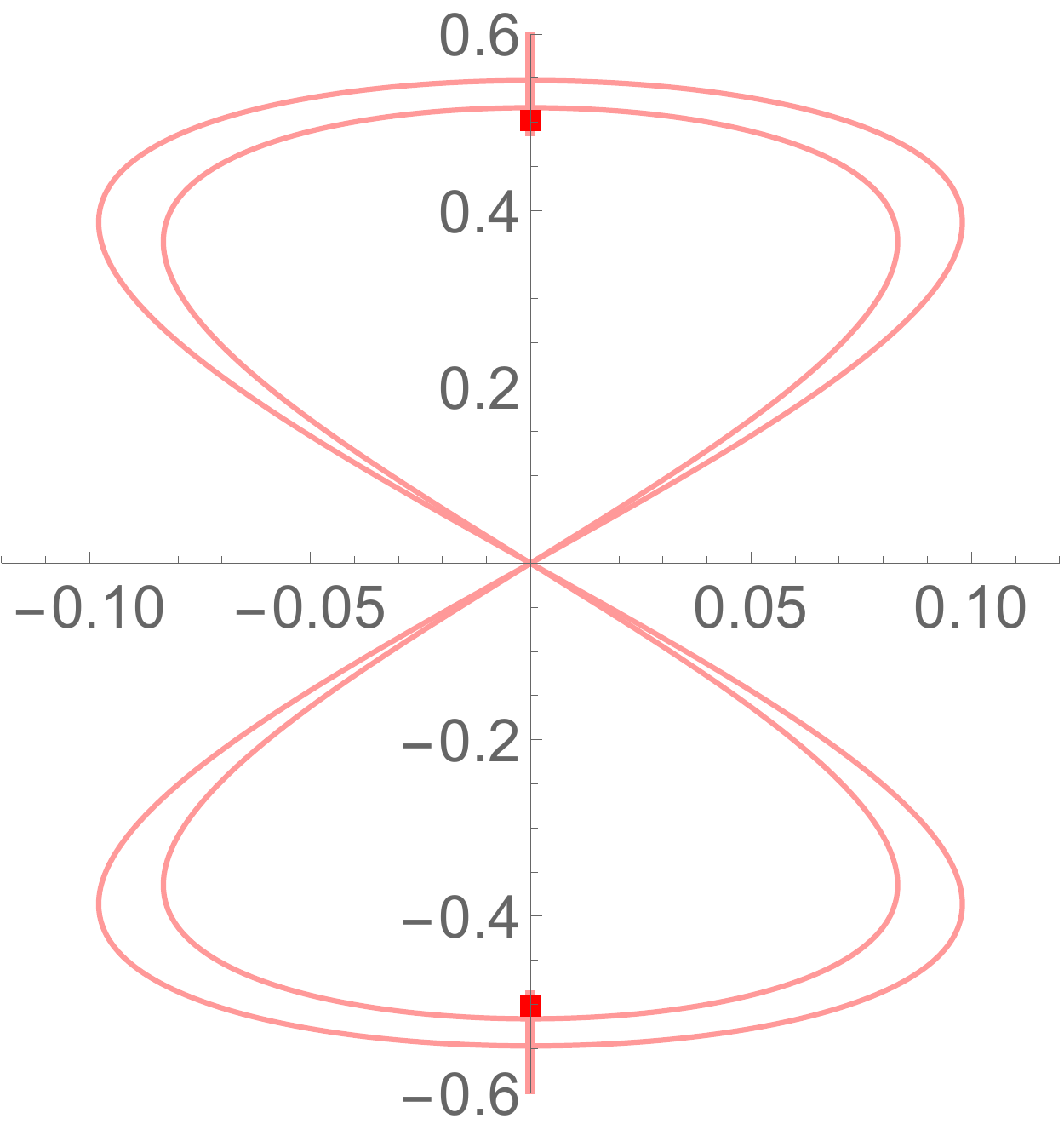} & \includegraphics[width=36mm]{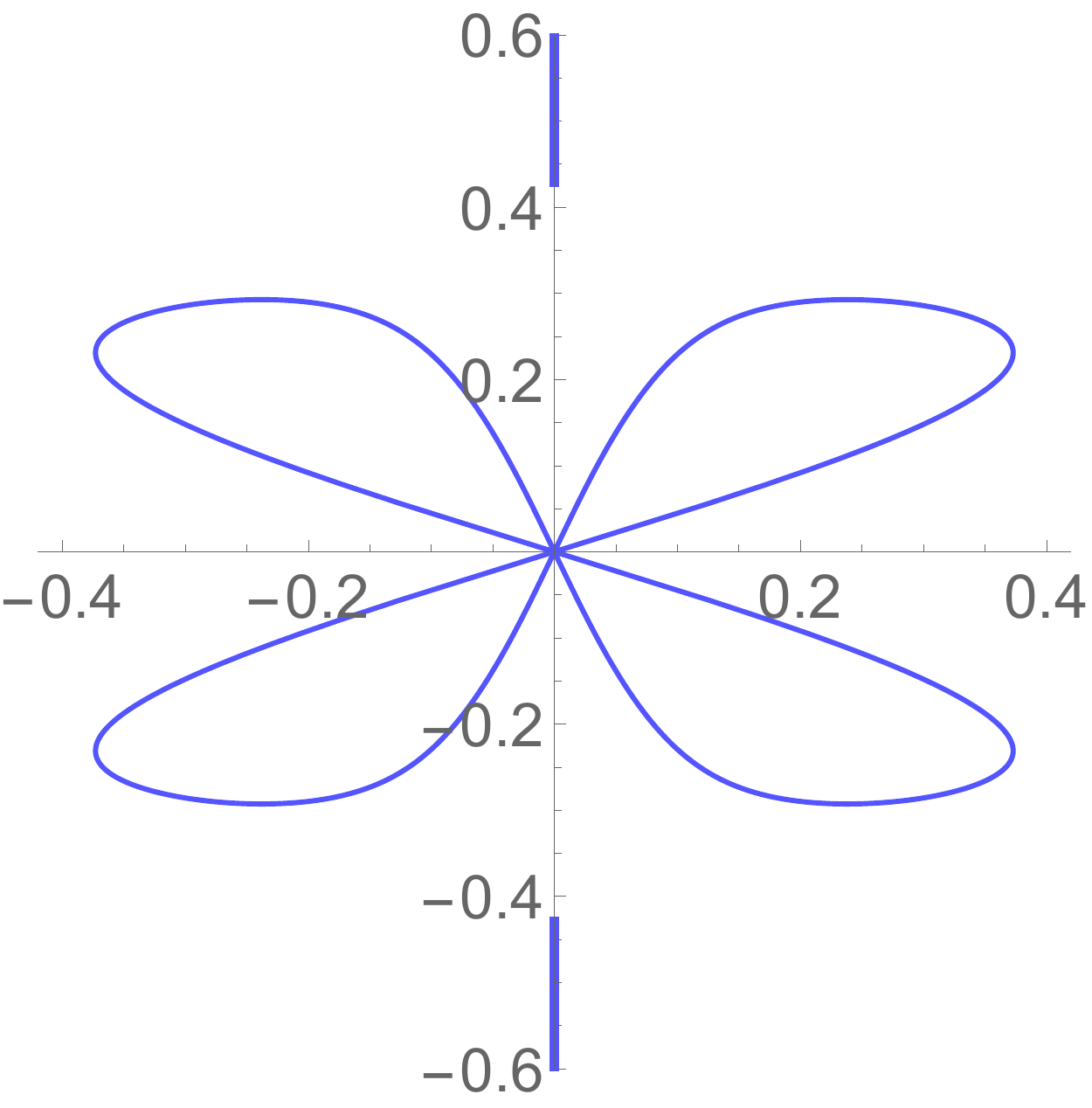} & \includegraphics[width=36mm]{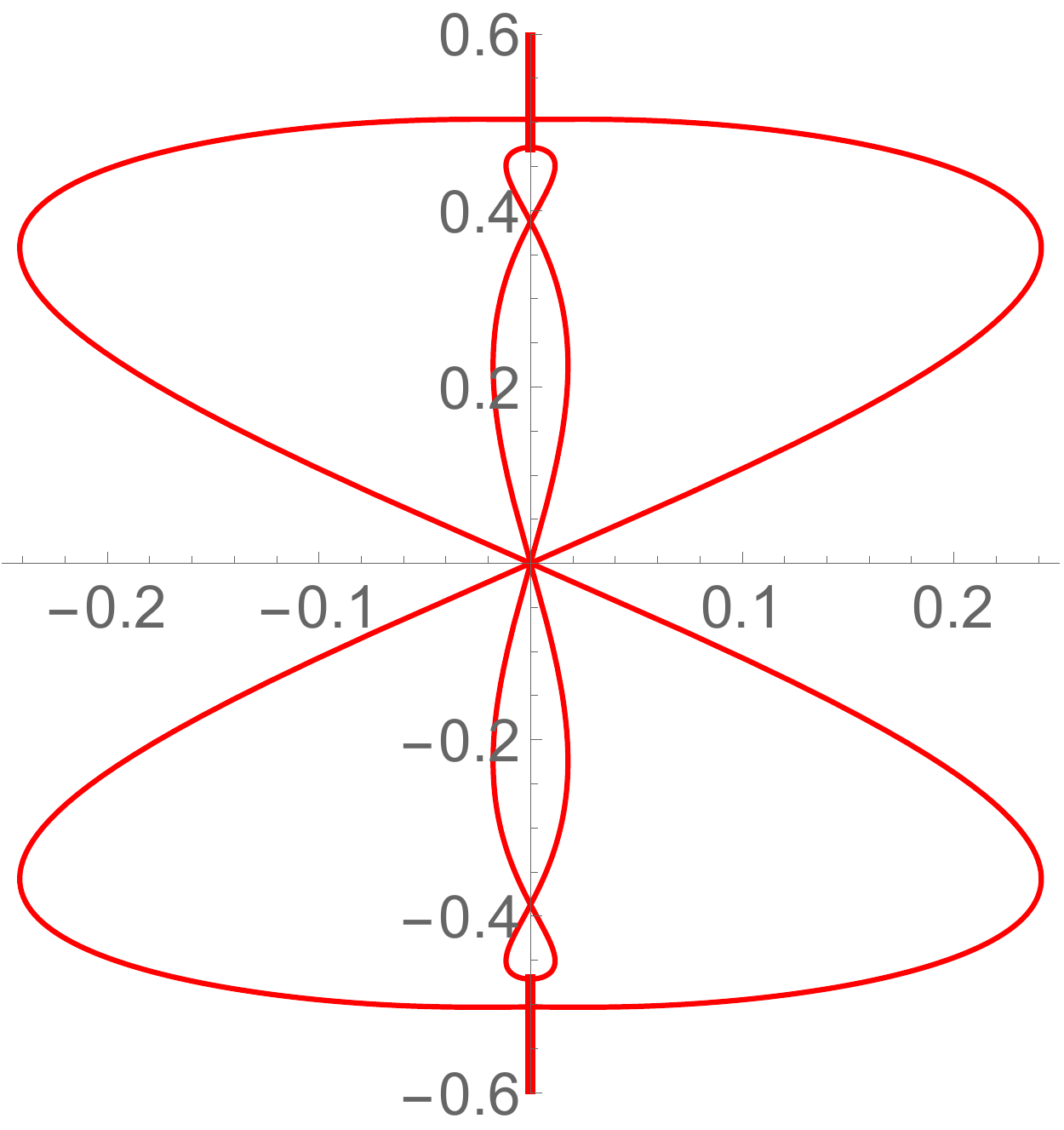} & \includegraphics[width=36mm]{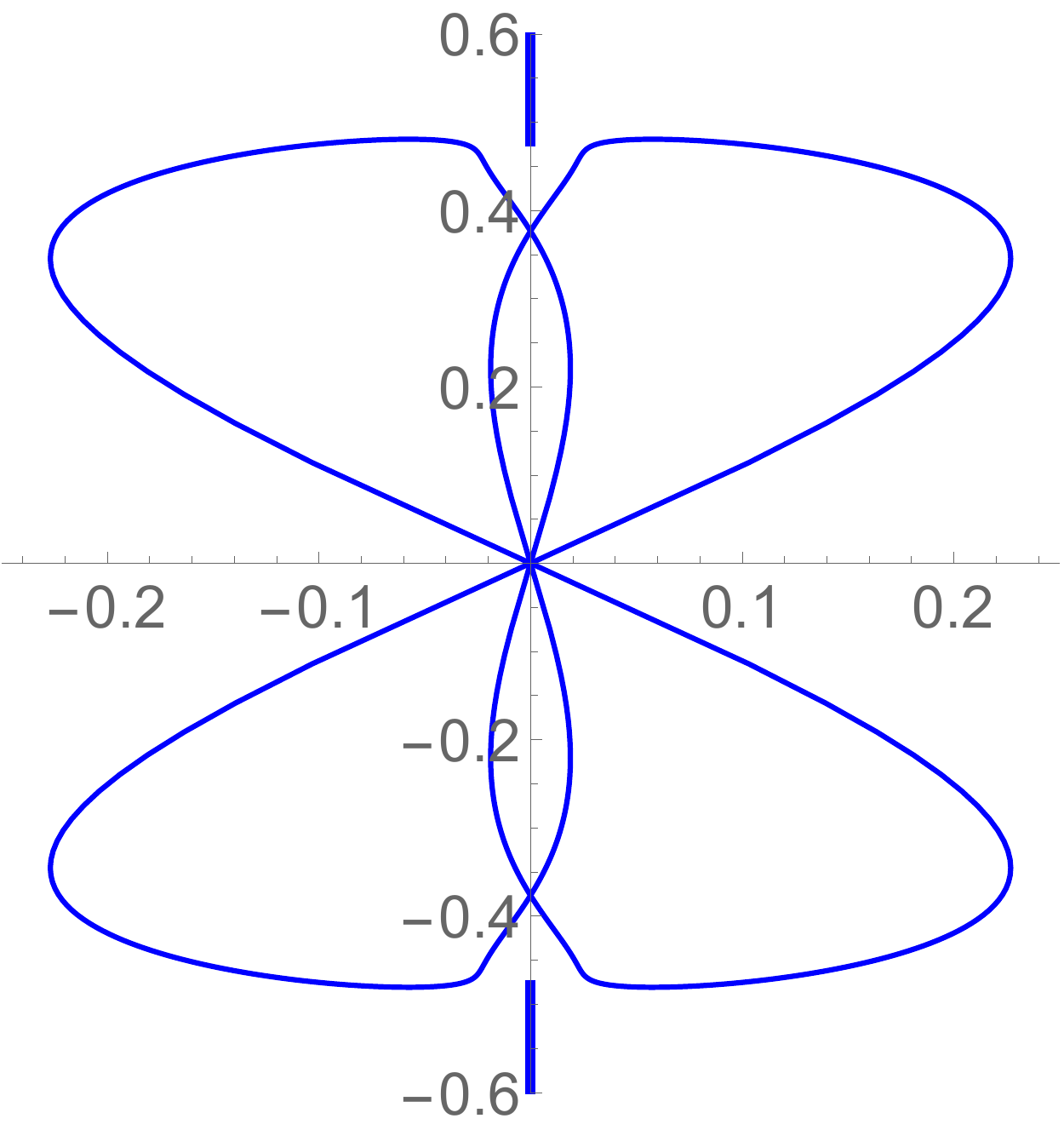} \\
(1a) & (1b) & (1c) & (1d)  \\[4pt]
\includegraphics[width=36mm]{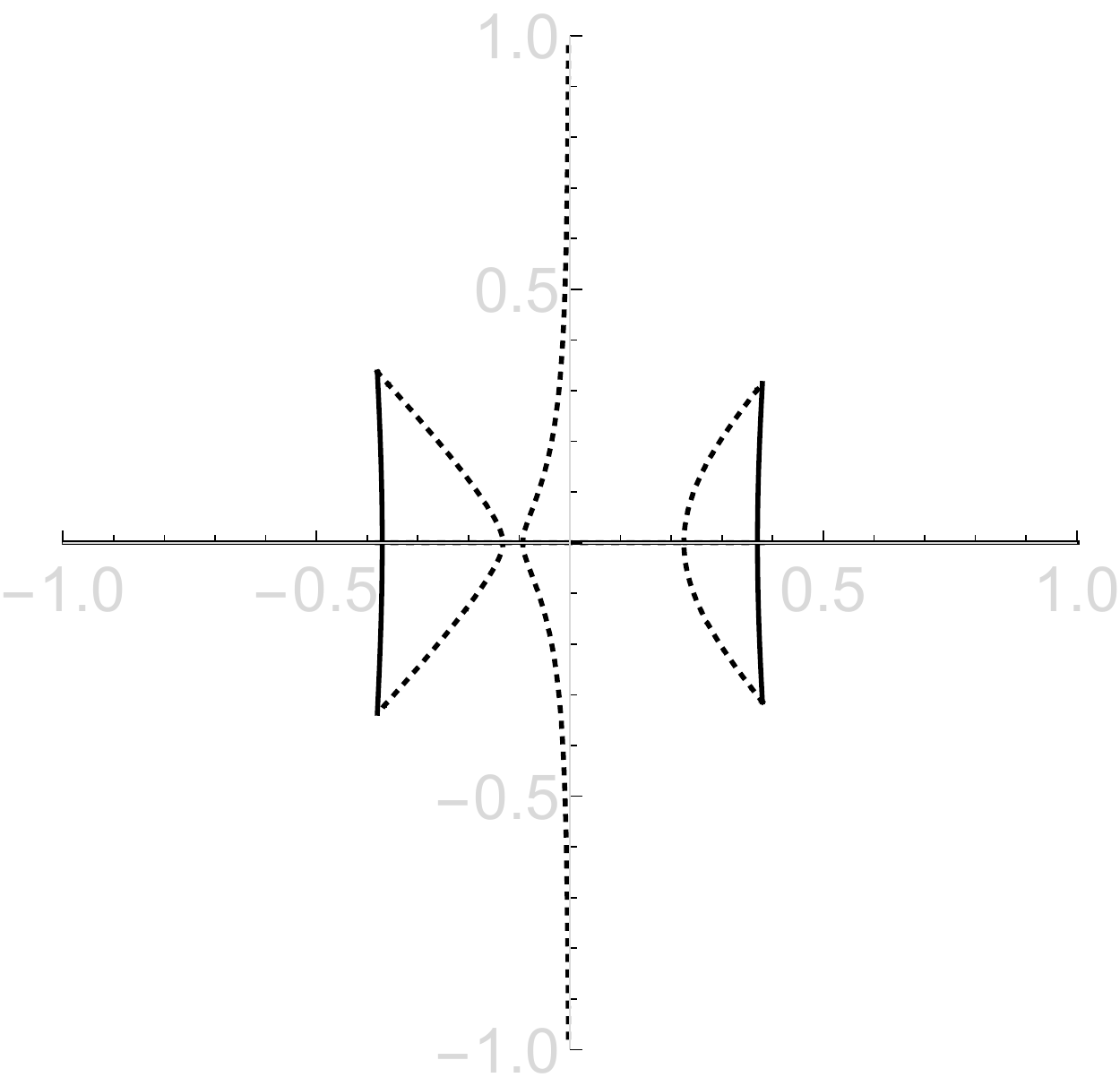} & \includegraphics[width=36mm]{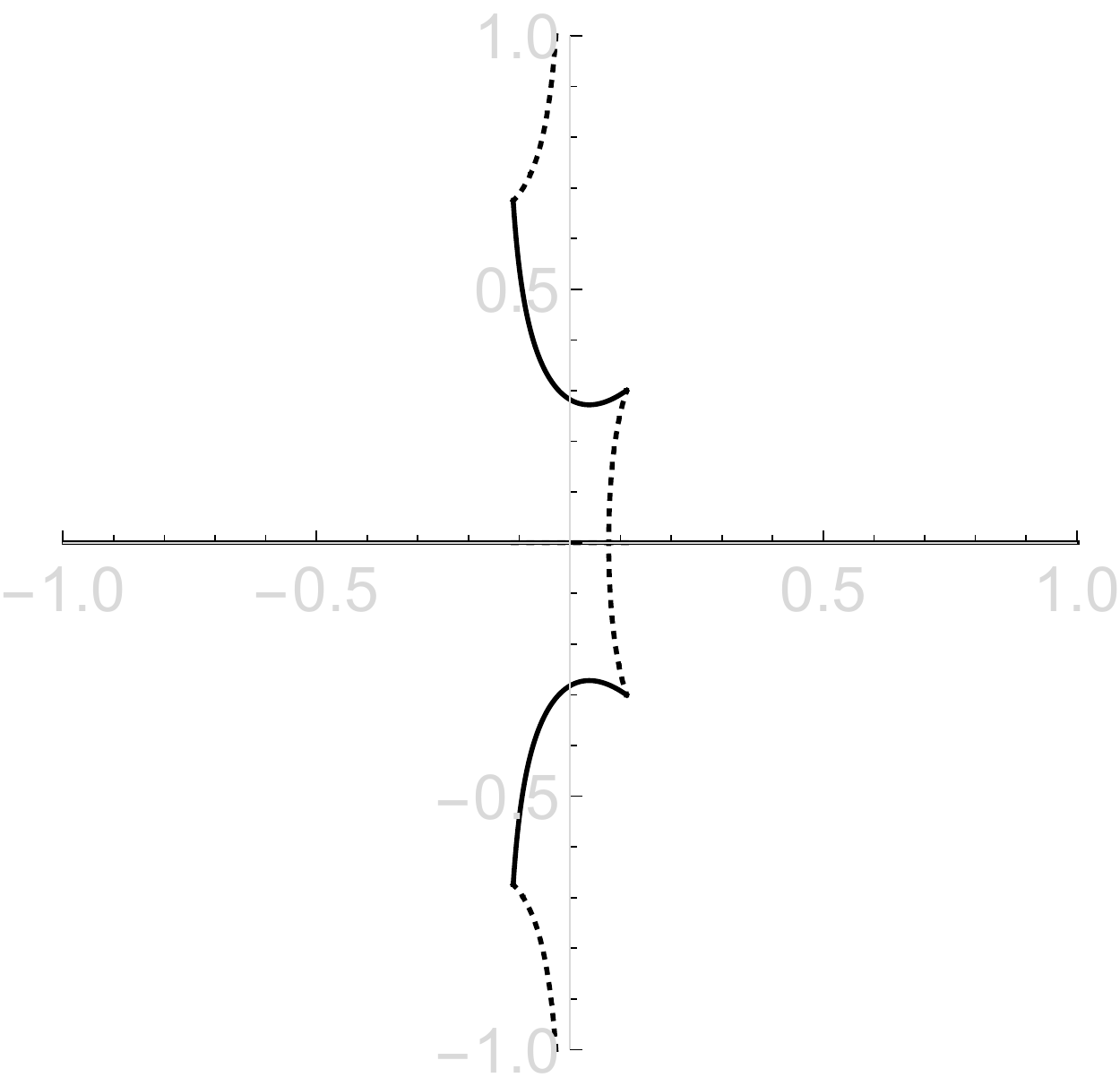} & \includegraphics[width=36mm]{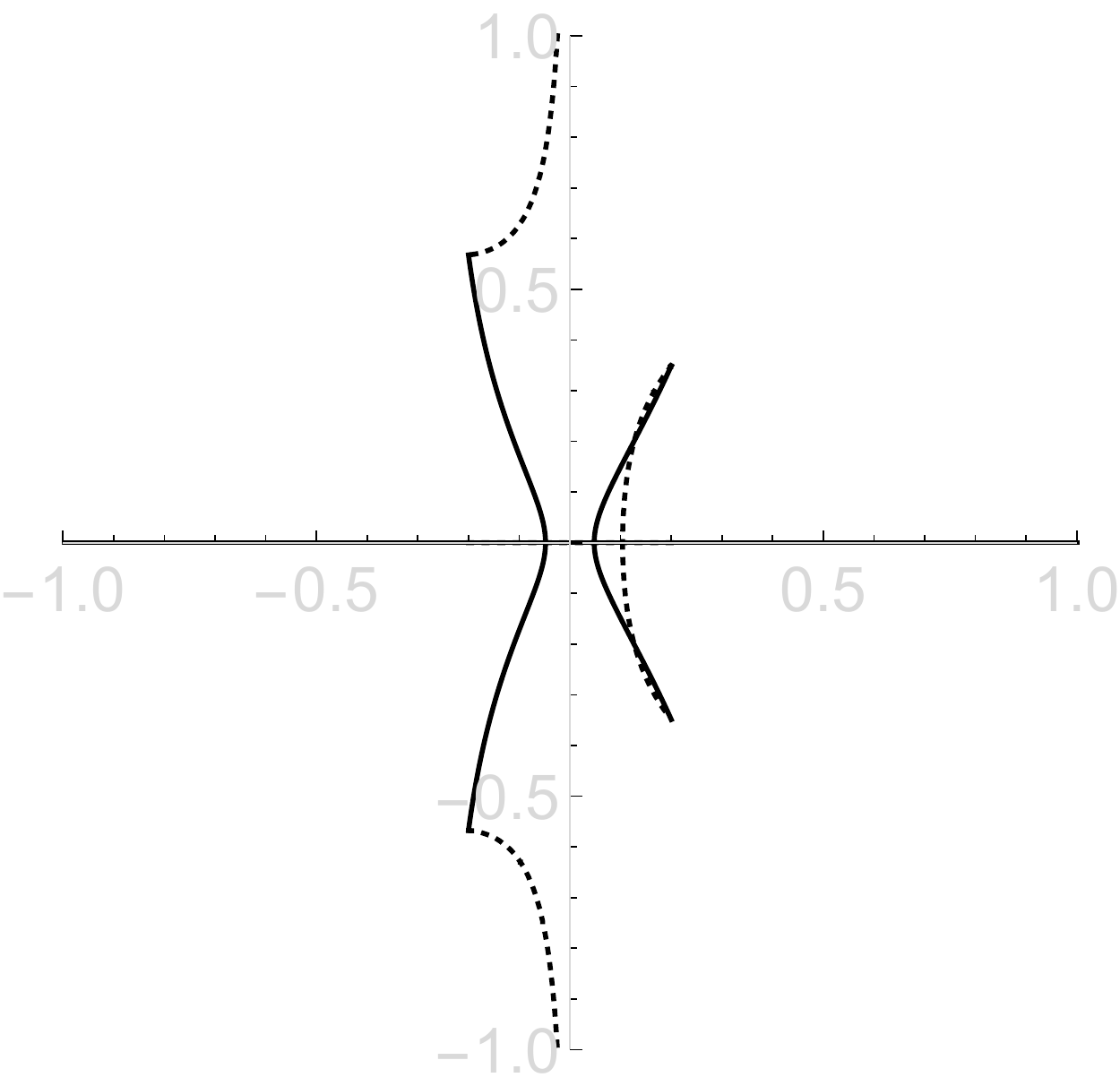} &  \includegraphics[width=36mm]{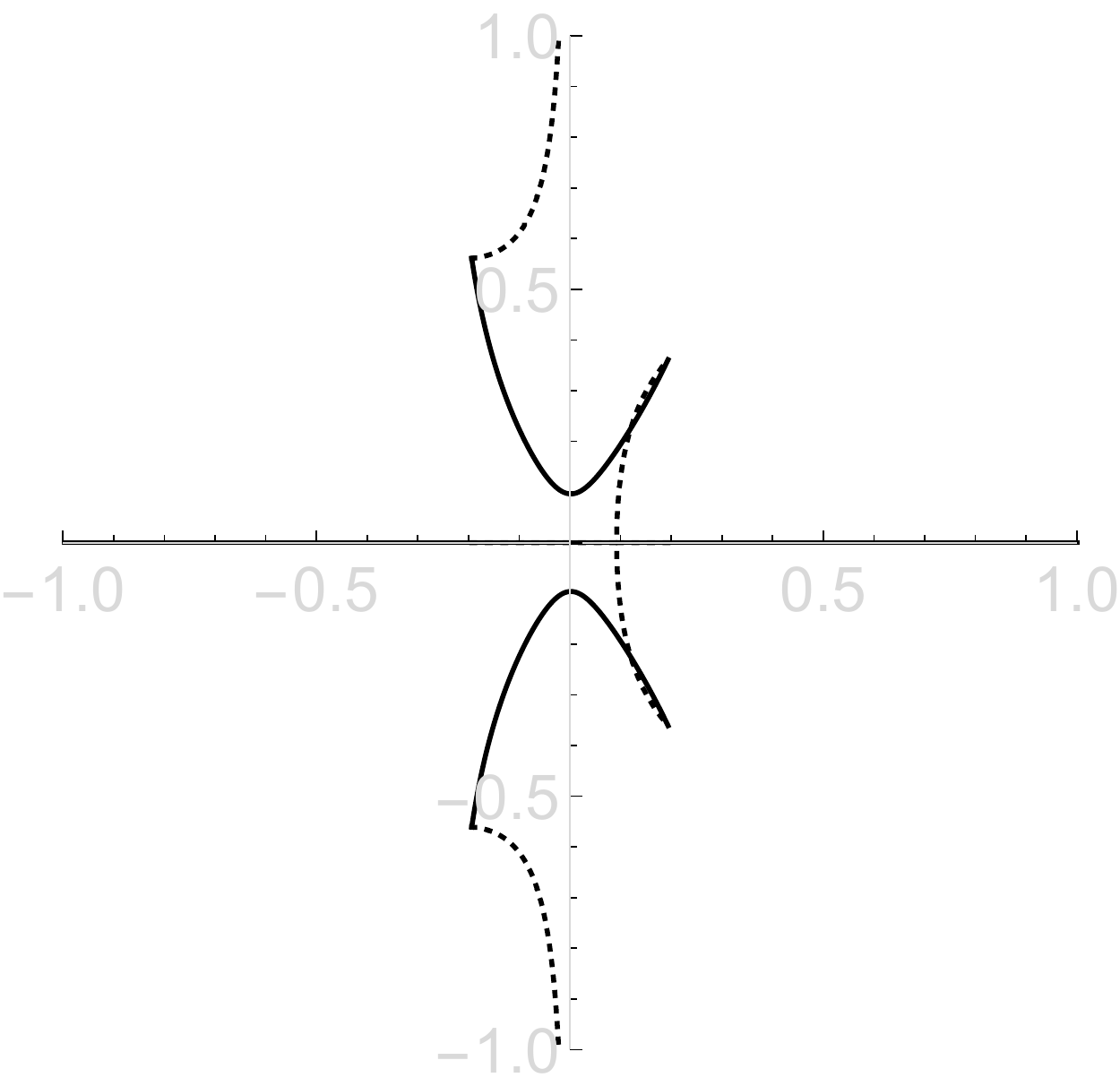} \\
(2a) & (2b) & (2c) & (2d)
\end{tabular}
\caption{(1) $\sigma_{\mathcal{L}}$ for the nontrivial-phase cases and (2) the corresponding $\sigma_L$ spectra (solid lines), values for which $\textrm{Re}\left(\Omega(\zeta)\right)=0$ (dotted). In (1), color corresponds to location in Figure \ref{allcases}. (a) Double-figure 8 solution, $(k,b)=(0.65,0.423);$ (b) non-self-intersecting butterfly solution, $(k,b)=(0.9,0.95);$ (c) triple-figure 8 solution, $(k,b)=(0.89,0.84);$ (d) self-intersecting butterfly solution, $(k,b)=(0.9,0.85)$.}
\label{ntpspectrum}
\end{figure}

To better examine this we look at the tangent vector field to the level curve (\ref{intcondI}). If we let $\zeta=\zeta_r+i \zeta_r$, then
\beq \label{Icond} I(\zeta)=I(\zeta_r+i \zeta_i) =  -2 i (\zeta_r+i \zeta_i) K(k) \pm 2 \left(\zeta_w(\alpha)K(k)-\left(E(k)-\frac{1}{3}\left(2-k^2\right) K(k)\right)\alpha \right).\eeq
The level curve $\left\{\zeta\in\mathbb{C}:\textrm{Re}\left[I(\zeta)\right] =0\right\},$ is exactly the condition for $\zeta\in\sigma_L$. Taking derivatives with respect to $\zeta_r$ and $\zeta_i$ gives a normal vector field to the level curves of the general condition  $\textrm{Re}\left[I(\zeta)\right] =C$ for any constant $C$, specifically, the normal vector is given by
$$ \left(\frac{\textrm{d}\textrm{Re} \left[I(\zeta_r+i \zeta_i)\right]}{\textrm{d} \zeta_r},\frac{\textrm{d} \textrm{Re}\left[I(\zeta_r+i \zeta_i)\right]}{\textrm{d} \zeta_i} \right).$$
Thus, the tangent vector field is
$$ \left(-\frac{\textrm{d} \textrm{Re}\left[I(\zeta_r+i \zeta_i)\right]}{\textrm{d} \zeta_i},\frac{\textrm{d}\textrm{Re} \left[I(\zeta_r+i \zeta_i)\right]}{\textrm{d} \zeta_r} \right).$$
By applying the chain rule and using the fact that $\textrm{Re}[i z] = -\textrm{Im}[z],$ we have that the tangent vector field to the level curves is
$$ \left(\textrm{Im}\left[ \frac{\textrm{d} I}{\textrm{d}\zeta} \right],\textrm{Re} \left[ \frac{\textrm{d} I}{\textrm{d} \zeta}\right] \right).$$
Substituting this back into (\ref{derivintcond}), the tangent vectors are
\beq \label{tanvfield} \left(\textrm{Im}\left[ \frac{2 E(k)-\left(1+b-k^2+4 \zeta^2\right)K(k)}{2 \Omega(\zeta)} \right],\textrm{Re} \left[\frac{2 E(k)-\left(1+b-k^2+4 \zeta^2\right)K(k)}{2 \Omega(\zeta)}\right] \right). \eeq
Thus, given a point in the $\sigma_L$ spectrum (lying on the 0 level curve of $\textrm{Re}\left[I(\zeta)\right]$), we can follow the tangent vector field to find other points in the $\sigma_L$ spectrum.
\subsection{Stokes wave case}
Applying this idea to the Stokes wave case, we see that generically
$$\zeta_c=\frac{\sqrt{1-b}}{2},\,\frac{\sqrt{1-b}}{2}, - \frac{\sqrt{1-b}}{2}\pm i \sqrt{b},$$
 \textit{i.e.}, there is a double root on the real axis and two conjugate roots. Following level curves we see that
\beq \label{zeta-stokes} \forall \zeta_i\in [-\sqrt{b},\sqrt{b}], \;\;\; -\frac{\sqrt{1-b}}{2}+i \zeta_i\in \sigma_L. \eeq
Substituting this into (\ref{OmegacondW}), we find that the $\sigma_{\mathcal{L}}$ spectrum for Stokes waves is given parametrically as a single-covered figure 8:
\beq \label{stokes-param} \lambda = \pm \left( 2\sqrt{b \zeta_i^2-\zeta_i^4}+2i \,\textrm{sgn}(\zeta_i)\sqrt{(1-b)(b-\zeta_i^2)}\right) \textrm{ for } \zeta_i \in [-\sqrt{b},\sqrt{b}]. \eeq
Plots of the $\sigma_L$ and the $\sigma_{\mathcal{L}}$ spectra are seen in Figure \ref{boundaryspectrum}(a) for $k=0,\,b=0.08$.
\subsection{dn case}
Similarly, in the $\dn$ case we find that
\beq \label{dn-zeta} \left[-\frac{1+\sqrt{1-k^2}}{2} i,-\frac{1-\sqrt{1-k^2}}{2} i\right]\cup \left[\frac{1-\sqrt{1-k^2}}{2} i,\frac{1+\sqrt{1-k^2}}{2} i\right]\in \sigma_L, \eeq
where $[\cdot,\cdot]$ corresponds to the straight line segment between its two endpoints.
Mapping this back to $\sigma_{\mathcal{L}}$ via (\ref{OmegacondW}), we find that there is a quadruple covering of the real axis
\beq \label{dn-real} \left[-\sqrt{1-k^2},\sqrt{1-k^2}\right]\in \sigma_{\mathcal{L}}. \eeq
Representative plots of these spectrum are seen in Figure \ref{boundaryspectrum}(b). This corrects a typo in \cite{K03}, and confirms the conjecture made in \cite{De7}.
\subsection{cn case}
For the $\cn$ case, less is known explicitly. Representative plots of the $\sigma_L$ spectrum are shown in Figure \ref{boundaryspectrum}(2c,2d). In both cases we have a quadrafold symmetry. The distinguishing factor between the two cases in (c) and (d) is whether or not $\overline{\sigma_L \setminus \mathbb{R}}$ leaving $\zeta_c$ crosses the real axis or the imaginary axis. Examining (\ref{tanvfield}) on the real axis we can determine the condition for a vertical tangent to occur. This happens when
\beq \label{CNtop} \zeta = \pm \frac{\sqrt{2 E(k)-K(k)}}{2 \sqrt{K(k)}}. \eeq
Equating $\zeta = 0,$ we solve for $k$ such that the vertical tangent occurs at the origin. With $2 E(k^*)-K(k^*) =0,$ we find that $k^*\approx 0.908909$.
This gives two cases: if $k<k^*$ then $\overline{\sigma_L \setminus \mathbb{R}}$ crosses the real axis, and if $k>k^*$ then $\overline{\sigma_L \setminus \mathbb{R}}$ crosses the imaginary axis. When $k<k^*$ we know the crossing of the real axis occurs when $\zeta$ satisfies (\ref{CNtop}). Mapping this back to $\sigma_{\mathcal{L}}$, we see that this point corresponds to the top (or bottom) of the figure 8
\beq \label{CNtopval} \lambda = \pm i \frac{\sqrt{(1-k^2)K^2(k)-2(1-k^2)E(k)K(k)+E^2(k)}}{K(k)}. \eeq
For all $k<k^*$ the figure 8 is pierced by the covering on the imaginary axis as seen in Figure \ref{boundaryspectrum}(1c), but as $k\rightarrow k^*,$ (\ref{CNtopval}) approaches $\pm i/2$ which is the extent of the covering on the imaginary axis as seen in Section \ref{imagaxis}. Thus for $k>k^*$, the figure 8 is no longer pierced by $\sigma_{\mathcal{L}}\cap i \mathbb{R}$, as is the case in Figure \ref{boundaryspectrum}(1d).

\subsection{Nontrival-phase cases}

Plots of generic cases of the $\sigma_L$ spectrum are seen in Figure \ref{ntpspectrum}(2a-d). The idea of whether $\overline{\sigma_L \setminus \mathbb{R}}$ crosses the real or imaginary axis still applies. The same analysis as above yields conditions on when $\zeta$ crosses the real axis.
We find that when
\beq \label{gentop} \zeta = \pm \frac{\sqrt{2 E(k)-K(k)-(b-k^2) K(k)}}{2 \sqrt{K(k)}}, \eeq
$\overline{\sigma_L \setminus \mathbb{R}}$ crosses the real axis. Mapping this back to $\sigma_{\mathcal{L}}$, this corresponds to the top (or bottom) of the outside figure 8:
\beq \label{largefigure8} \lambda = \pm i \sqrt{\frac{E^2(k)}{K^2(k)}-2 (1-b)\frac{E(k)}{K(k)}+(1-2b^2-k^2+2b k^2)+2c\sqrt{2 \frac{E(k)}{K(k)}+k^2-b-1}}, \eeq
and the top (or bottom) of the enclosed figure 8 (or triple-figure 8):
\beq \label{smallfigure8} \lambda = \pm i \sqrt{\frac{E^2(k)}{K^2(k)}-2 (1-b)\frac{E(k)}{K(k)}+(1-2b^2-k^2+2b k^2)-2c\sqrt{2 \frac{E(k)}{K(k)}+k^2-b-1}}.\eeq

We note that $\zeta < 0$ in (\ref{gentop}) corresponds to the top (or bottom) of the outside figure 8 in (\ref{largefigure8}), while $\zeta>0$ in (\ref{gentop}) corresponds to the top (or bottom) of the enclosed figure 8 in (\ref{smallfigure8}). This is difficult to show directly, but is seen from the more general result that for any $\zeta\in \mathbb{R},$ $\Omega^2(- |\zeta|)>\Omega^2(|\zeta|)$, which is derived directly from (\ref{Omegaw}).

Equating $\zeta=0$ in (\ref{gentop}) gives the condition for differentiating between figure 8's and butterflies:
\beq \label{splitbutterflies} b = -1+k^2+\frac{2 E(k)}{K(k)}. \eeq
If $b$ is less than this value the spectrum looks like in Figure \ref{ntpspectrum}(1a or 1c), and if $b$ is greater than this value the spectrum looks like in Figure \ref{ntpspectrum}(1b or 1d). In Figure \ref{ntpspecialspectrum}(a) we show the case when (\ref{splitbutterflies}) is exactly satisfied.

\begin{figure}
\centering
\begin{tabular}{cccc}
  \includegraphics[width=36mm]{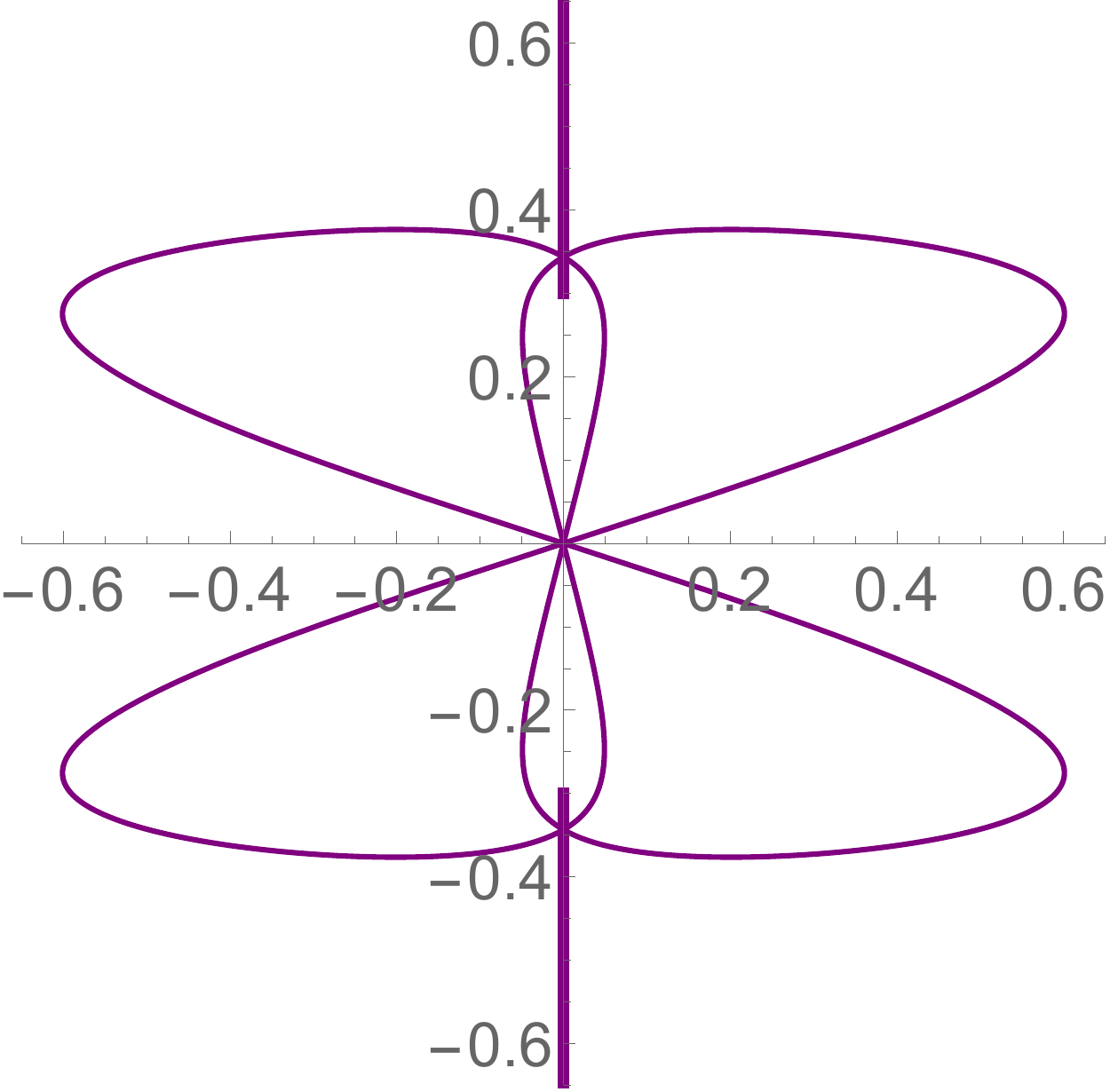} & \includegraphics[width=36mm]{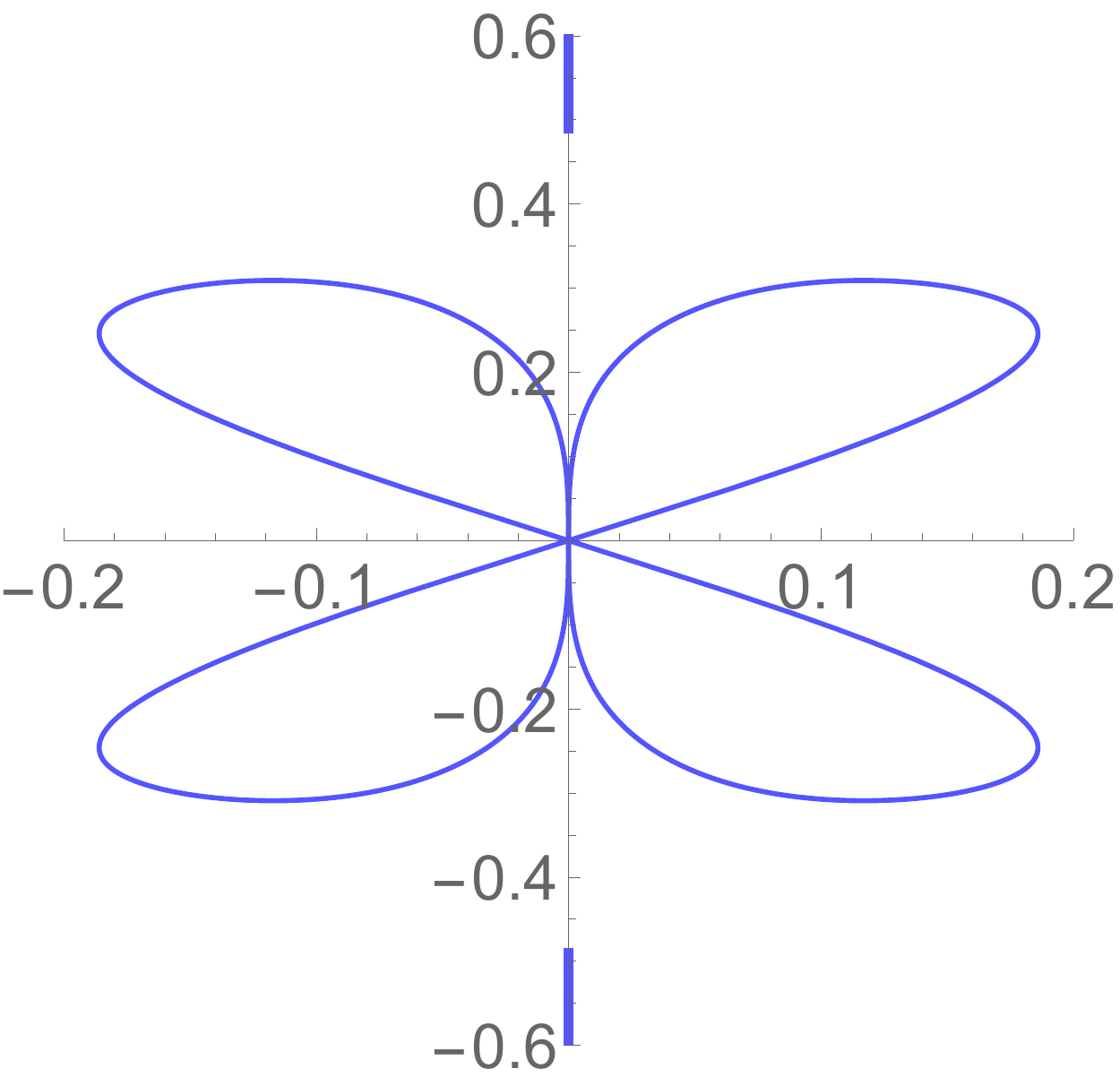} & \includegraphics[width=36mm]{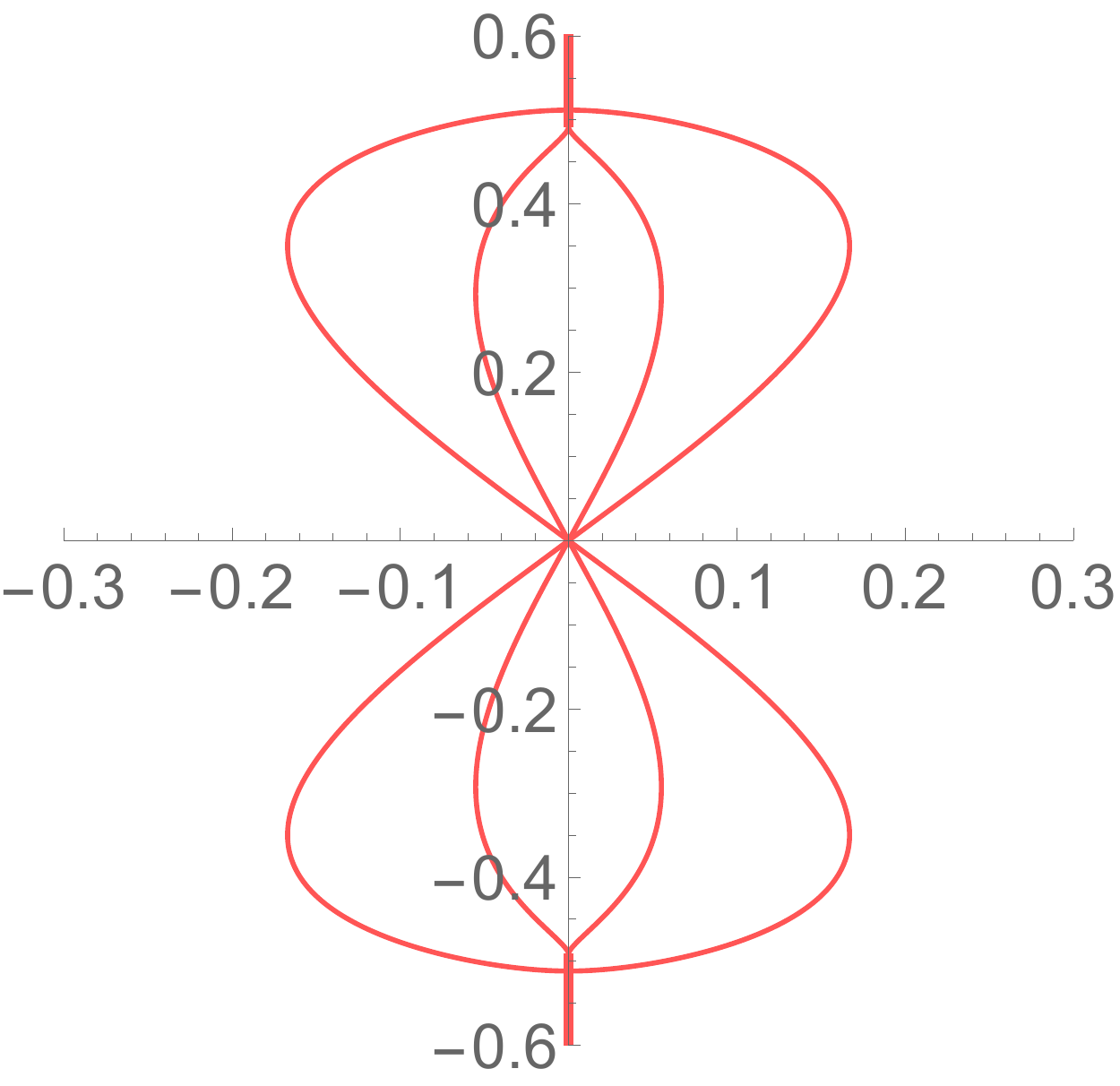} & \includegraphics[width=36mm]{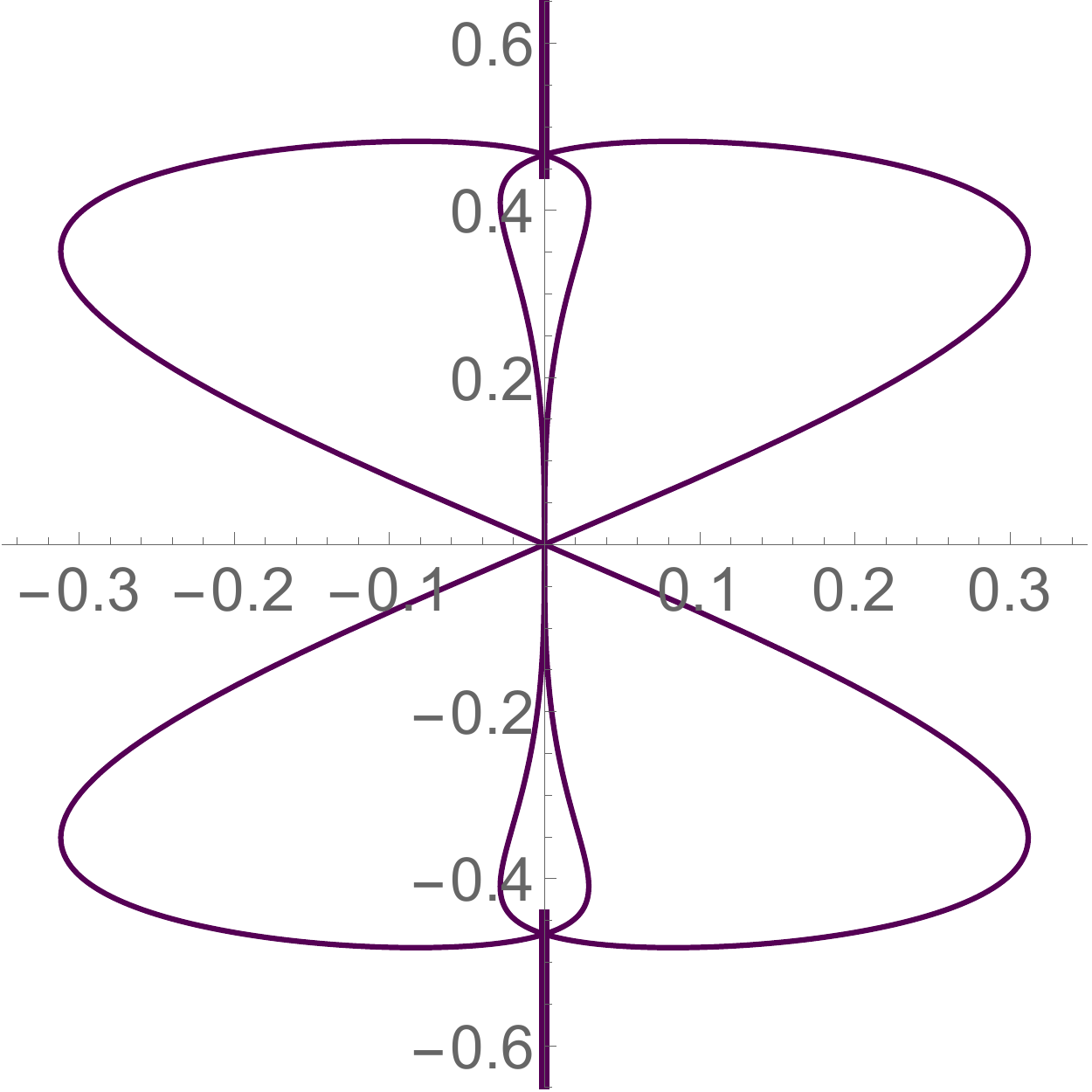} \\
(1a) & (1b) & (1c) & (1d)  \\[4pt]
  \includegraphics[width=36mm]{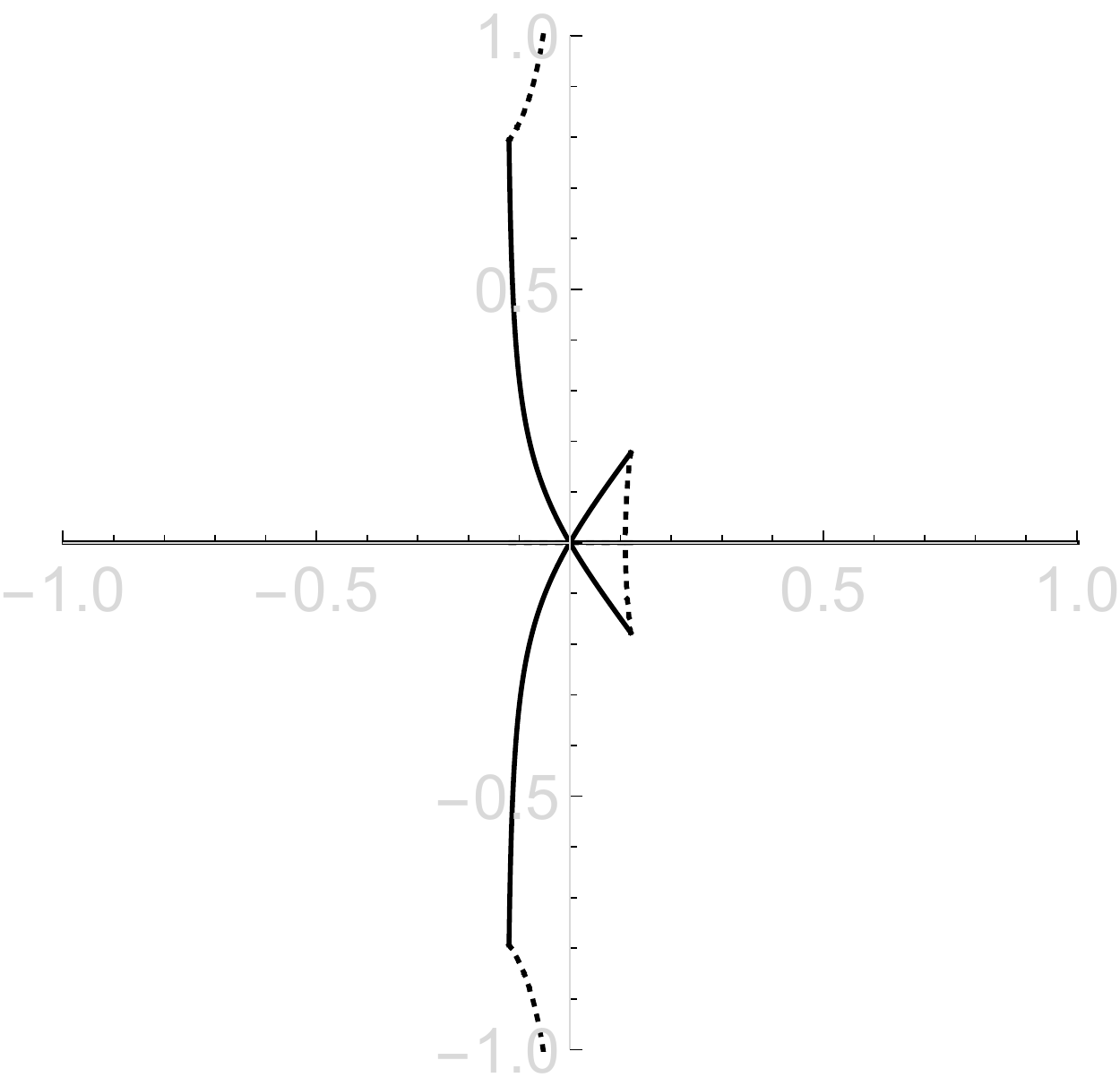} & \includegraphics[width=36mm]{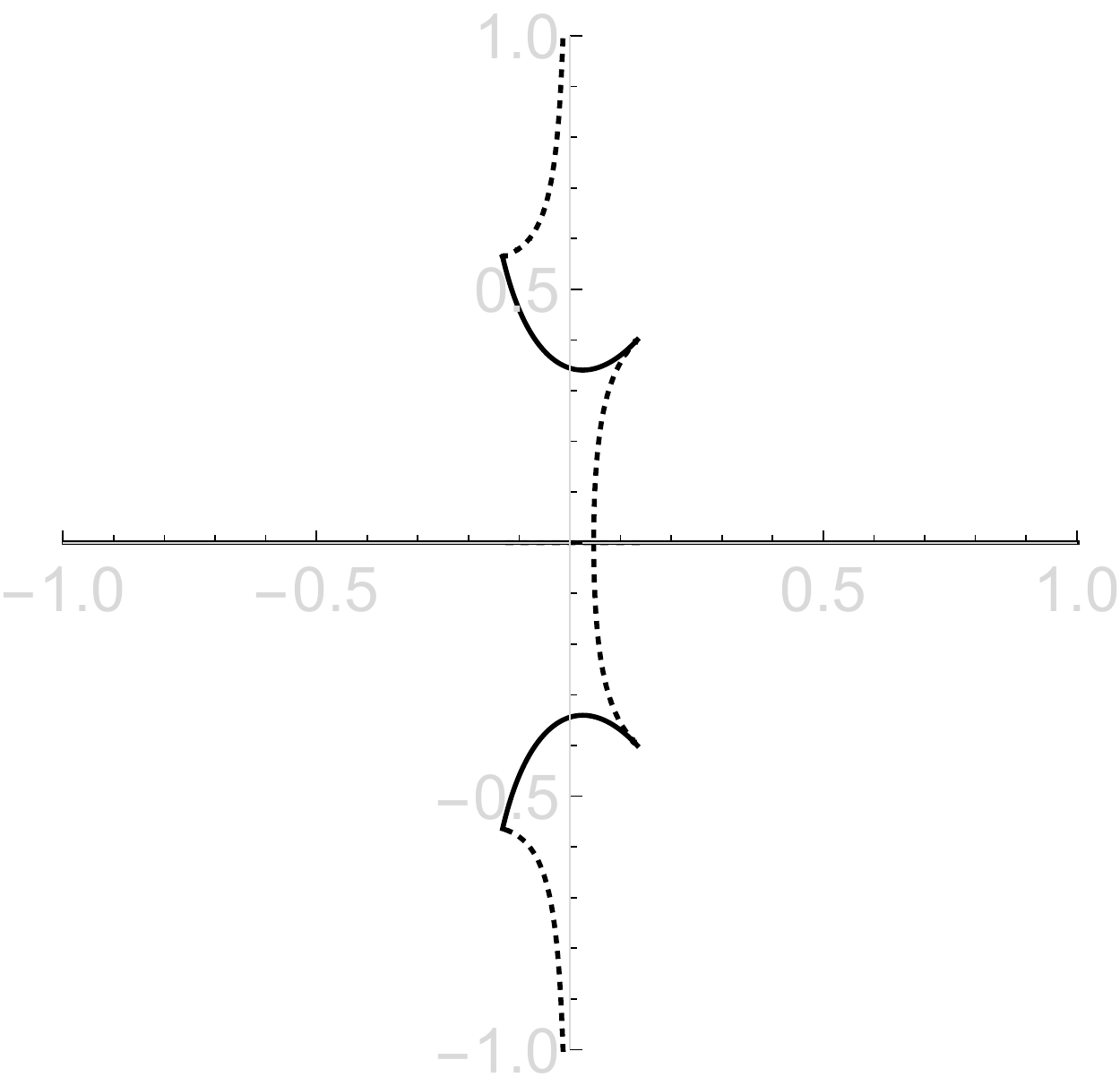} & \includegraphics[width=36mm]{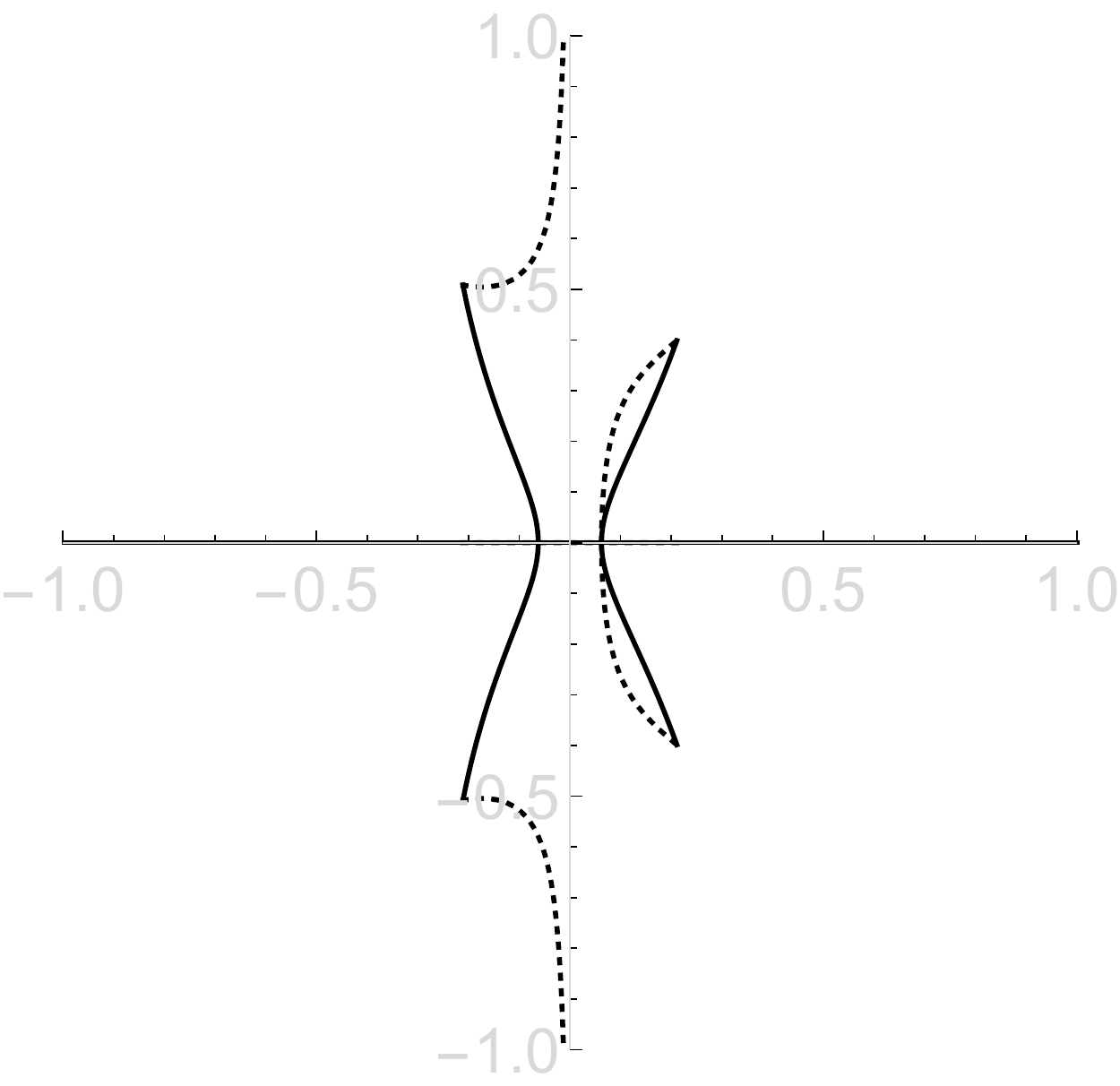} & \includegraphics[width=36mm]{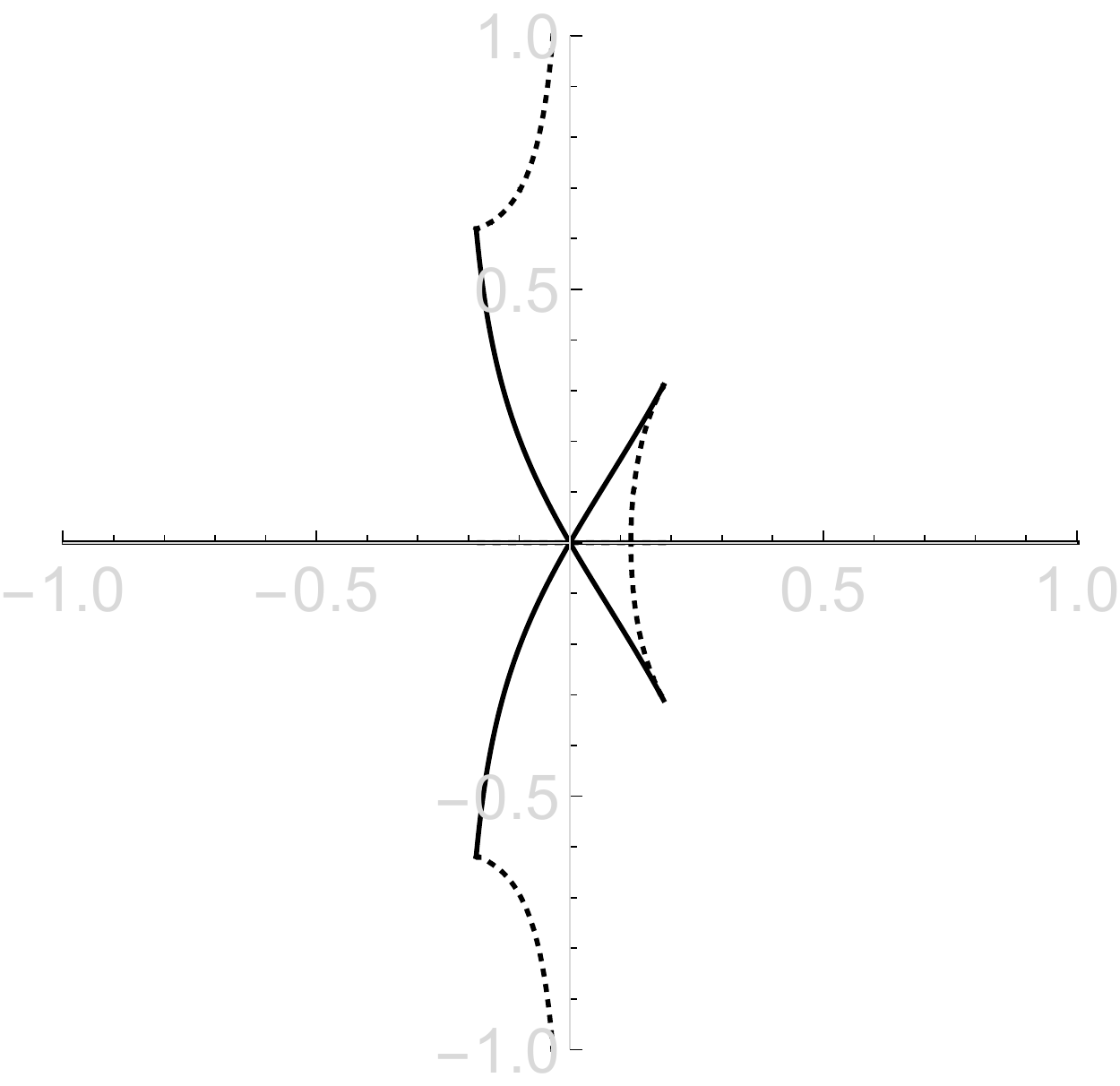} \\
(2a) & (2b) & (2c) & (2d)
\end{tabular}
\caption{(1) $\sigma_{\mathcal{L}}$ for the cases separating regions and (2) the corresponding $\sigma_L$ spectra (solid lines), values for which $\textrm{Re}\left(\Omega(\zeta)\right)=0$ (dotted). In (1), color corresponds to location in Figure \ref{allcases}. (a) Split between figure 8s and butterflies, $(k,b)=(0.75,0.942384);$ (b) split between self-intersecting and non-self-intersecting butterflies, $(k,b)=(0.95,0.929542);$ (c) lower split between figure 8 and triple-figure 8, $(k,b)=(0.9,0.821993);$ (d) four-corners point, $(k,b)=(0.876430,0.863399)$.}
\label{ntpspecialspectrum}
\end{figure}

Next we examine the slopes of the $\sigma_{\mathcal{L}}$ curves at the origin. Because $\sigma_{\mathcal{L}}=2 S_\Omega$ it suffices to examine the slopes for the set $S_\Omega$. We let $\Omega = \Omega_r+i \Omega_i,$ and we consider $\zeta_i$ as a function of $\zeta_r$ so that $\Omega\left(\zeta_r,\zeta_i(\zeta_r)\right)$. Applying the chain rule we have that the slope at any point in the set $S_\Omega$ is
\beq \label{derivO} \frac{ \textrm{d}\Omega_i}{\textrm{d} \Omega_r} = \frac{\textrm{d}\Omega_i/\textrm{d}\zeta_r}{\textrm{d}\Omega_r/\textrm{d}\zeta_r} = \frac{ \frac{\textrm{d}\Omega_i}{\textrm{d} \zeta_r}+\frac{\textrm{d} \Omega_i}{\textrm{d} \zeta_i} \frac{ \textrm{d} \zeta_i}{\textrm{d} \zeta_r}}{\frac{\textrm{d} \Omega_r}{\textrm{d}\zeta_r}+\frac{\textrm{d} \Omega_r}{\textrm{d} \zeta_i} \frac{\textrm{d} \zeta_i}{\textrm{d} \zeta_r}},\eeq
where
\beq \label{derivZ} \frac{\textrm{d} \zeta_i}{\textrm{d} \zeta_r} = - \frac{\textrm{d} \textrm{Re}(I)/\textrm{d}\zeta_r}{\textrm{d} \textrm{Re}(I)/\textrm{d} \zeta_i}. \eeq
We examine (\ref{derivO}) near where $\Omega=0$ and $\zeta=\zeta_c$.
The slopes around the origin are
\beq \label{slopes1}\frac{ \textrm{d}\Omega_i}{\textrm{d} \Omega_r}=\pm \frac{\left(2\sqrt{b(1-b)(b-k^2)}+\sqrt{1-b}(k^2-2b)\right)E(k)}{\left(\sqrt{b-k^2}-\sqrt{b}\right)\left(\left(b-1-\sqrt{b(b-k^2)}\right)E(k)+(1-k^2)K(k)\right)}, \eeq
\beq \label{slopes2}\frac{ \textrm{d}\Omega_i}{\textrm{d} \Omega_r}=\pm \frac{\left(2\sqrt{b(1-b)(b-k^2)}-\sqrt{1-b}(k^2-2b)\right)E(k)}{\left(\sqrt{b-k^2}+\sqrt{b}\right)\left(-\left(b-1-\sqrt{b(b-k^2)}\right)E(k)+(1-k^2)K(k)\right)}. \eeq
In the $\cn$ case ($b=k^2$) the slopes at the origin simplify to
\beq \frac{ \textrm{d}\Omega_i}{\textrm{d} \Omega_r}= \pm \frac{k E(k)}{\sqrt{1-k^2}(E(k)-K(k))}. \eeq
For the $\cn$ solutions, these slopes are always finite. This is not necessarily the case for nontrivial-phase solutions. Specifically, while the slopes in (\ref{slopes2}) are always finite, the slopes in (\ref{slopes1}) can be infinite if
\beq \label{topboundary} \left(b-1-\sqrt{b(b-k^2)}\right)E(k)+(1-k^2)K(k)=0. \eeq
Spectra corresponding to solutions for which this condition is satisfied are shown in in Figure \ref{ntpspecialspectrum}(b). The condition corresponds to the splitting between the two butterfly regions, as well as the upper splitting between the triple-figure 8 and the figure 8s regions. See Figure \ref{allcases}.
Further application of the chain rule can yield expressions for derivatives around the origin of any order, and the same technique can be applied around the top of the figure 8s. In doing this we can obtain Taylor series approximations of $\sigma_{\mathcal{L}}$ to any order.

Finally, an expression is obtained for the lower boundary of the triple-figure 8s and figure 8s regions. A representative example of this case is seen in Figure \ref{ntpspecialspectrum}(c). The boundary between these regions occurs at the bifurcation when $\sigma_{\mathcal{L}}\cap i\mathbb{R}$ and $\overline{\sigma_{\mathcal{L}}\setminus \i \mathcal{R}}$ have a threefold intersection, see Figure~\ref{triple8seq}(b). This occurs when
\beq \label{triplefig8cond} \textrm{Rt(b,k)}=\sqrt{\frac{2E(k)-K(k)-b K(k)+k^2 K(k)}{4K(k)}}, \eeq
where $\textrm{Rt(b,k)}$ is the smallest real root of the cubic equation
\beq -c +\left(-1+3b-k^2\right) Y+4 Y^3 = 0. \eeq
This is seen directly as the left-hand side of (\ref{triplefig8cond}) gives the point when $\textrm{Re}(\Omega)=0$ intersects the real axis and the right-hand side is (\ref{gentop}), the point where $\overline{\sigma_L\setminus \mathbb{R}}$ intersects the real axis.

\begin{figure}
\centering
\begin{tabular}{ccc}
  \includegraphics[height=36mm]{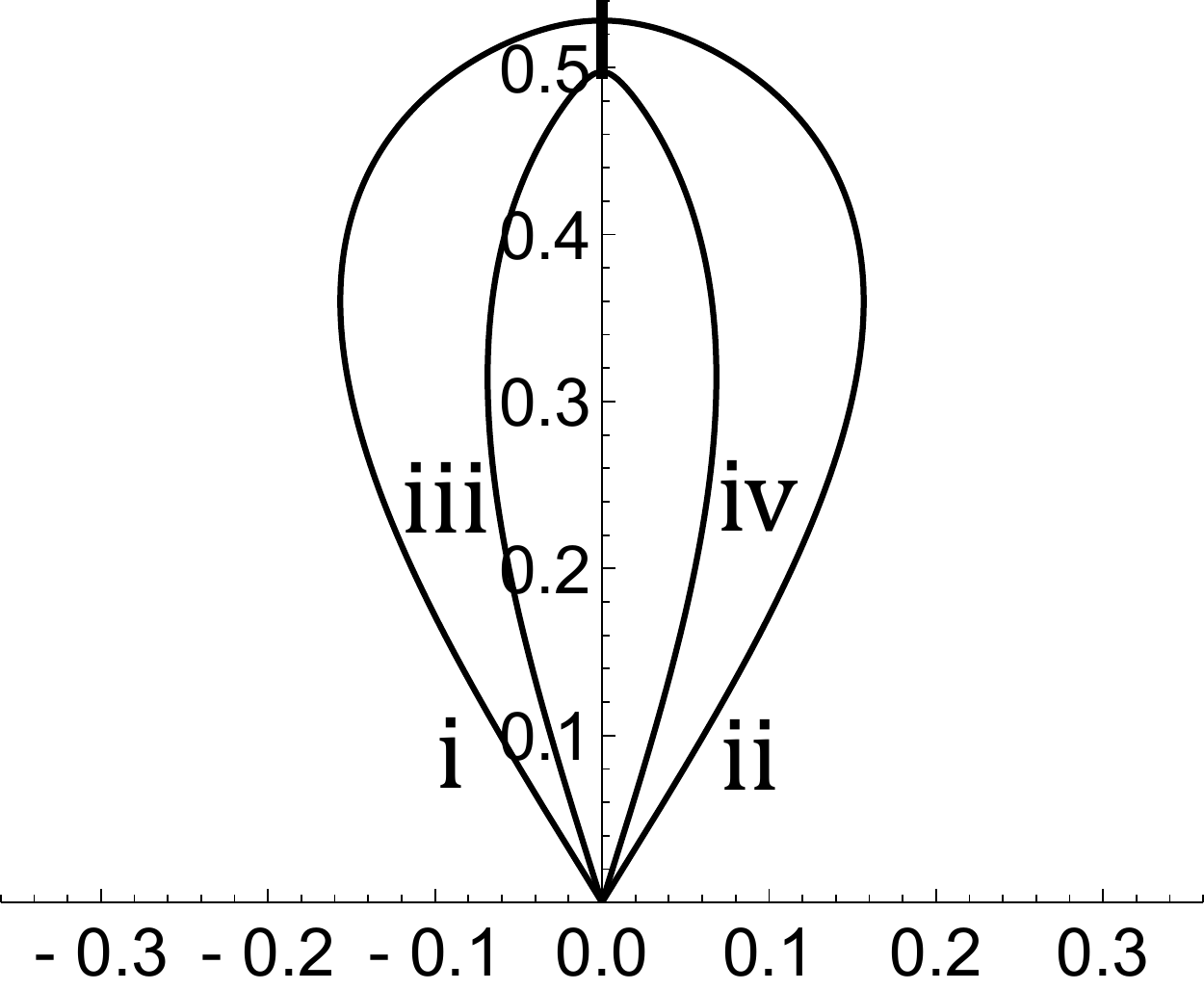} & \includegraphics[height=36mm]{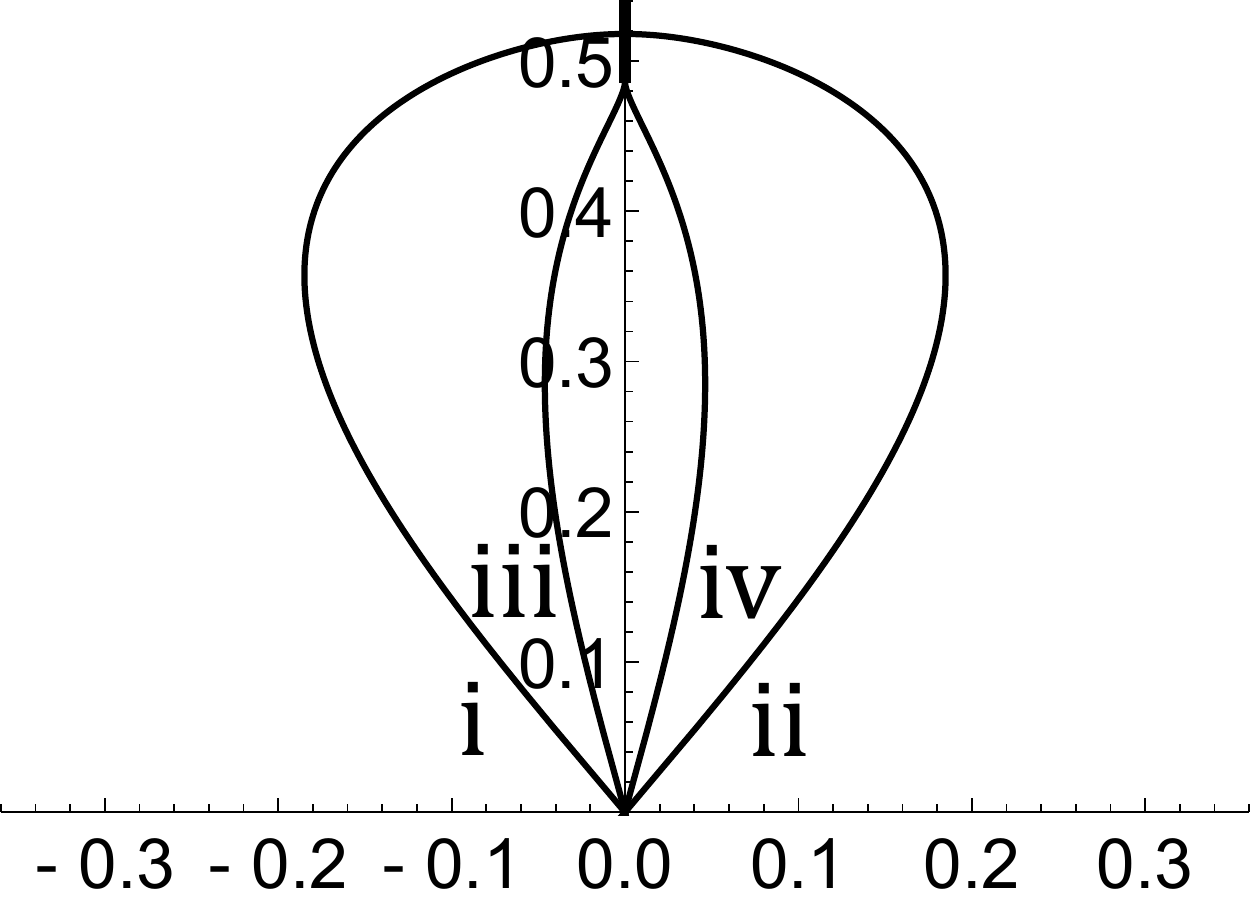} & \includegraphics[height=36mm]{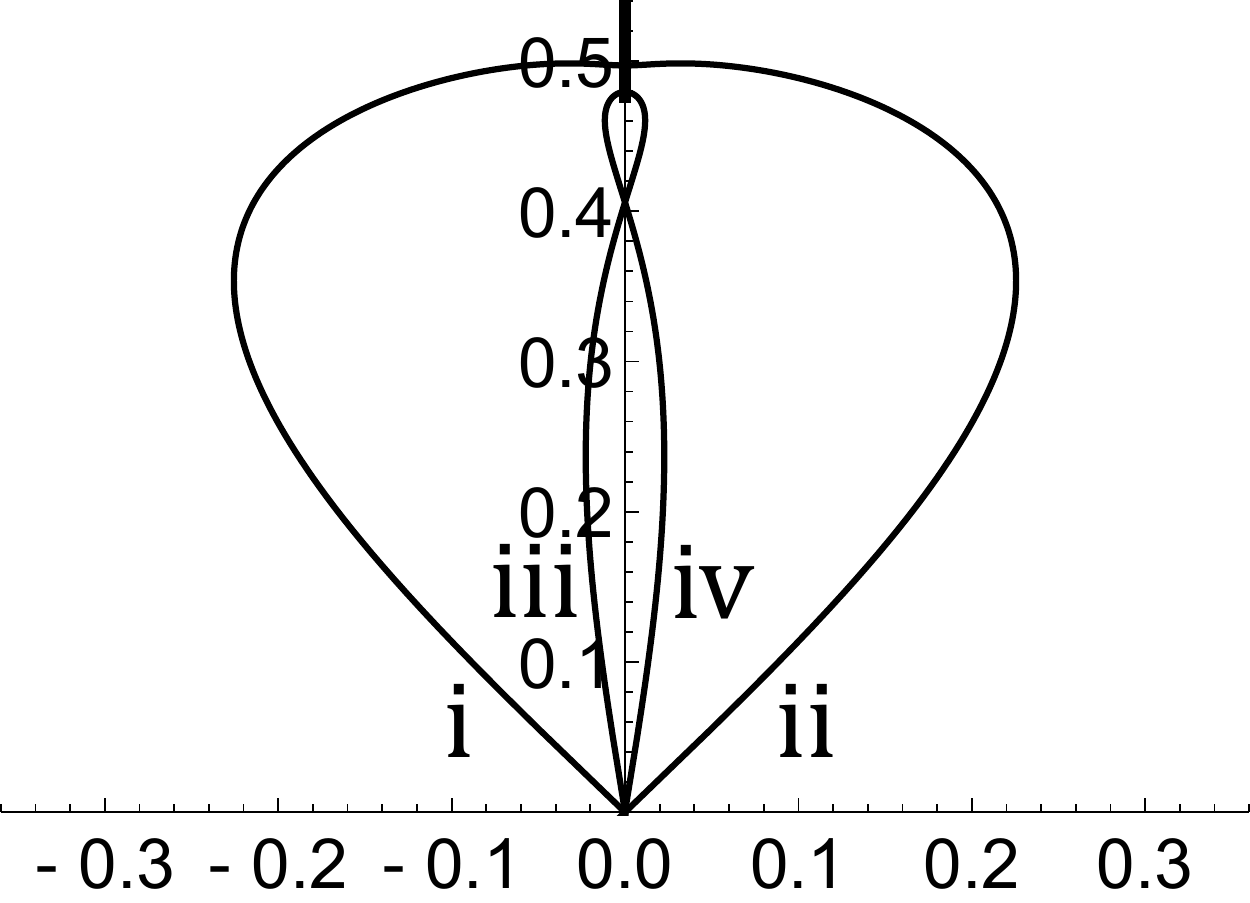} \\
(a) & (b) & (c)  \\[4pt]
  \includegraphics[height=36mm]{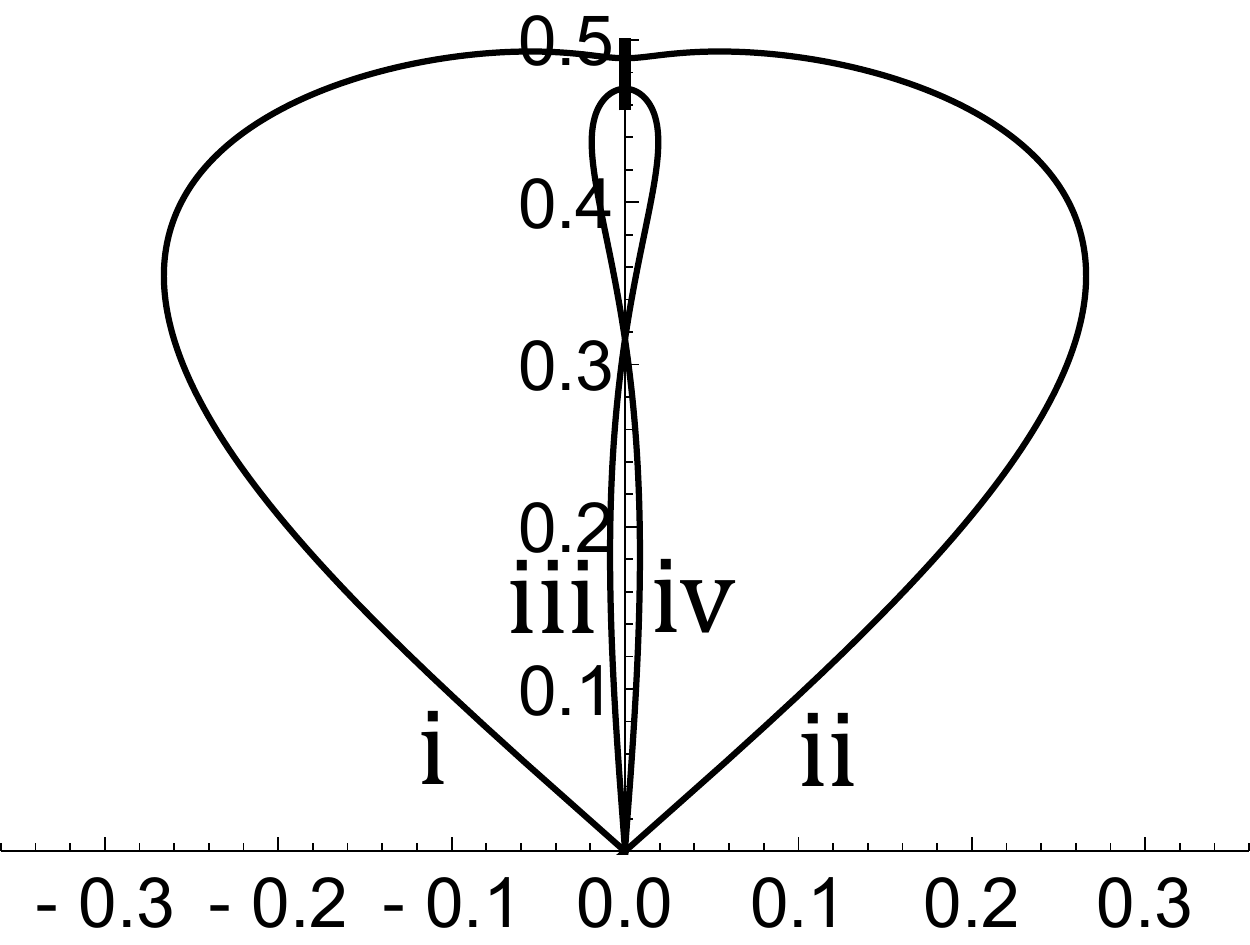} & \includegraphics[height=36mm]{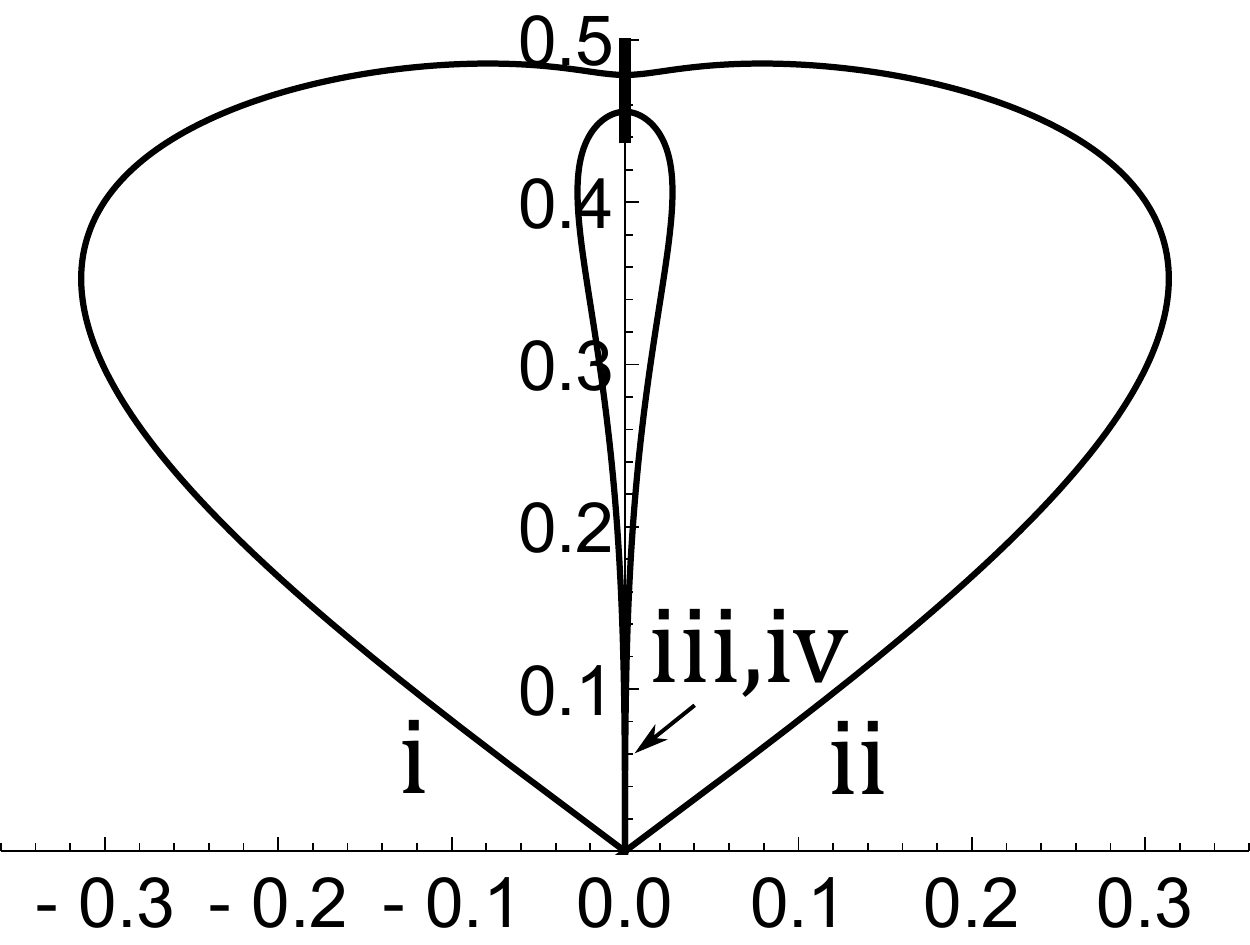} & \includegraphics[height=36mm]{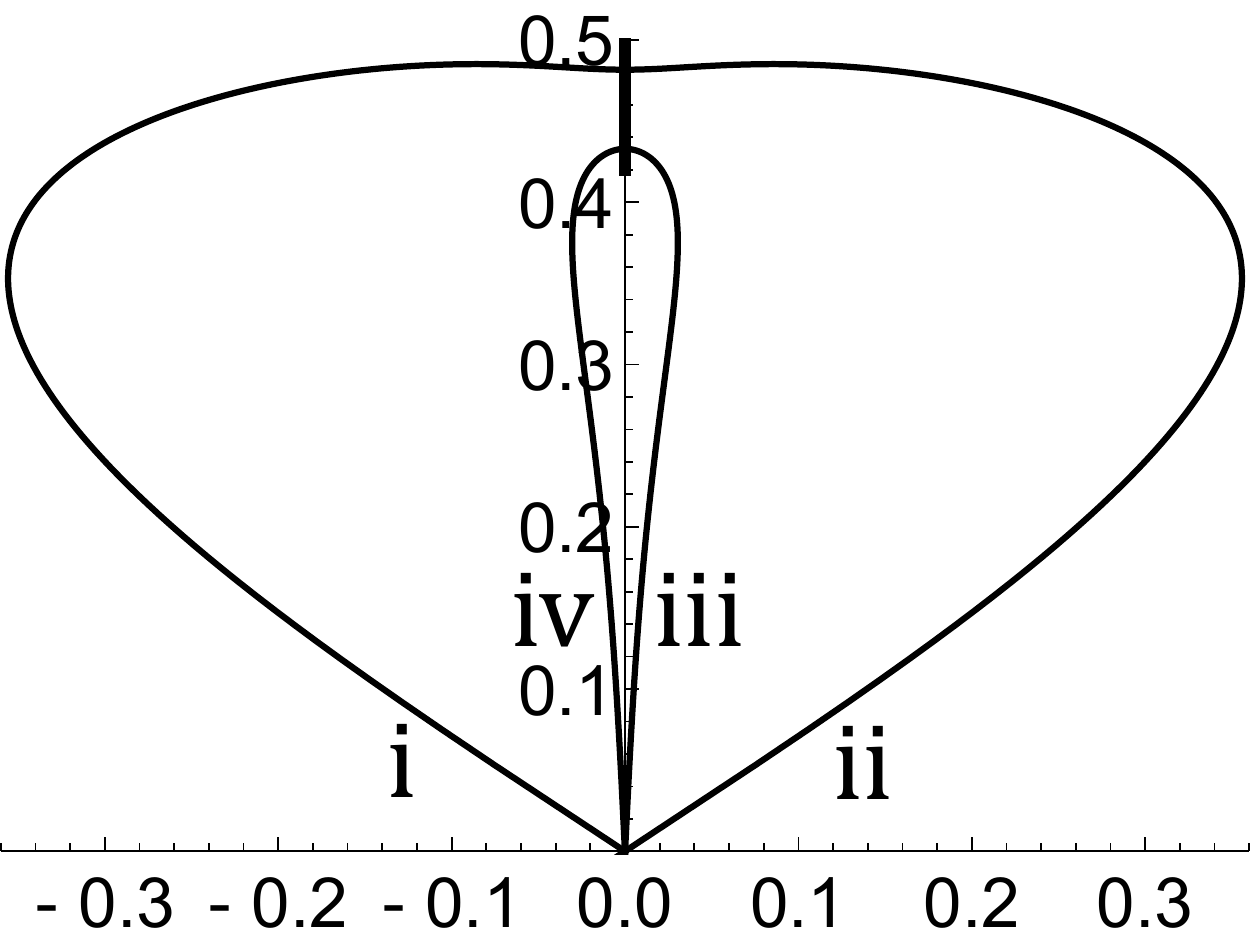} \\
(d) & (e) & (f)
\end{tabular}
\caption{(1) $\sigma_{\mathcal{L}}$ in the upper-half plane for a sequence of parameter values demonstrating the boundaries of the triple-figure 8 region.  (a) Two figure 8s, lower region, $(k,b)=(0.89,0.8);$ (b) lower boundary of triple-figure 8 region, $(k,b)=(0.895,0.819747),$ the enclosed figure 8 is not smooth at the top; (c) triple-figure 8 near lower boundary, $(k,b)=(0.895,0.84);$ (d) triple-figure 8 near upper boundary, $(k,b)=(0.887,0.85)$; (e) upper boundary of the triple-figure 8 region, $(k,b) =(0.875,0.862349);$ (f) Two figure 8s, upper region, $(k,b) = (0.86,0.87)$.}
\label{triple8seq}
\end{figure}

In Figure~\ref{triple8seq}, we plot $\sigma_{\mathcal{L}}\cap \mathbb{C}^+$. In $\mathbb{C}^+$ there are two lobes to the triple-figure 8, one near the origin and one away from the origin, see Figure~\ref{triple8seq}(c,d). For triple-figure 8s near the lower boundary of the region as in Figure~\ref{triple8seq}(c), the lobe of $\overline{\sigma_{\mathcal{L}}\setminus i \mathbb{R}}$ near the origin is larger than the lobe away from the origin. In contrast, for triple-figure 8s near the upper boundary of the region, see Figure~\ref{triple8seq}(d), the lobe of $\overline{\sigma_{\mathcal{L}}\setminus i \mathbb{R}}$ away from the origin is larger.

We also mention the four curves of $\overline{\sigma_{\mathcal{L}}\setminus i\mathbb{R}}$ near the origin which we label i, ii, iii and iv in Figure~\ref{triple8seq}. These curves give a distinguishing feature between regions I and II in Figure~\ref{allcases} both with two figure 8s. Specifically, the curves iii and iv of $\overline{\sigma_{\mathcal{L}}\setminus i\mathbb{R}}$ for the enclosed figure 8 near the origin switch places. This is seen from examining the slopes of these curves in (\ref{slopes1}) and also by comparing the relative positions of curves c and d in Figure~\ref{triple8seq}(a) and Figure~\ref{triple8seq}(f).

Lastly, we mention the four-corners point seen in Figure \ref{ntpspecialspectrum}(d). This point occurs at the intersection of (\ref{splitbutterflies}) and (\ref{topboundary}), the intersection of all four nontrivial-phase regions. At this point, $\overline{\sigma_{\mathcal{L}}\setminus i \mathbb{R}}$ has vertical tangents at the origin as well as a four-way intersection point on the imaginary axis corresponding to $\zeta=0$ in $\sigma_L$.

\section{Floquet theory and subharmonic perturbations}\label{subharmonic}

We examine $\sigma_{\mathcal{L}}$ using a Floquet parameter description. We use this to prove some spectral stability results with respect to perturbations of an integer multiple of the fundamental period of the solution, \textit{i.e.}, subharmonic perturbations.

Note that the solutions to the stationary problem (\ref{stationaryeqn}) are not periodic in general, as they may have a nontrivial phase. On the other hand, (\ref{spectralprob}) is a spectral problem with periodic coefficients since it depends only on $R(x)$.

We write the eigenfunctions from (\ref{spectralprob}) using a Floquet-Bloch decomposition
\beq \label{floquetdecom}
\left( \ba{c} U(x) \\ V(x) \ea \right) = e^{i \mu x} \left( \ba{c} \hat U (x)\\\hat V (x) \ea \right),~~ \hat U\left(x+ T(k)\right)=\hat U(x),~~ \hat V\left(x+ T(k)\right)=\hat V(x).
\eeq
with $\mu \in \left[-\pi/T(k),\pi/T(k)\right)$ \cite{De7,BDN}. Here $T(k)=2K(k)$ for all solutions, except $T(k)=4K(k)$ for the $\cn$ solution. From Floquet's Theorem \cite{De7}, all bounded solutions of (\ref{spectralprob}) are of this form, and our analysis includes perturbations of an arbitrary period. Specifically, $\mu = 2m\pi/T(k)$ for $m\in \mathbb{Z}$ corresponds to perturbations of the same period $T(k)$ of our solutions, and in general,
\beq\label{muperiods} \mu = \frac{2m\pi}{P T(k)},~~m,P\in \mathbb{Z},\eeq
corresponds to perturbations of period $P T(k)$. The choice of the specific range of $\mu$ is arbitrary, as long as it is of length $2\pi/T(k)$. For added clarity in this section, we plot some figures using the larger range $\left[-2\pi/T(k),2\pi/T(k)\right),$ before modding out, reducing the $\mu$ interval to $\left[-\pi/T(k),\pi/T(k)\right).$

In the previous sections $\sigma_{\mathcal{L}}$ is parameterized in terms of $\zeta$. We wish to parameterize $\sigma_{\mathcal{L}}$ in terms of $\mu$. We examine the $U$ eigenfunction from (\ref{floquetdecom}). From the periodicity of $\hat U$ we have
\beq e^{i \mu T(k)} = \frac{ U(x+T(k))}{U(x)}.\eeq
Using (\ref{eigenfunctions}), (\ref{varphidef}), and (\ref{gammacon}), we have
\beq e^{i \mu T(k)} =  \exp \left(-2\int_0^{T(k)} \frac{(A(x)-\Omega) \phi+B_x(x)+i \zeta B(x)}{B(x)}\textrm{d}x \right)\exp \left(i \theta(T(k))\right), \eeq
where we have used the periodicity properties
\beq A\left(x+T(k)\right)=A(x),~~B\left(x+T(k)\right)=B(x)e^{i\theta(T(k))},~~\theta\left(x+T(k)\right)=\theta(x)+\theta\left(T(k)\right).\eeq
Using (\ref{intcond4}),

\beq \label{muparametric} \mu(\zeta) = \frac{2 i I(\zeta)}{T(k)}+\frac{\theta\left(T(k)\right)}{T(k)}+\frac{2\pi n}{T(k)}, \eeq

where $I(\zeta)$ is given in (\ref{Icond}), $n\in \mathbb{Z}$, and
\beq \theta(T(k)) =
\begin{cases}
\int_0^{T(k)} \frac{\sqrt{b(1-b)(b-k^2)}}{b-k^2 \sn^2(y,k)}\textrm{d}y, &\mbox{if } b>k^2, \\
\pi, & \mbox{if }b=k^2,
\end{cases}\eeq
from (\ref{thetaeqn}).
Equation (\ref{muparametric}) relates the two spectral parameters $\zeta$ and $\mu$.

\begin{figure}
\centering
\begin{tabular}{cccc}
  \includegraphics[width=36mm]{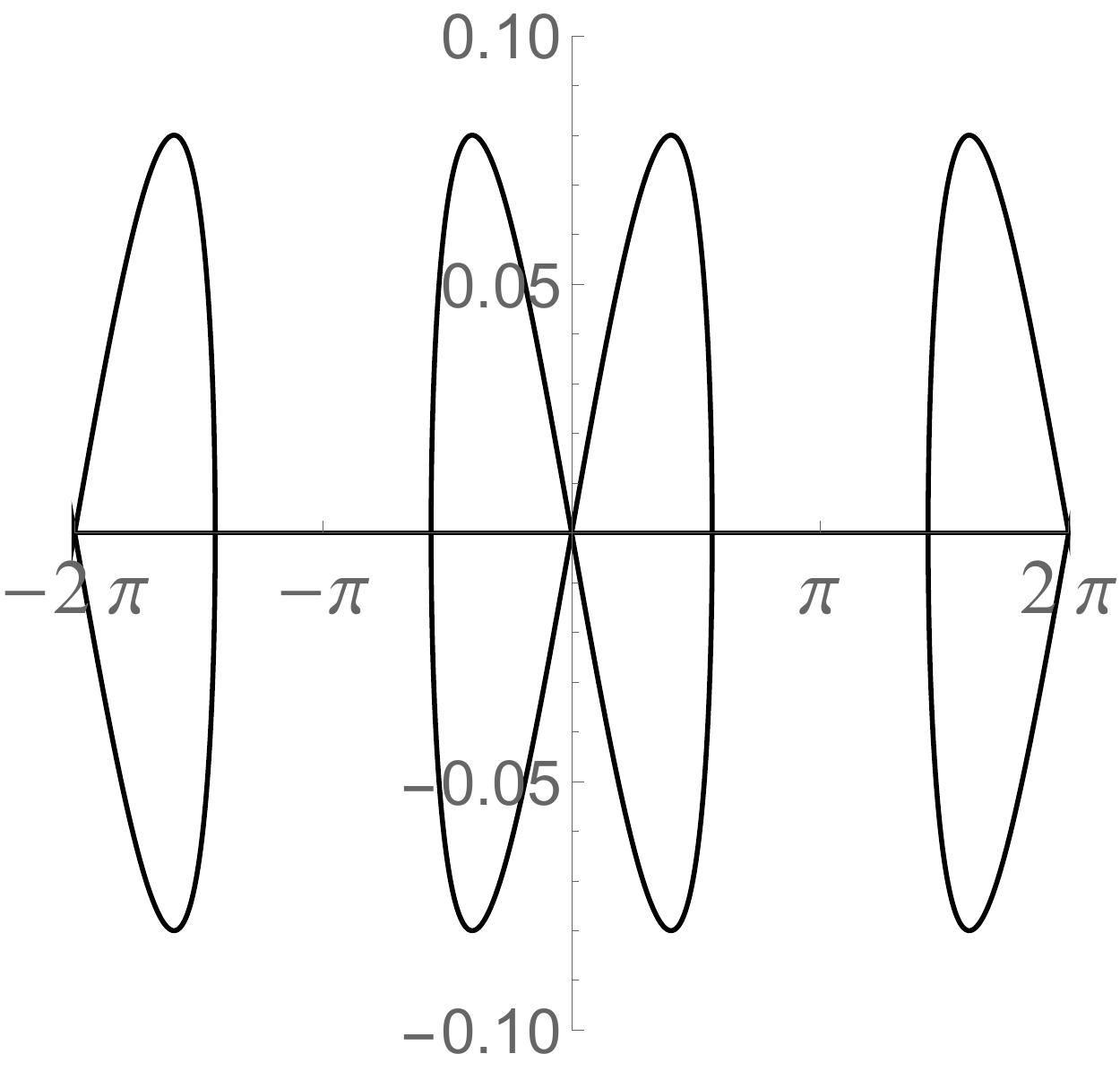} & \includegraphics[width=36mm]{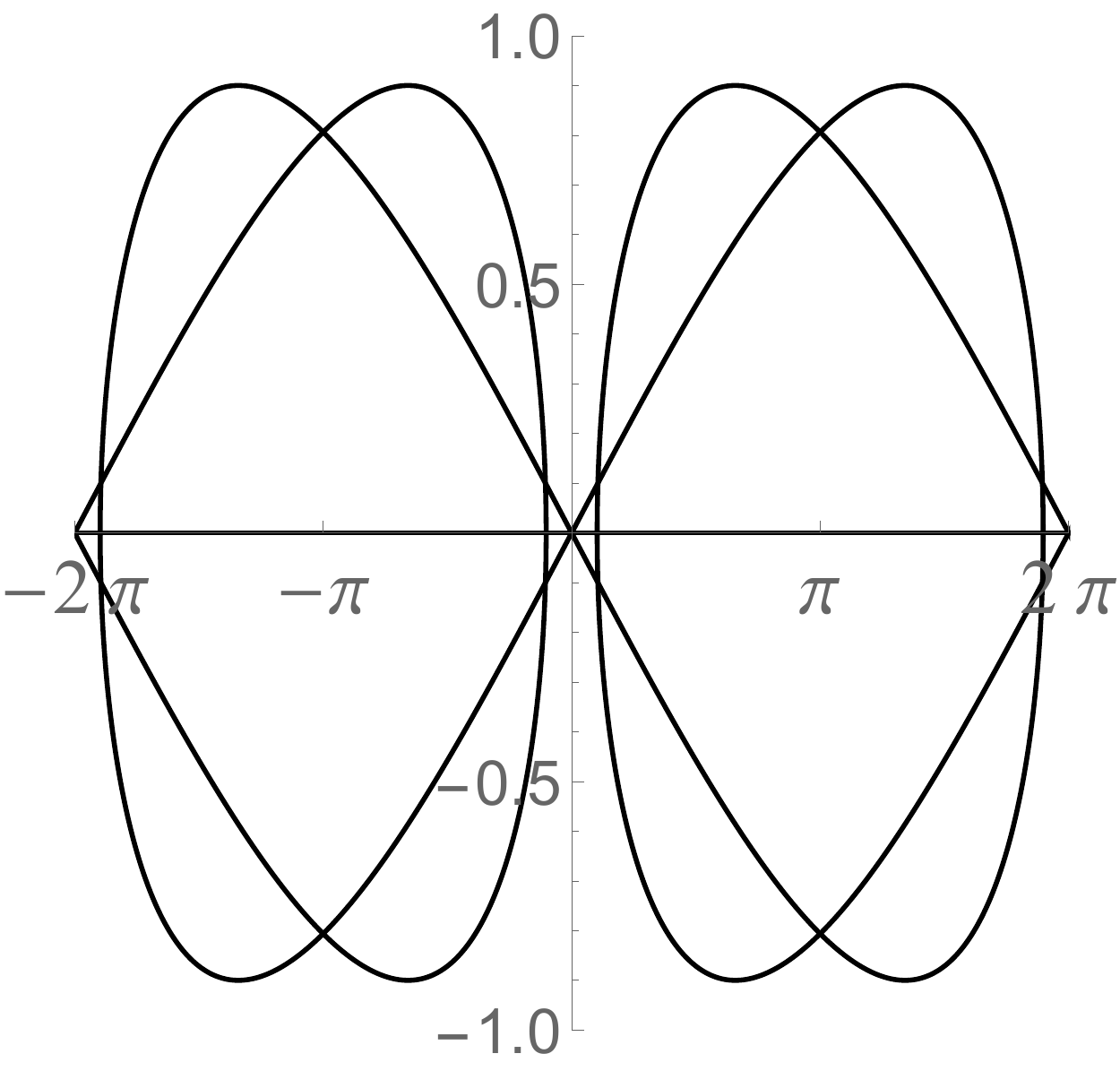} & \includegraphics[width=36mm]{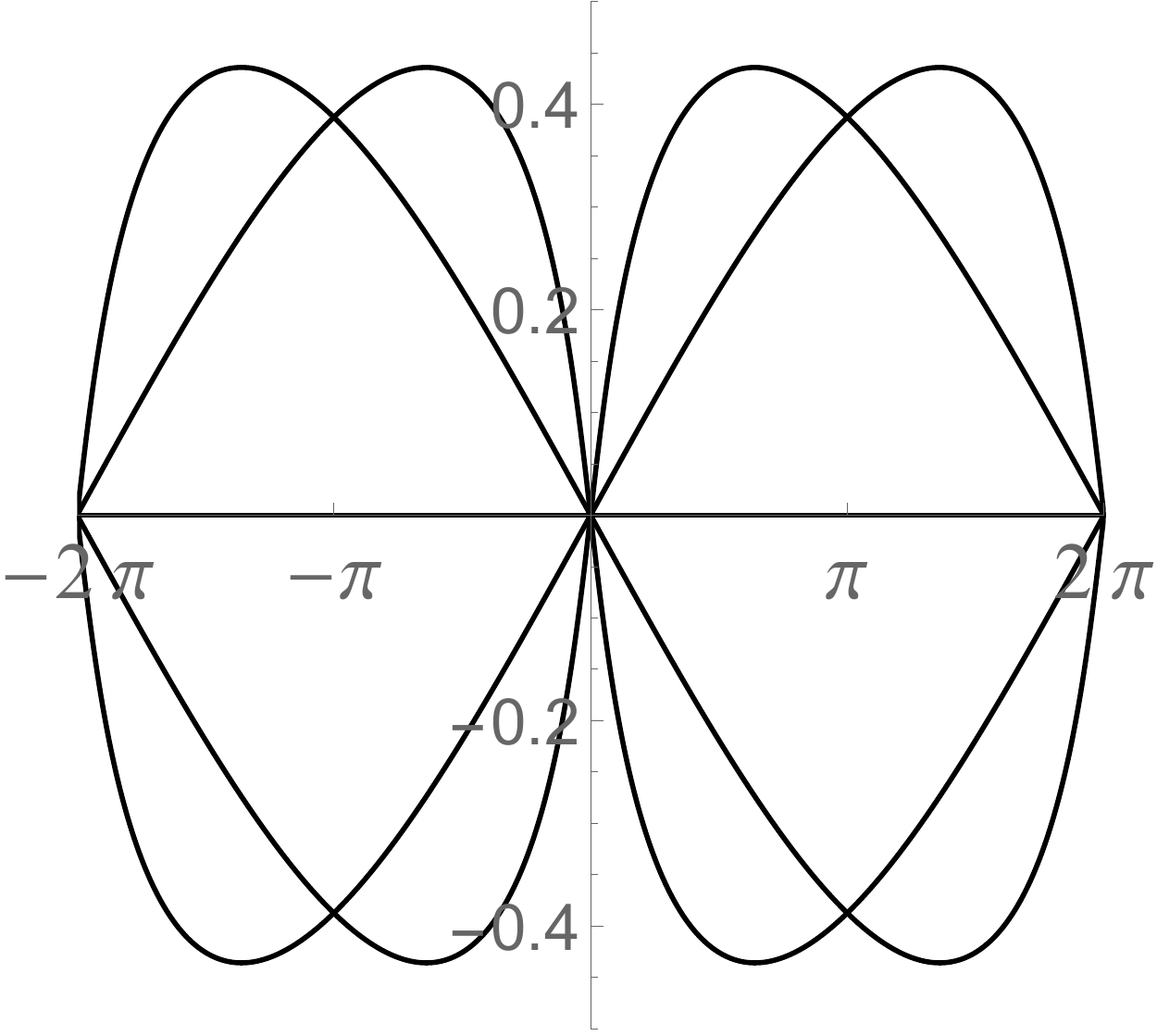} & \includegraphics[width=36mm]{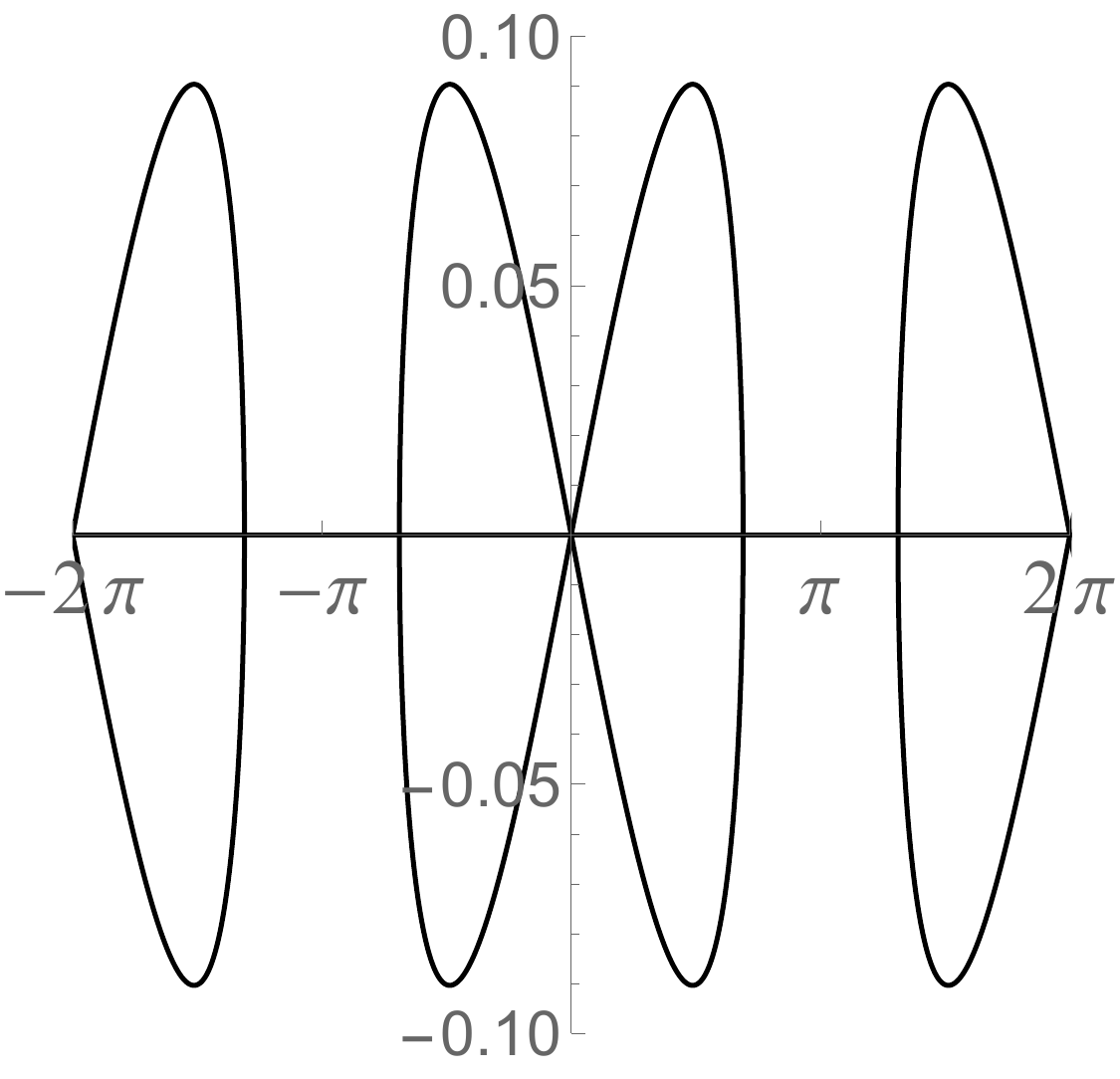} \\
(a) & (b) & (c) & (d)  \\[4pt]
\includegraphics[width=36mm]{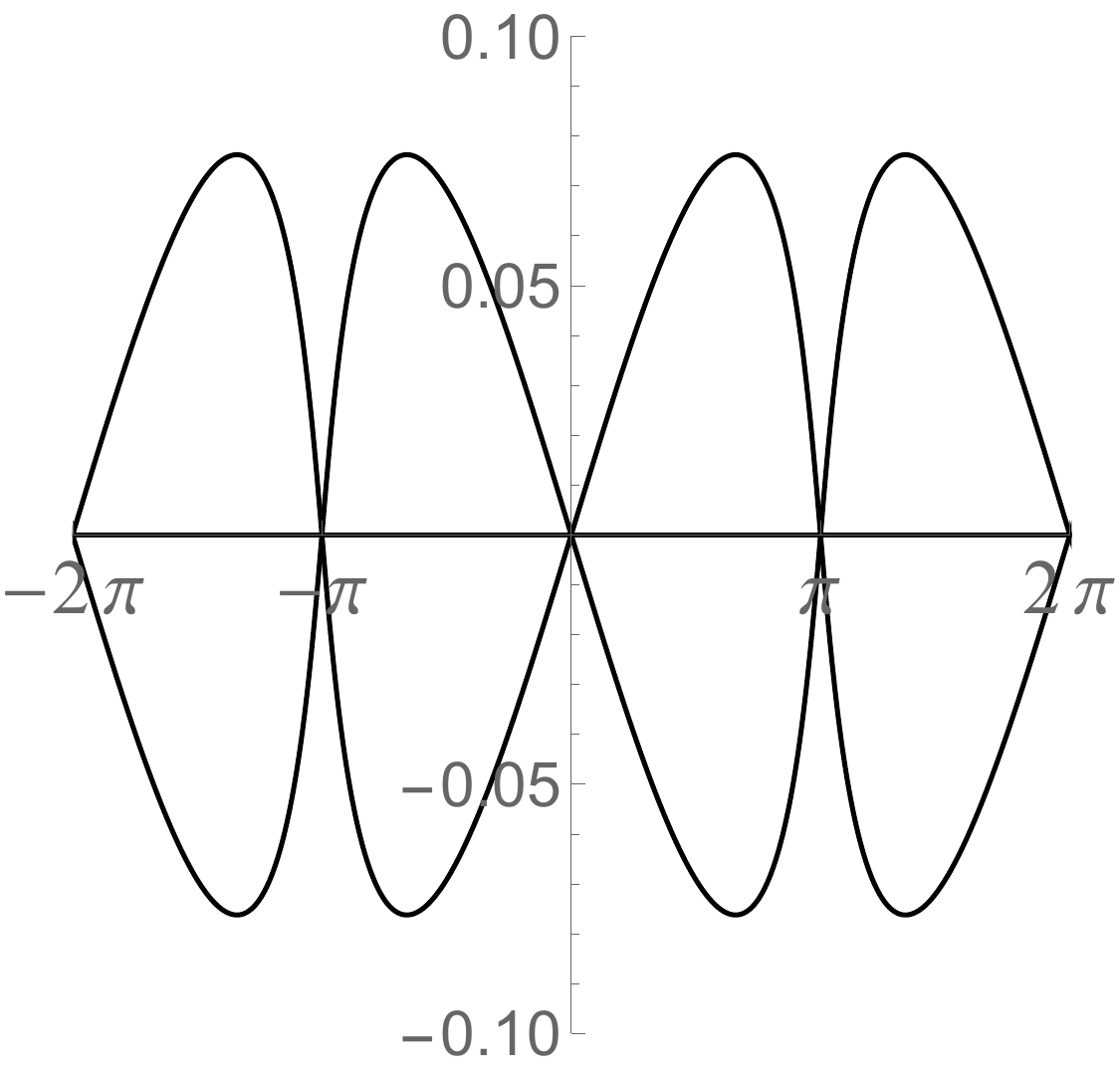} & \includegraphics[width=36mm]{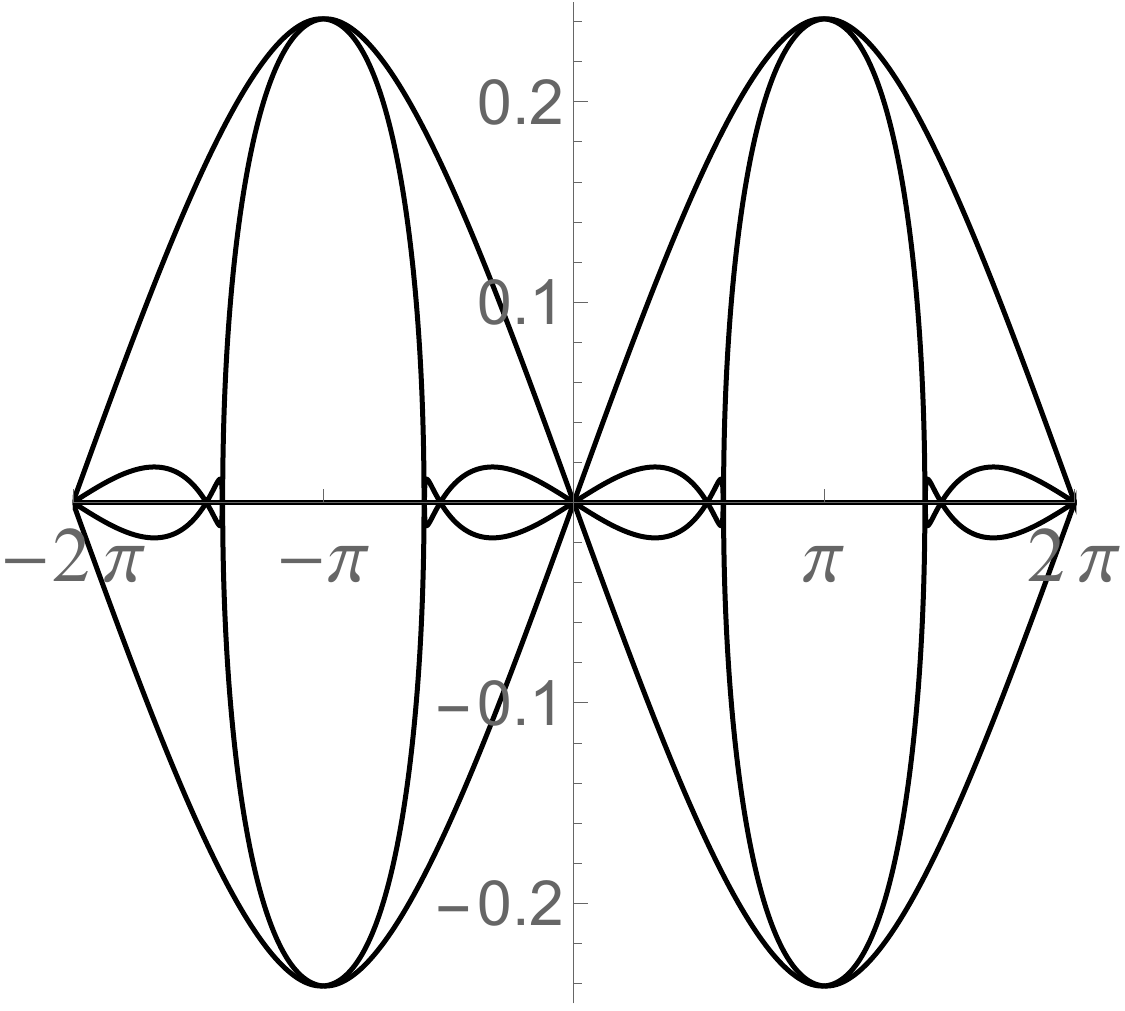} & \includegraphics[width=36mm]{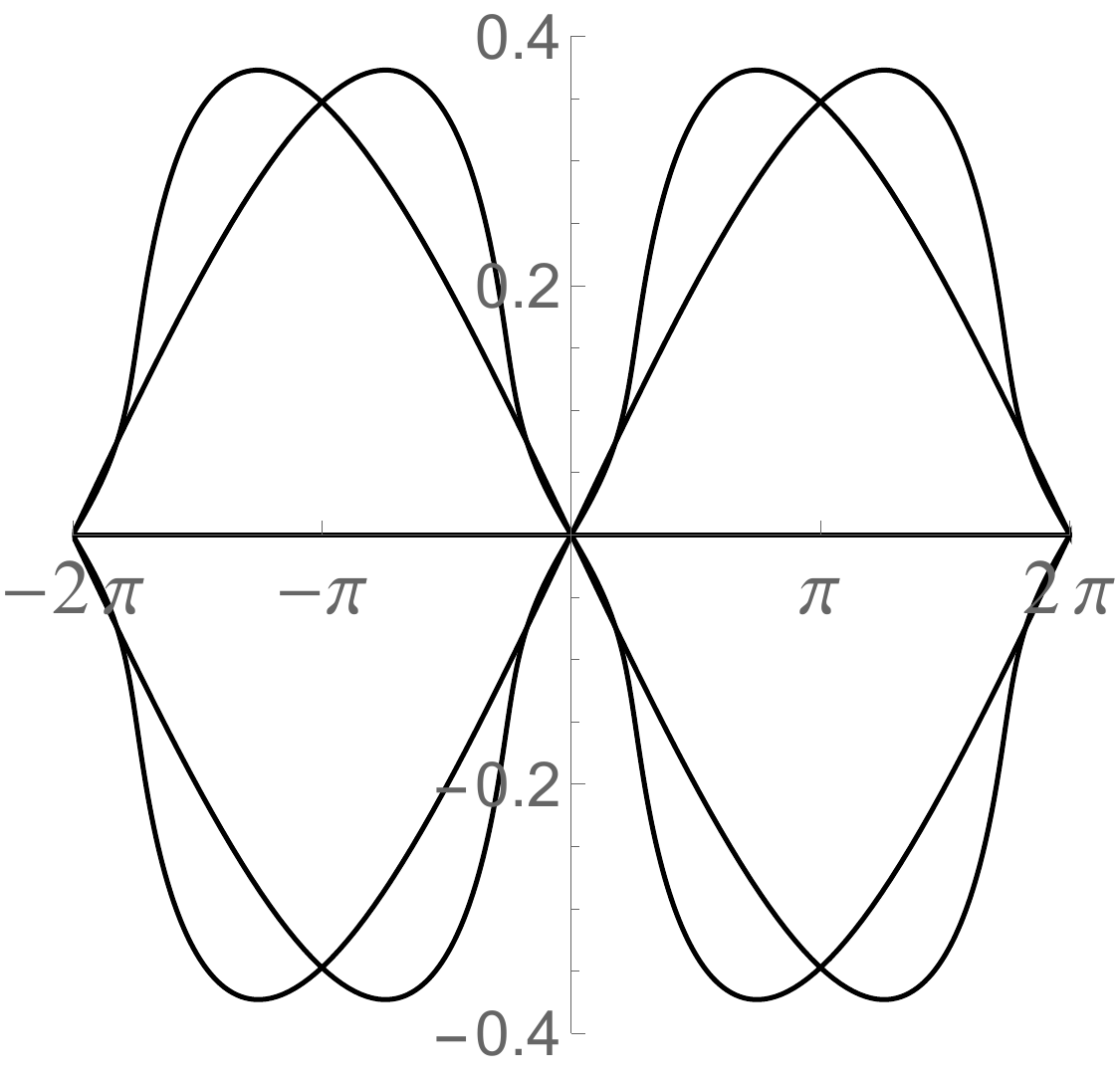} & \includegraphics[width=36mm]{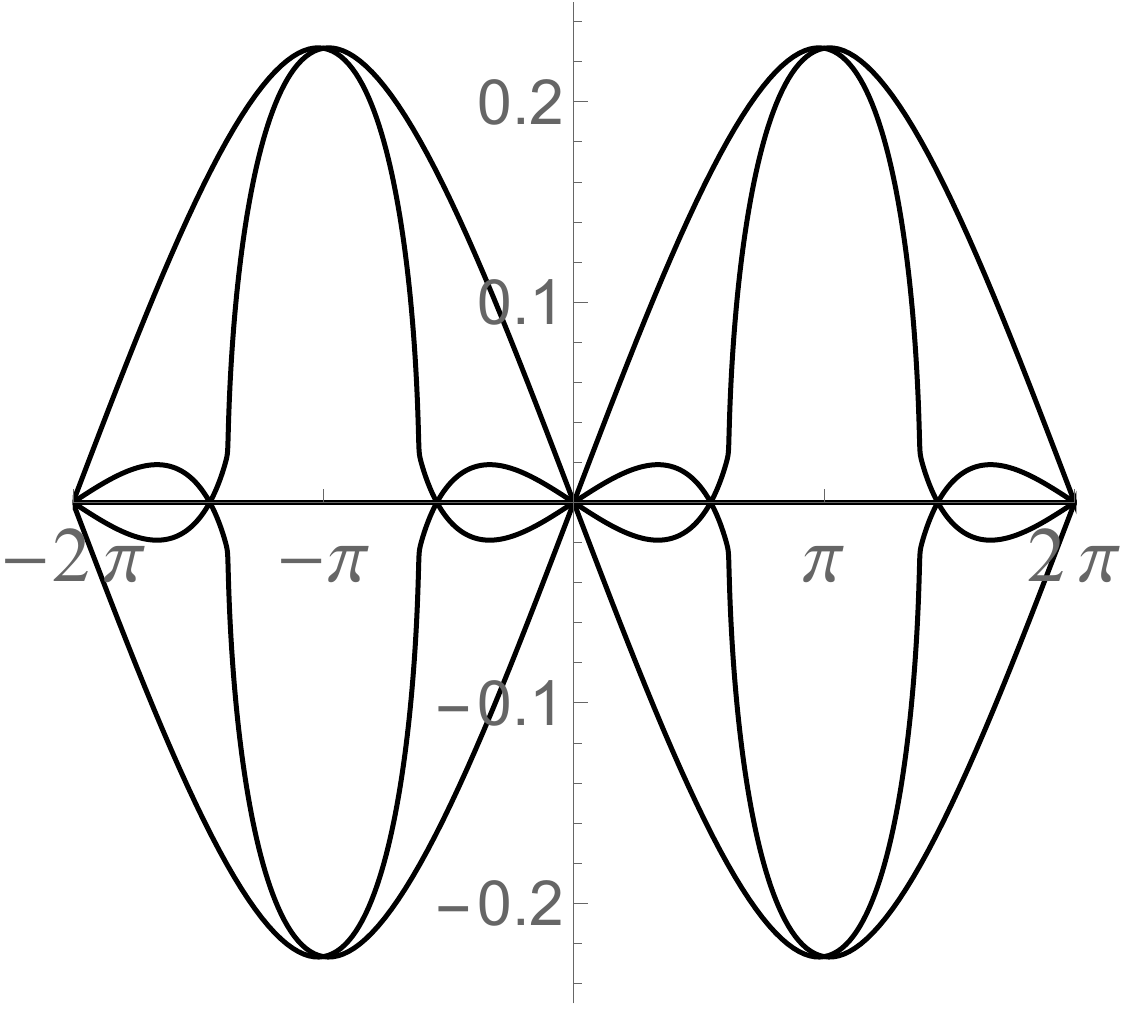} \\
(e) & (f) & (g) & (h)
\end{tabular}
\caption{The real part of the spectrum $\textrm{Re}(\lambda)$ (vertical axis) as a function of $\mu T(k)$ (horizontal axis). $T(k) \mu=2m\pi/P$ for integers $m$ and $P$ corresponds to perturbations of period $P$ times the period of the underlying solution. (a) Stokes wave solution, $(k,b)=(0,0.08);$ (b) Stokes wave solution, $(k,b)=(0,0.9);$ (c) $\dn$ solution, $(k,b)=(0.9,1);$ (d) $\cn$ solution, $(k,b)=(0.65,0.4225);$ (e) $\cn$ solution, $(k,b)=(0.95,0.9025);$ (f) triple-figure 8 solution, $(k,b)=(0.89,0.84);$ (g) non-self-intersecting butterfly solution, $(k,b)=(0.9,0.95);$ (h) self-intersecting butterfly solution, $(k,b)=(0.9,0.85).$ }
\label{mu-figures}
\end{figure}

In what follows we discuss the stability of solutions with respect to perturbations of integer multiples of their fundamental periods, so-called subharmonic perturbations \cite{GCT}. The expression (\ref{muparametric}) gives an easy way to do this. Specifically, from (\ref{muperiods}) we know which values of $\mu$ correspond to perturbations of what type. For stability, we need all spectral elements associated with a given $\mu$ value to have zero real part. In Figure \ref{mu-figures} we plot the real part of $\sigma_{\mathcal{L}}$ as a function of $\mu T(k)$ using (\ref{lambdacond}), (\ref{OmegacondW}), and (\ref{muparametric}). We rescale $\mu$ by the fundamental period $T(k)$ for consistency in our figures. Specifically,
$$\mu T(k) = \frac{2\pi m}{P},$$
 corresponds to perturbations of $P T(k)$ for any integer $n$. In what follows, we omit $\sigma_{\mathcal{L}}\cap i \mathbb{R}$.

\subsection{Stokes wave case}

We begin with the spectrum for Stokes waves (see Figures \ref{mu-figures}(a,b)). After simplification,
\begin{align} \label{muTk} \mu T(k) & = -2 \pi \textrm{sgn}(s)\sqrt{b-s^2}+2 \pi n, \\
\label{reallambda} \textrm{Re} (\lambda) & = \pm 2 \sqrt{b s^2-s^4},
\end{align}
for $s\in [-\sqrt{b},\sqrt{b}]$ and $n\in \mathbb{Z}$. Qualitatively, for any value of $n$, the parametric plot of $\textrm{Re}\left(\sigma_{\mathcal{L}}\right)$ as a function of $ \mu T(k)$ looks like a figure 8 on its side. Specifically, The figure 8 is centered at $(2\pi n,0)$ and extends left and right to $(2\pi n\pm 2\pi \sqrt{b},0),$ with non-zero values in between, see Figures \ref{mu-figures}(a,b). This leads to the following theorem:
\vspace{2mm}
\begin{thm}\label{thm-stokes}
For any positive integer $P$, Stokes wave solutions to (\ref{fNLS}) with $b\le 1/P^2$ are stable with respect to perturbations of period $P\pi.$
\end{thm}
\begin{proof}
First, $T(k) = T(0) =\pi$.
Let $P\in\mathbb{N}_0.$ For stability with respect to perturbations of period $PT(k)$ we need that for $\mu T(k) = \frac{2\pi m}{P},$ the spectral elements $\lambda\in \sigma_{\mathcal{L}}$ have zero real part, \textit{i.e.}, for $\mu T(k) = 0, \frac{2\pi}{P},\ldots,\frac{2\pi(P-1)}{P},$ $\textrm{Re}(\lambda)=0.$ From (\ref{muTk}), $\mu T(k)=0$ only when $s = \pm \sqrt{b},$ which corresponds to $\textrm{Re}(\lambda)=0$ from (\ref{reallambda}). Thus it suffices to consider $\mu T(k) = \frac{2\pi}{P},\ldots,\frac{2\pi(P-1)}{P}.$
Qualitatively, we have figure 8s centered at $\mu T(k) = 2\pi n$ extending over $[2\pi n - 2\pi \sqrt{b},2\pi n+2\pi \sqrt{b}].$
Specifically, as $s$ ranges from $-\sqrt{b}$ to $0$, $\mu T(k)$ monotonically increases from $2\pi n$ to $2\pi (n+\sqrt{b})$. Over the same range, $|\textrm{Re}(\lambda)|$ increases from $0$ (at $s=-\sqrt{b}$) to $b$ (at $s=-\sqrt{b/2}$) then decreases back down to $0$ (at $s=0$) mapping out the right-half of the figure 8. For $s\in(0,\sqrt{b})$, the left-half of the figure 8 is produced symmetrically.
Relevant to the interval $[0,2\pi)$ are the figure 8s centered at $0$ and $2\pi$. If the right-most edge of the figure 8 centered at $\mu T(k)=0$ is less than $2\pi/P$ and the left most edge of the figure 8 centered at $\mu T(k) = 2\pi$ is greater than $2\pi (P-1)/P,$ then the real part of the spectrum at $\mu T(k) = \frac{2\pi}{P},\ldots,\frac{2\pi(P-1)}{P}$ is zero. These conditions are
\beq \label{twoconds} 2\pi \sqrt{b} \le \frac{2\pi}{P} \textrm{ and } 2\pi-2\pi \sqrt{b}\ge \frac{2\pi(P-1)}{P}.\eeq
Simplifying both conditions gives $0\le b\le 1/P^2,$ completing the proof.
\end{proof}

For more intuition about this result, one can examine Figure \ref{mu-figures}. In Figure \ref{mu-figures}(a), $b=0.08.$ Here $b<1/P^2$ for $P=1,2,3$ so this Stokes wave solution is stable with respect to perturbations of periods $\pi,2\pi,3\pi$. This is readily seen in Figure \ref{mu-figures}(a) where the figure 8 centered at the origin extends to $\mu T(k) = \pm 2\pi \sqrt{0.08}\approx 0.567 \pi,$ so $\textrm{Re}(\lambda)=0$ when $\mu T(k)=0,2\pi/3,\pi,4\pi/3,2\pi.$ In Figure \ref{mu-figures}(b), $b<1/P^2$ only for $P=1,$ so the Stokes wave solution is only stable with respect to perturbations of the fundamental period $\pi$. Indeed, the figure 8 centered at the origin extends to $\mu T(k) = \pm 2\pi \sqrt{0.9}\approx 1.90 \pi,$ so $\textrm{Re}(\lambda)=0$ only for $\mu T(k) = 0$.

In order to proceed with results for the $\dn,$ $\cn,$ and general nontrivial-phase solutions we provide the following useful lemma:
\vspace{2mm}
\begin{lemma}\label{monotonelemma}For any analytic function $f(z) = u(x,y) +i v(x,y),$ on a contour where $u(x,y)=$ constant, $v(x,y)$ is strictly monotone, provided the contour does not traverse a saddle point.
\end{lemma}
\begin{proof} This is an immediate consequence of the Cauchy-Riemann relations \cite{BORN}.
\end{proof}
Thus along contours where $\textrm{Re}(I(\zeta))=0,$ if there are no saddle points, then $\textrm{Im}(I(\zeta))$ is monotone. If we fix $b$ and $k$, using (\ref{muparametric}) we see that $\mu(\zeta)$ is also monotone along curves with $\textrm{Re}(\zeta) = 0$.

\subsection{dn case}
A representative plot of $\mu T(k)$ vs $\textrm{Re}(\lambda)$ for a $\dn$ solution is shown in Figure \ref{mu-figures}c.
 We prove the following theorem:
\vspace{2mm}
\begin{thm}
The $\dn$ solutions to (\ref{fNLS}) are stable with respect to co-periodic perturbations, but not to subharmonic perturbations.
\end{thm}
\begin{proof} It suffices to consider values of $\zeta$ in the range given by (\ref{dn-zeta}), as these are the only $\zeta$ which correspond to $\lambda$ with positive real part. We can limit our study to
$$ \zeta\in \left[\frac{1-\sqrt{1-k^2}}{2} i,\frac{1+\sqrt{1-k^2}}{2} i\right]:=\left[\zeta_b,\zeta_t\right],$$
as $\zeta$ with negative imaginary part correspond to symmetric values of $\mu T(k)$. For $\zeta = \zeta_t$, $\mu T(k) = 0,$ and $\textrm{Re}(\lambda)=0$. Similarly, for $\zeta =\zeta_b$, $\mu T(k) = 2\pi,$ and $\textrm{Re}(\lambda)=0$. From Lemma \ref{monotonelemma} we know that $\mu T(k)$ increases monotonically as $\zeta$ ranges from $\zeta_t$ to $\zeta_b,$ and since $\textrm{Re}(\lambda(\zeta))>0$ in that range we have that some $\zeta$ in the range will correspond to a $\lambda$ with positive real part. Hence, $\dn$ solutions are unstable with respect to perturbations other than their fundamental period. Additionally, since $\zeta_t$ and $\zeta_b$ are the only values of $\zeta$ corresponding to $\mu T(k) = 2\pi n$ we have that $\dn$ solutions are stable with respect to perturbations of their fundamental period.
\end{proof}

\subsection{cn case}
Note that $T(k)=4K(k)$ for $\cn$ solutions.
\vspace{2mm}
\begin{thm}\label{cn-thm}
The $\cn$ solutions with $k<k^*$ are stable with respect to perturbations of period $P T(k)$, if they satisfy the condition:
\beq
\label{cn-condition} \pi-2 i I(-\zeta_t) \le \frac{2\pi}{P},
\eeq
for
\beq \label{CNtopminus} \zeta_t = \frac{\sqrt{2 E(k)-K(k)}}{2 \sqrt{K(k)}}. \eeq
\end{thm}
\begin{proof}
We examine $\zeta\in \sigma_L$ that satisfy (\ref{intcond4}), see Figure \ref{boundaryspectrum}(2c). The figure 8 spectrum is double covered, so without loss of generality, we consider only values of $\zeta$ in the left-half plane. Specifically we consider values of $\zeta$ ranging from $\zeta_- = -\frac{\sqrt{1-k^2}}{2}-\frac{k}{2}i$ to $\zeta_+ =\frac{\sqrt{1-k^2}}{2}- \frac{k}{2}i$ passing along the level curve through $\zeta=-\zeta_t$. At $\zeta_-,$ $\mu T(k) = 0$ and $\textrm{Re}(\lambda)=0$. As $\zeta$ moves from $\zeta_-$ to $-\zeta_t,$ $\mu T(k)$ monotonically increases (Lemma \ref{monotonelemma}) until it reaches $\mu_t T(k)=\pi-2 i I(-\zeta_t)$ at $\zeta =- \zeta_t$. At $-\zeta_t,$ $\textrm{Re}(\lambda) = 0$. Note that we are only considering the lower-left quarter plane. The analysis for $\zeta$ ranging from $\zeta_+$ to $\zeta_t$ is symmetric in $\mu T(k)$.

The only values of $\zeta$ which have $\textrm{Re}(\lambda)>0$ are within the ranges $[2\pi n-\mu_t,2\pi n+\mu_t].$ As in Theorem \ref{thm-stokes}, relevant to the interval $[0,2\pi)$ are the figure 8s centered at $0$ and $2\pi$. For stability the right-most edge of the figure 8 centered at $\mu T(k) =0$ needs to be less than $2\pi/P$ and the left-most edge of the figure 8 centered at $\mu T(k) = 2\pi$ to be greater than $2\pi (P-1)/P.$ These conditions are
\beq \label{mutop-cond} \mu_t \le \frac{2\pi}{P} \textrm{ and } 2\pi - \mu_t \ge \frac{2\pi(P-1)}{P},\eeq
which are the same conditions as (\ref{cn-condition}).
\end{proof}

\begin{thm}\label{cnstar-thm}
The $\cn$ solutions with $k>k^*$ are stable with respect to perturbations of period $T(k)$ and period $2T(k)$.
\end{thm}

\begin{proof}
We examine $\zeta \in \sigma_L$ that satisfy (\ref{intcond4}), see Figure \ref{boundaryspectrum} (2d). Similar to the proof of Theorem~\ref{cn-thm} we consider $\zeta$ in the lower-left quarter plane only. The parameter $\zeta$ ranges from $\zeta_3$ to $\zeta_t$ with $\zeta_t\in i \mathbb{R}$. At $\zeta_3,$ $\mu T(k) =0,$ and $\textrm{Re}(\lambda) = 0$. As $\zeta$ moves to $\zeta_t,$ we know that $\mu T(k)$ increases monotonically (Lemma \ref{monotonelemma}) until it reaches $\zeta_t$. We do not know explicitly where on the imaginary axis $\zeta_t$ is, but it satisfies (\ref{intcond4}). For any $\zeta$ on the imaginary axis, we can compute directly $\mu T(k) = \pi$, $\textrm{Re}(\lambda) = 0$. Thus the figure 8 centered at $\mu T(k) = 0$ extends outward to $\mu T(k) = \pi$. Similarly, using symmetries, the figure 8 centered at $\mu T(k) = 2\pi$ extends backward to $\mu T(k) = \pi,$ see Figure \ref{mu-figures}(e). Both figure 8s have $\textrm{Re}(\lambda) =0$ at $\mu T(k) = 0$ and $\mu T(k) = \pi$, so we have stability with respect to perturbations of periods $2 T(k)$ and $T(k)$.
\end{proof}

\subsection{Nontrivial-phases cases}
\begin{thm}\label{figure8s-thm}
Nontrivial-phase solutions in the figure 8s region and the triple-figure 8 region are stable with respect to subharmonic perturbations of period $P T(k)$ if they satisfy the condition
\beq
\label{figure8s-condition} \theta(T(k))-2 i I(-\zeta_t) \le \frac{2\pi}{P},
\eeq
with
\beq \label{figure8sthmresult} \zeta_t = \frac{\sqrt{2 E(k)-K(k) -(b-k^2) K(k)}}{2 \sqrt{K(k)}}. \eeq
\end{thm}
\begin{proof}
We examine $\zeta \in \sigma_L$ which satisfy (\ref{intcond4}), see Figure \ref{ntpspectrum}(2a,2c). Recall that $\zeta_i$ corresponds to the root of $\Omega^2(\zeta)$ in the $i$th quadrant from (\ref{zeta-roots}). The $\zeta$ spectrum has three components which we examine separately:
\begin{enumerate}
\item \textit{$\zeta$ strictly real, corresponding to $\sigma_{\mathcal{L}}\cap i \mathbb{R}$}.

$\zeta$ strictly real corresponds to $\lambda$ strictly imaginary, so these values do not need to be examined further.

\item \textit{$\zeta$ ranging between $\zeta_3$ and $\zeta_2,$ corresponding to the outside figure 8.}

For $\zeta$ ranging between $\zeta_3$ and $\zeta_2$ we follow identical steps from the proof of Theorem \ref{cn-thm}. Taking the right-most edge of the outside figure 8 centered at $\mu T(k) = 0$ to be less than $2\pi/P$ and the left-most edge of the outside figure 8 centered at $\mu T(k) = 2\pi$ to be greater than $2\pi (P-1)/P,$ we arrive at analogous conditions to (\ref{mutop-cond}) which reduce to (\ref{figure8sthmresult}) as desired. Note that we have shown only that (\ref{figure8sthmresult}) is a necessary condition.

\item \textit{$\zeta$ ranging between $\zeta_4$ and $\zeta_1,$ corresponding to the enclosed figure 8 or the triple-figure 8.}

For $\zeta$ ranging between $\zeta_4$ and $\zeta_1$, we know from Section \ref{regions} that this corresponds to the enclosed figure 8 (or triple-figure 8). Specifically, the top of this figure 8 (or triple-figure 8) is lower than the top of the other figure 8. It suffices to show that the extent of this figure 8 (or triple-figure 8) in $\mu T(k)$ is less than that of the larger figure 8. Indeed, if the enclosed figure 8 (or triple-figure 8) extends less in $\mu T(k)$ than the larger figure 8 does, then the stability bounds above are sufficient.

It suffices to show that $-2i I(\zeta_t)<-2i I(-\zeta_t).$ Let $g(\zeta) = -2i I(\zeta).$ We know $g(\zeta)$ is a real-valued function with real coefficients for $\zeta\in \mathbb{R}.$ Furthermore, from (\ref{derivintcond}),
\beq \frac{\textrm{d} g(\zeta)}{\textrm{d}\zeta} = \frac{2 E(k)-\left(1+b-k^2+4 \zeta^2\right)K(k)}{i \Omega(\zeta)} . \eeq
The only roots of $\textrm{d} g(\zeta)/\textrm{d}\zeta$ are $\zeta = \pm \zeta_t$. By checking $\textrm{d}^2 g(\zeta)/\textrm{d}\zeta^2$ we see that $g(\zeta_t)$ is a local minimum and $g(-\zeta_t)$ is a local maximum. Since there are no other extrema, $g(-\zeta_t)>g(\zeta_t)$ and (\ref{figure8sthmresult}) is a sufficient condition.
\end{enumerate}
\end{proof}

\begin{figure}
\centering
\includegraphics[width=8cm]{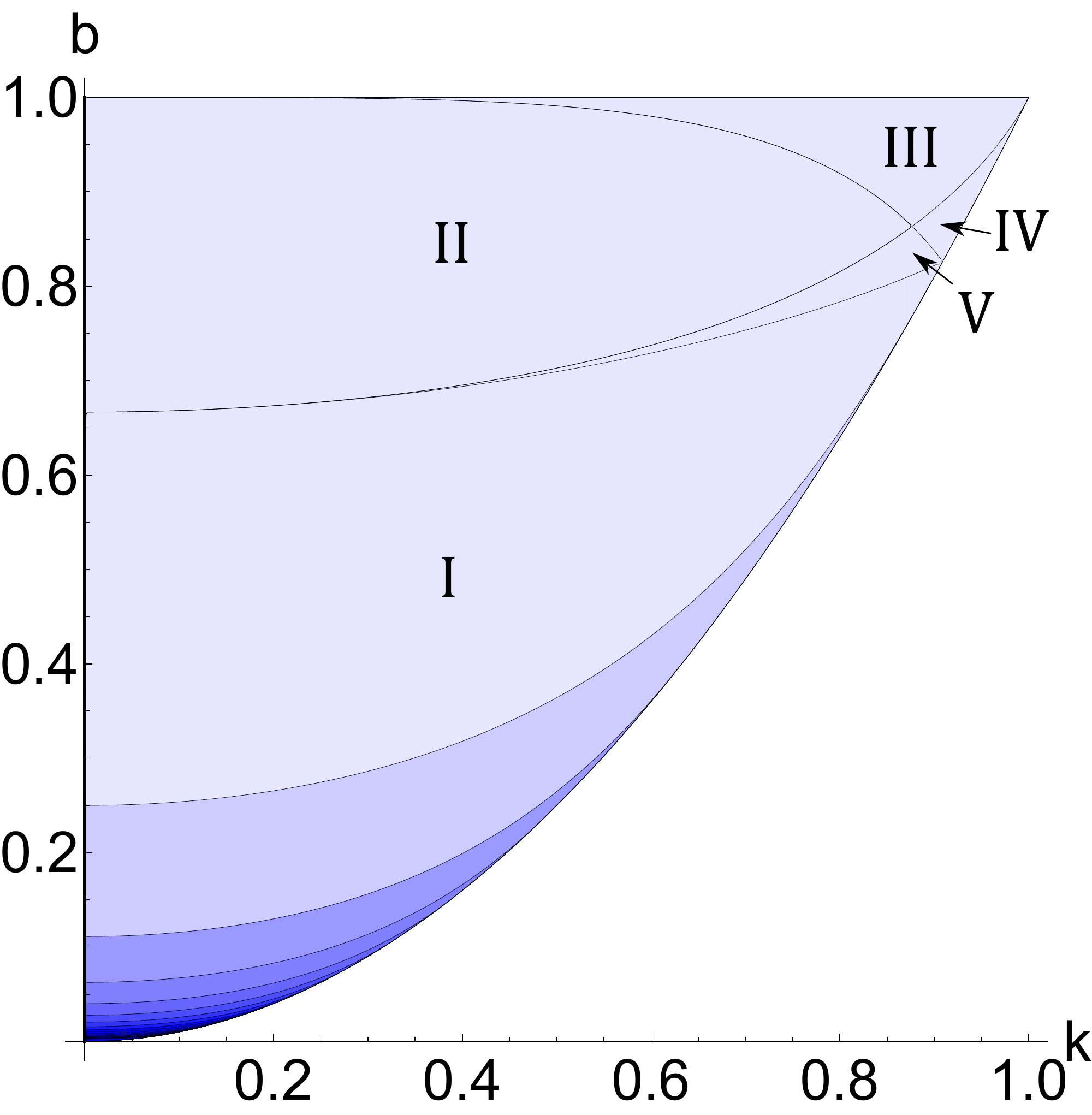}
\caption{A plot of parameter space showing the spectral stability of solutions with respect to various subharmonic perturbations. Lightest blue or darker (entire region): solutions stable with respect to perturbations of the fundamental period. Second lightest blue or darker: solutions stable with respect to perturbations of two times the fundamental period. Third lightest blue or darker: solutions stable with respect to perturbations of three times the fundamental period. Etc. }
\label{mu-allcases}
\end{figure}

\begin{thm}\label{butterfly-thm}
Nontrivial-phase solutions of butterfly type are stable with respect to perturbations of the fundamental period.
\end{thm}

\begin{proof}
We examine $\zeta \in \sigma_L$ satisfying (\ref{intcond4}), see Figure \ref{ntpspectrum}(2b,2d). The $\zeta$ spectrum consists of three components:
\begin{enumerate}
\item $\zeta$ strictly real, corresponding to $\sigma_{\mathcal{L}}\subset i \mathbb{R}$.
\item $\zeta$ ranging between $\zeta_3$ and $\zeta_4,$ corresponding to two of the butterfly wings.
\item $\zeta$ ranging between $\zeta_2$ and $\zeta_1,$ corresponding to the other two butterfly wings.
\end{enumerate}
Case 1 consists only of values of $\zeta$ corresponding to $\lambda$ with zero real part so it need not be examined. Cases 2 and 3 are symmetric in $\mu$ so it suffices to look at case 2. With $\zeta = \zeta_3,$ $\mu T(k) = 0$ with $\textrm{Re}(\lambda) = 0$. Then, from Lemma \ref{monotonelemma}, $\mu T(k)$ increases monotonically as $\zeta$ varies from $\zeta_3$ to $\zeta_4$. At $\zeta = \zeta_4$, $\mu T(k)  =2\pi$, with $\textrm{Re}(\lambda) = 0$. Because of the monotone increase in $\mu T(k)$, $\zeta_3$ and $\zeta_4$ are the only possible values of $\zeta$ which correspond to $\mu T(k) =0,2\pi$. Since $\textrm{Re}(\lambda) = 0$ for both of these values of $\zeta,$ we have stability for perturbations of period $T(k)$ as desired.
\end{proof}

The above results are summarized in Figure \ref{mu-allcases} where we plot the different regions of parameter space corresponding to spectral stability with respect to different classes of subharmonic perturbations.

\section{Approximating the greatest real part of the spectrum}\label{approximations}

In this section we find an approximation to the value of the spectral element $\sigma_{max} \in \sigma_{\mathcal{L}}$ with greatest real part. This value is significant because it corresponds to the eigenfunction with the fastest growth rate. For the Stokes wave case and for the $\dn$ solution case $\text{Re}(\sigma_{max})$ is known explicitly, so in this section we focus on approximating $\sigma_{max}$ for the $\cn$ solutions and nontrivial-phase solutions. In the Stokes wave case $\text{Re}(\sigma_{max})=b$. This is seen from maximizing the real component of (\ref{stokes-param}). For the $\dn$ solution case, from (\ref{dn-real}) we know that the spectrum extends to $\text{Re}(\sigma_{max})=\sqrt{1-k^2}$.

From (\ref{derivO}) and (\ref{derivZ}) we have an expression for the slope at any point in the set $S_\Omega$. $\sigma_{max}$ occurs when the slope at that point is infinity, \textit{i.e.}, when the denominator in (\ref{derivO}) is identically zero:
\beq \label{grp-eqn} \frac{\textrm{d} \Omega_r}{\textrm{d}\zeta_r}+\frac{\textrm{d} \Omega_r}{\textrm{d} \zeta_i} \frac{\textrm{d} \zeta_i}{\textrm{d} \zeta_r} = 0.\eeq
To simplify this equation we note that the expressions for $ \textrm{d} \Omega_r/\textrm{d}\zeta_r$ and $\textrm{d} \Omega_r\textrm{d} \zeta_i $ are found using (\ref{Omegaw}) by substituting in $\Omega = \Omega_r+i \Omega_i$ and $\zeta=\zeta_r + i \zeta_i$, taking real and imaginary parts, and differentiating with respect to $\zeta_r$ and $\zeta_i$. For the expression $ \textrm{d} \zeta_i/\textrm{d} \zeta_r $ we use (\ref{derivZ}) and the fact that
\beq \label{dzetar} \frac{\textrm{d} \textrm{Re}(I)}{\textrm{d} \zeta_r} = \textrm{Re}\left[\frac{\textrm{d}I}{\textrm{d} \zeta}\right],\eeq
\beq \label{dzetai} \frac{\textrm{d} \textrm{Re}(I)}{\textrm{d} \zeta_i} = - \textrm{Im}\left[\frac{\textrm{d}I}{\textrm{d} \zeta}\right],\eeq
from Section \ref{regions}. Using (\ref{derivintcond}), we find the real and imaginary components of $\textrm{d}I/\textrm{d}\zeta$ as
\begin{align}
\label{redI} \textrm{Re}\left[\frac{\textrm{d}I}{\textrm{d}\zeta}\right] & = \frac{2 E(k) \Omega_r+K(k) \left( -8 \zeta_i \zeta_r \Omega_i-\left(1+b-k^2+4\zeta_r^2-4 \zeta_i^2 \right) \Omega_r \right)}{2 \left(\Omega_i^2+\Omega_r^2\right)}, \\
\label{imdI} \textrm{Im}\left[\frac{\textrm{d}I}{\textrm{d}\zeta}\right] & =- \frac{2 E(k) \Omega_i-K(k) \left( -8 \zeta_i \zeta_r \Omega_r+\left(1+b-k^2+4\zeta_r^2-4\zeta_i^2 \right) \Omega_i \right)}{2 \left(\Omega_i^2+\Omega_r^2\right)},
\end{align}
Using (\ref{Omegaw}), (\ref{redI}), and (\ref{imdI}) we simplify (\ref{grp-eqn}):
\begin{align}
\begin{split}
\label{conic-eqn}
\big(-1&+3b^2+k^4+16\zeta_i^4-2b\left(-1+2k^2+8\zeta_i^2\right)\\
&+ 8c \zeta_r+16\zeta_r^2+32 \zeta_i^2\zeta_r^2+16\zeta_r^4+8k^2(\zeta_i^2-\zeta_r^2)\big)K(k) \\
&+  \left(2-6b +2k^2+8 \zeta_i^2-24 \zeta_r^2\right) E(k)=0.
\end{split}
\end{align}

This equation gives a condition on the real and imaginary parts of $\zeta.$ By construction, if (\ref{conic-eqn}) and (\ref{intcond4}) are satisfied, then $\zeta\in\sigma_L$ maps to $\sigma_{max}$. We denote such $\zeta$ as $\zeta_{max}$. We note that in the trivial-phase case, (\ref{conic-eqn}) is an equation for a conic section in the variables $\zeta_r^2$ and $\zeta_i^2$. In Figure \ref{grp-conics} we plot values of $(\zeta_r,\zeta_i)$ which satisfy (\ref{conic-eqn}) along with values of $\zeta=\zeta_r+i \zeta_i$ satisfying (\ref{intcond4}). The intersection of these curves gives $\zeta_{max}$.

\begin{figure}
\centering
\begin{tabular}{cccc}
  \includegraphics[width=36mm]{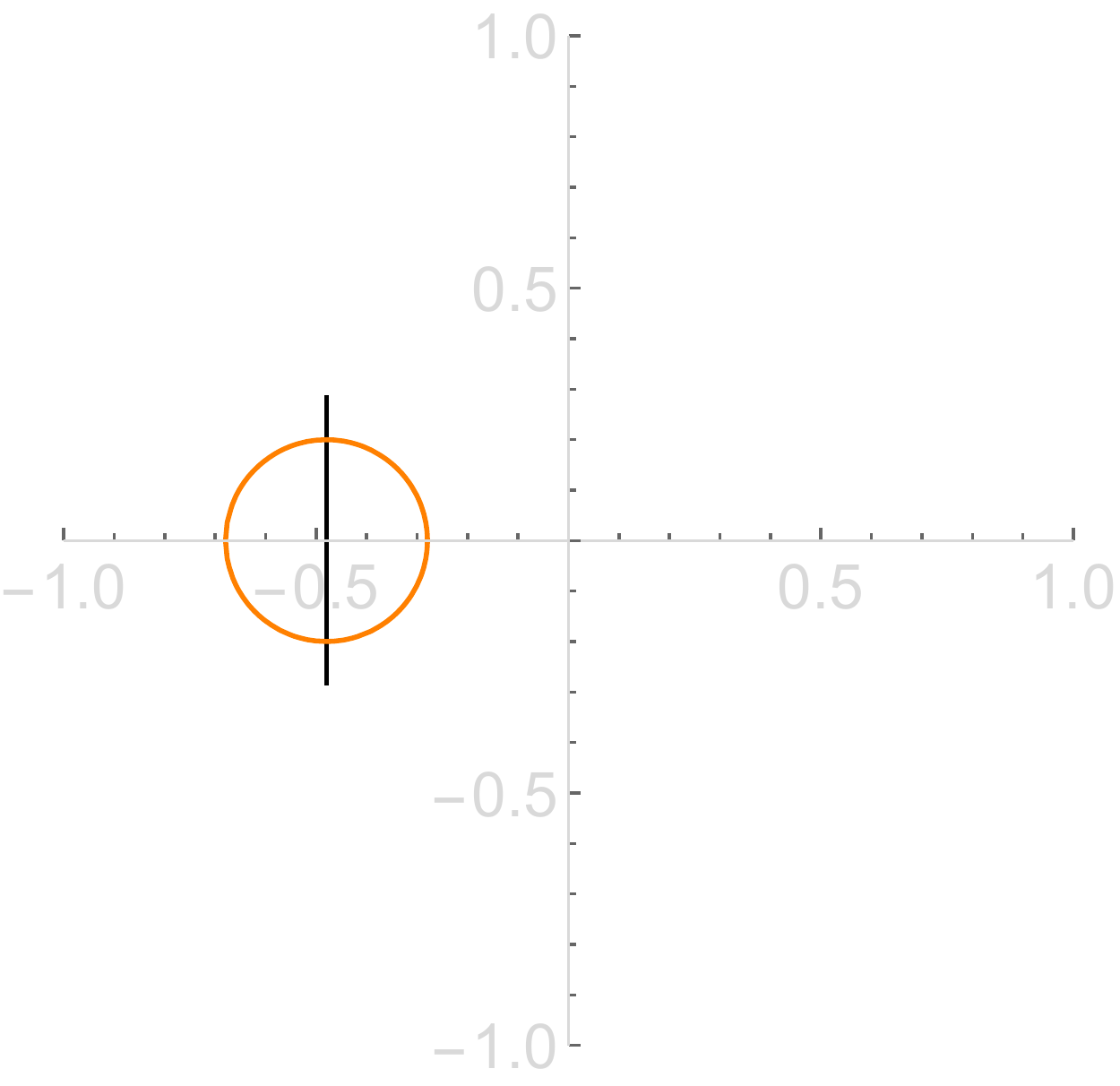} & \includegraphics[width=36mm]{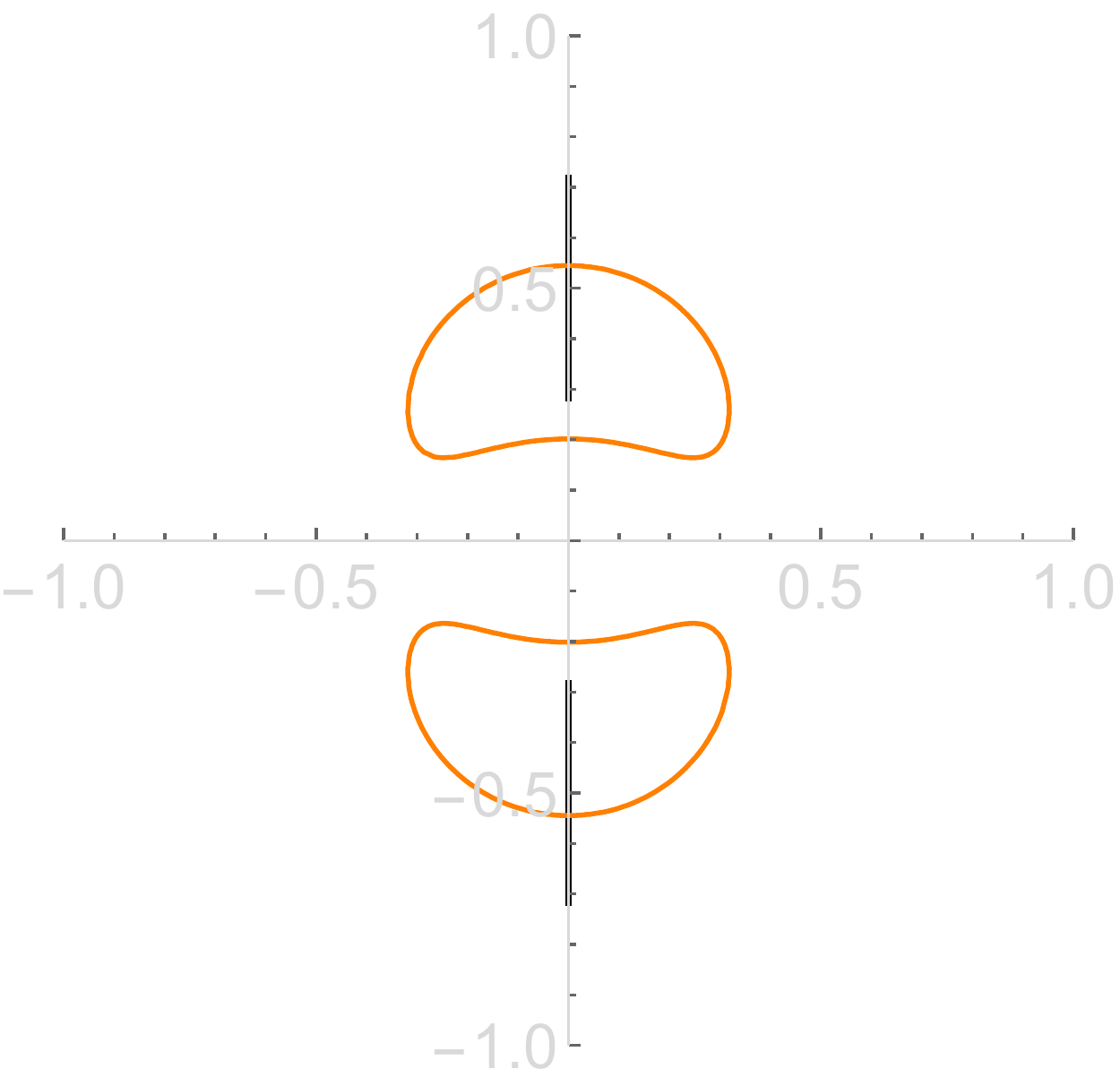} & \includegraphics[width=36mm]{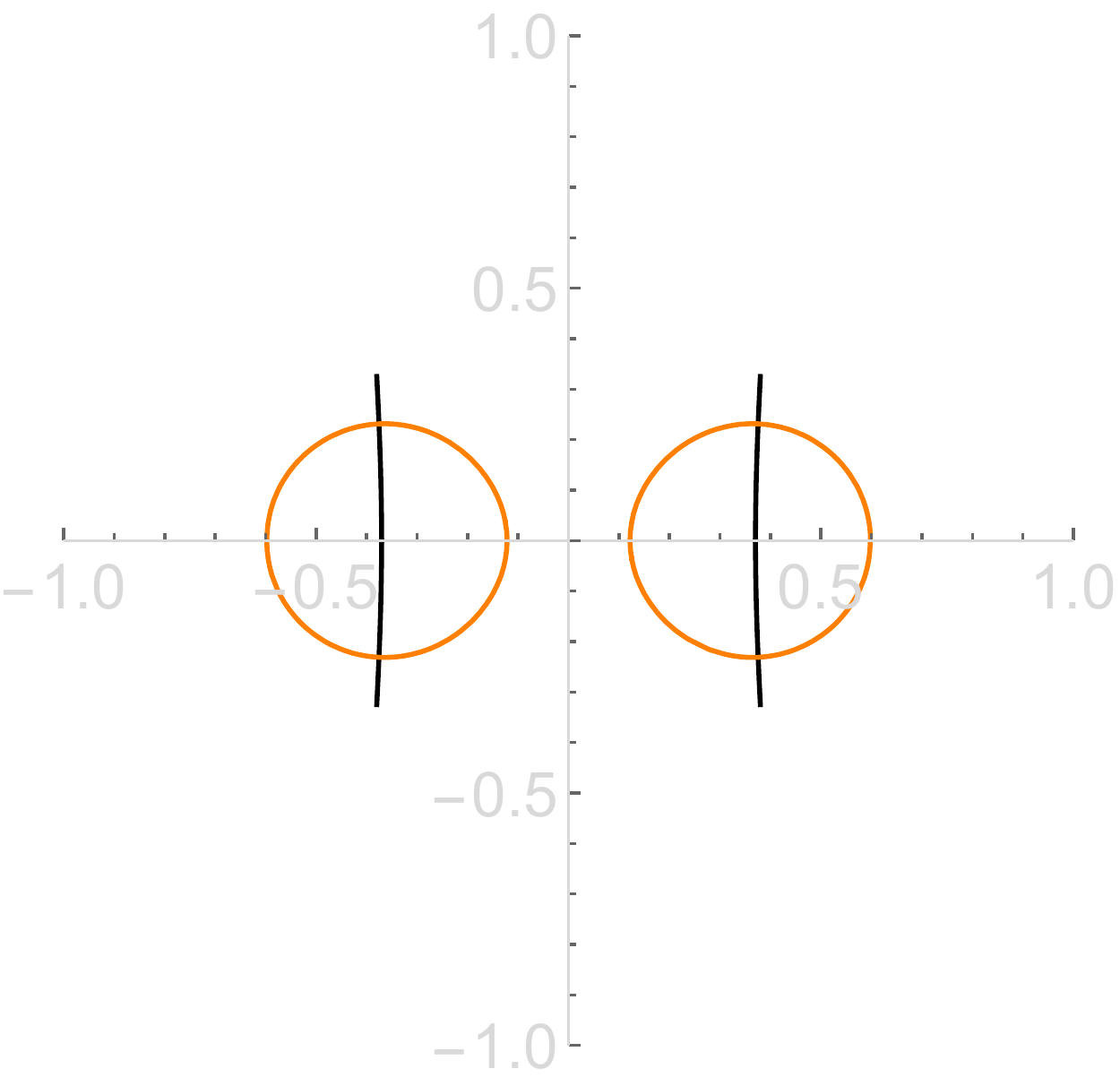} & \includegraphics[width=36mm]{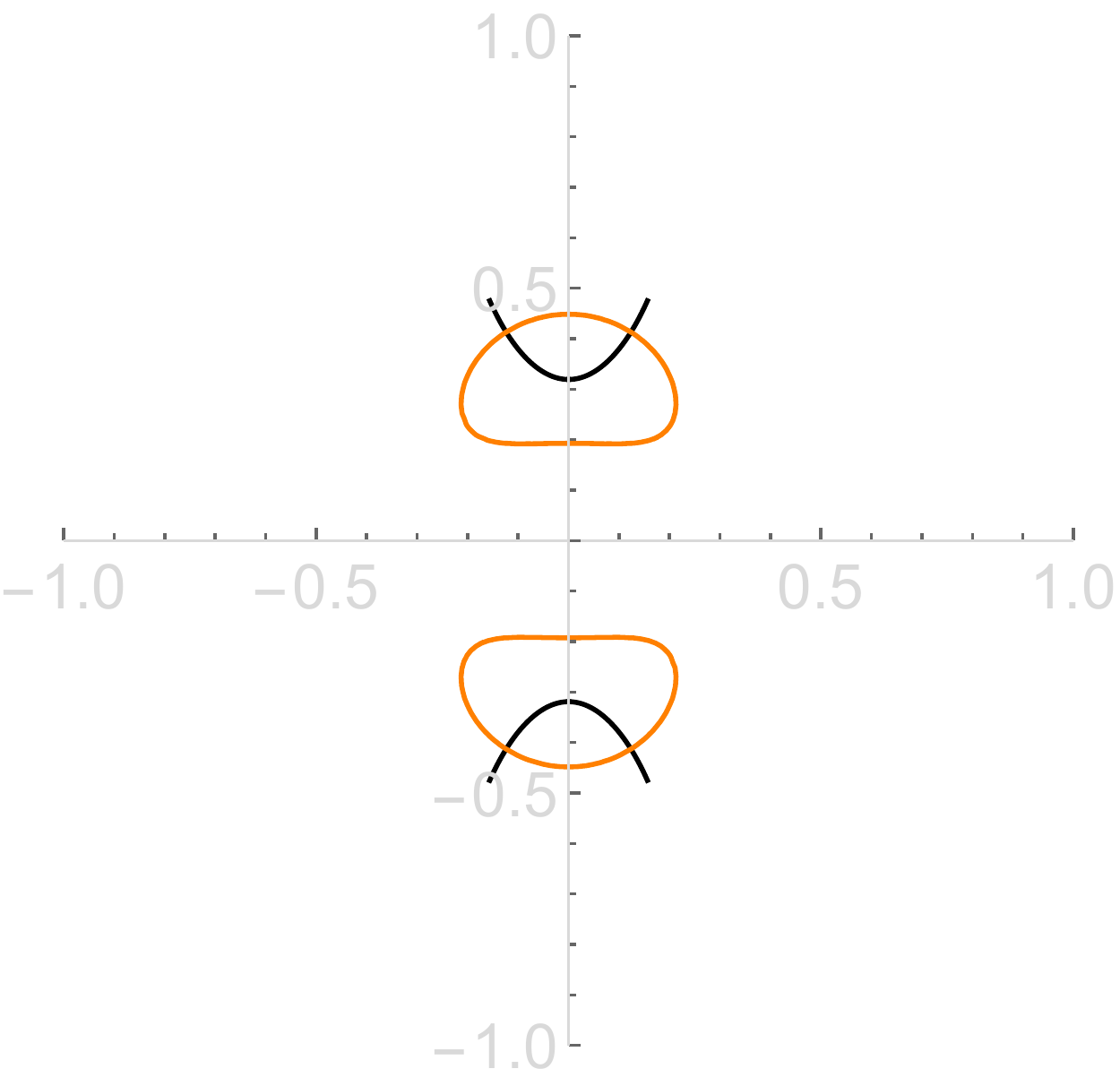} \\
(a) & (b) & (c) & (d)  \\[4pt]
  \includegraphics[width=36mm]{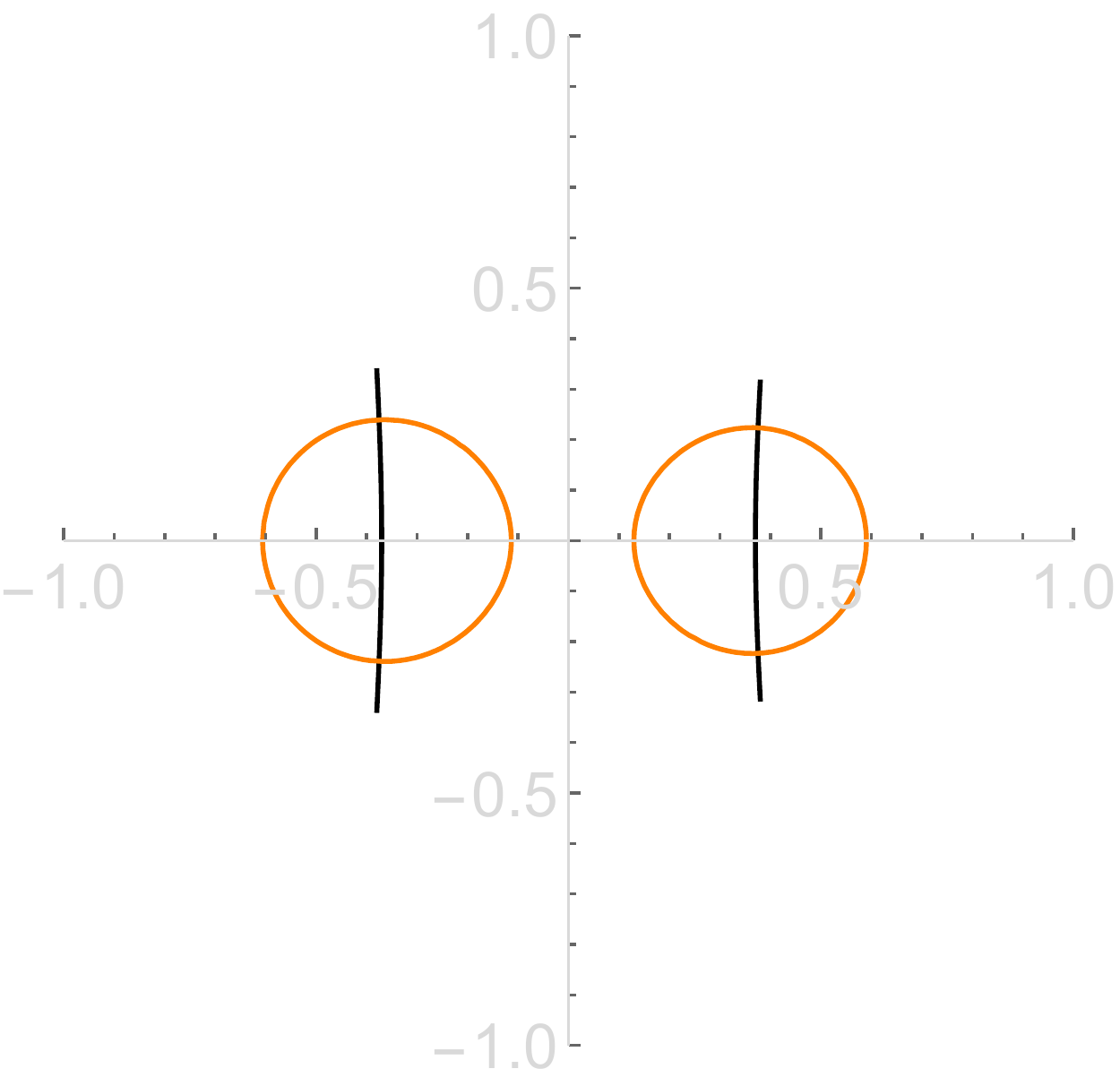} & \includegraphics[width=36mm]{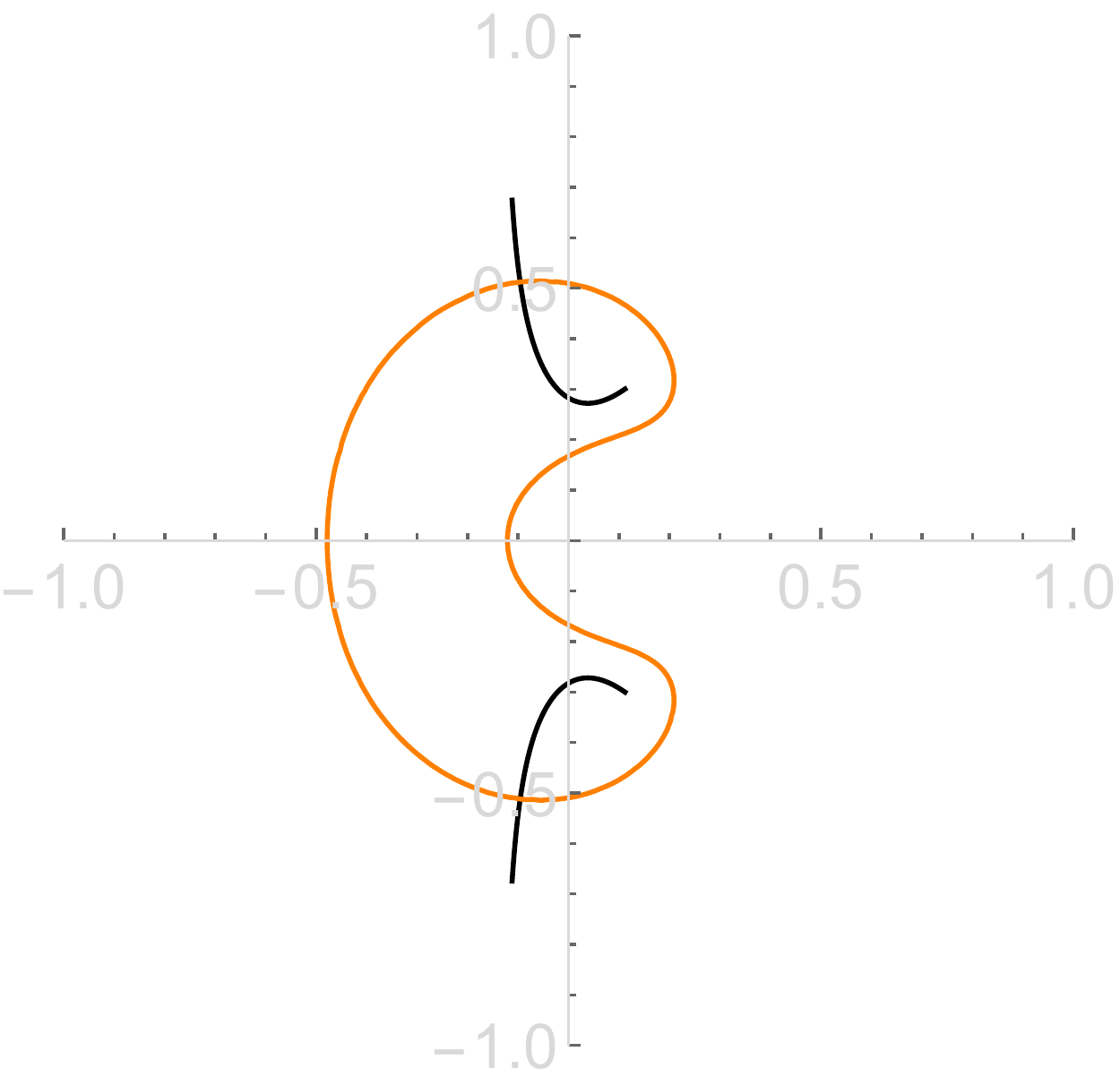} & \includegraphics[width=36mm]{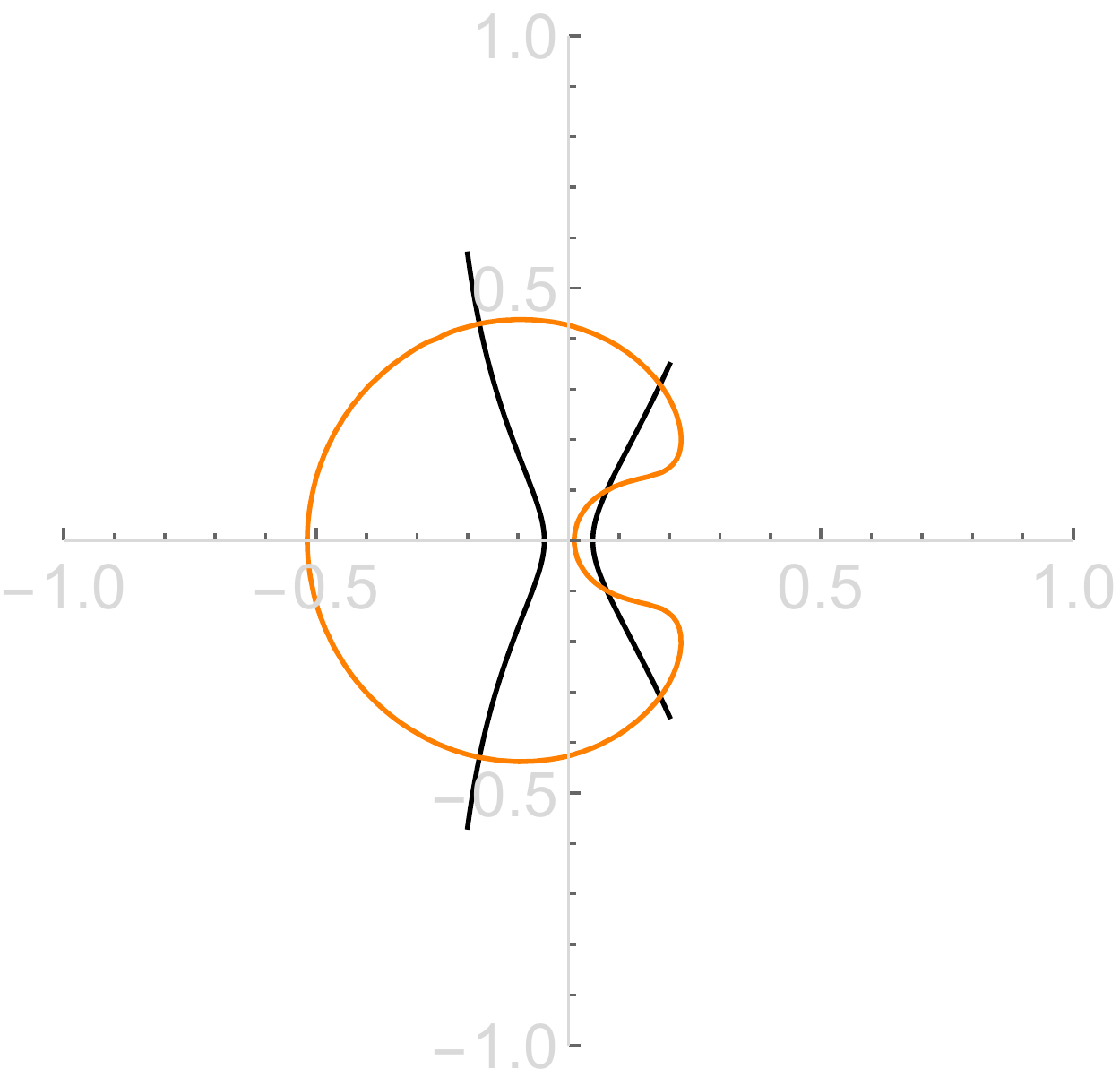} & \includegraphics[width=36mm]{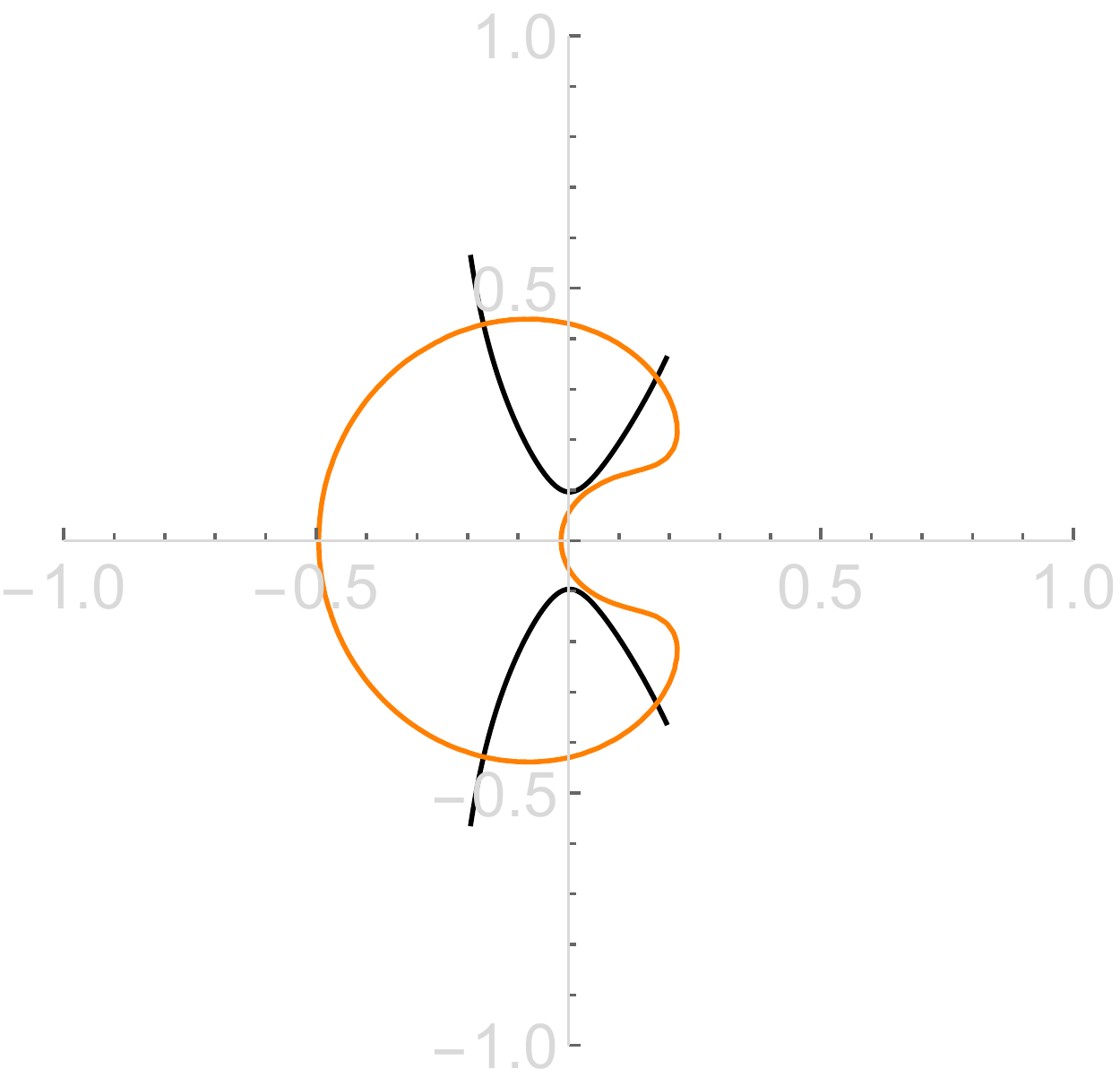} \\
(e) & (f) & (g) & (h)
\end{tabular}
\caption{The $\sigma_L$ spectrum (black) along with the curve corresponding to greatest real part of the $\sigma_{\mathcal{L}}$ spectrum (orange) satisfying (\ref{conic-eqn}). (a) Stokes wave solution, $(k,b)=(0,0.08);$ (b) $\dn$ solution, $(k,b)=(0.9,1);$ (c) $\cn$ solution with piercing, $(k,b)=(0.65,0.4225);$ (d) $\cn$ solution without piercing, $(k,b)=(0.95,0.9025);$ (e) double-figure 8 solution, $(k,b)=(0.65,0.423);$ (f) non-self-intersecting butterfly solution, $(k,b)=(0.9,0.95);$ (g) triple-figure 8 solution, $(k,b)=(0.89,0.84);$ (h) self-intersecting butterfly solution, $(k,b)=(0.9,0.85)$.}
\label{grp-conics}
\end{figure}

By simultaneously solving (\ref{conic-eqn}) and (\ref{intcond4}) and substituting into (\ref{lambdacond}) and (\ref{OmegacondW}) we have an exact expression for $\sigma_{max}$. For the rest of this section we generate series expansions for (\ref{intcond4}) and show that even using low-order approximations we are able to reproduce much of the spectrum, including $\sigma_{max}$.

From Section \ref{regions}, we know a few points of $\sigma_L$ explicitly. Because the functions we are working with are analytic, we can perform series expansions around these explicitly known points. The points we have explicit expressions for are $\zeta_c,$ \textit{i.e.}, the $\zeta$ corresponding to $\lambda = 0$, and $\zeta_t$, the $\zeta$ corresponding to the tops of the figure 8 or triple-figure 8. In what follows we outline a procedure for finding an approximation to points in $S_\Omega$ around these explicitly known points. These expansions provide approximations to the set $S_\Omega$, and using the mapping (\ref{lambdacond}) and (\ref{OmegacondW}), results in approximations to the $\sigma_{\mathcal{L}}$ spectrum.

\newpage

{\it Procedure for finding a series approximation to $\zeta$ satisfying (\ref{intcond4}) around $\zeta_c$:}
\begin{enumerate}
\item Expand the expression inside the real part of (\ref{intcond4}) around $\zeta_c$ in a Puiseux series \cite{D89} to give:
\beq \label{intcond-approx} \textrm{Re}\left( (a_1+b_1 i)(\zeta- \zeta_c)^{1/2}+(a_2+b_2 i)(\zeta- \zeta_c)^{3/2}+(a_3+b_3 i)(\zeta- \zeta_c)^{5/2}+ \ldots\right) = 0, \eeq
where $a_i,b_i\in\mathbb{R}$ are the real and imaginary parts of the coefficients of the terms in the Puiseux series.
\item Let
\beq \label{delta-eqn} \delta = \delta_r+i \delta_i = (\zeta-\zeta_c)^{1/2},\eeq
for $\delta_r,\delta_i\in \mathbb{R}$. Then (\ref{intcond-approx}) becomes
\beq \label{intcond-delta} \textrm{Re}\left( (a_1+b_1 i)\delta+(a_2+b_2 i)\delta^3+(a_3+b_3 i)\delta^5+ O(\delta^7) \right) = 0. \eeq
\item Near $\zeta=\zeta_c,$ $\delta$ is small. Let $\delta = \delta_r(\delta_i) + i \delta_i,$ with
\beq \label{delta-ri} \delta_r(\delta_i) = \delta_1 \delta_i+ \delta_3 \delta_i^3+\delta_5 \delta_i^5+O(\delta_i^7).\eeq
\item Substituting (\ref{delta-ri}) into (\ref{intcond-delta}) and simplifying the expression on the left-hand side, we equate powers of $\delta_i$ to solve for $\delta_1,\delta_3,\delta_5,\ldots$ sequentially. We find
\begin{align}
\delta_1 = &  \frac{b_1}{a_1},\\
\delta_3 = & \frac{3 a_1^2 a_3 b_1-a_3 b_1^3-a_1^3 b_3+3 a_1 b_1^2 b_3}{a_1^4}, \\
\begin{split}
\delta_5 = & \frac{1}{a_1^7} \bigg(a_1^6 b_5-a_1^5 (3 a_3 b_3+5 a_5 b_1)+a_1^4 b_1 \left(9 a_3^2-10 b_1 b_5-6 b_3^2\right)+10 a_1^3 b_1^2 (3 a_3 b_3+a_5 b_1) \\
& +a_1^2 b_1^3 \left(-12 a_3^2+5 b_1 b_5+18 b_3^2\right)-a_1 b_1^4 (15 a_3 b_3+a_5 b_1)+3 a_3^2 b_1^5\bigg),
\end{split} \\
\cdots &\nonumber
\end{align}
\item Solving (\ref{delta-eqn}) for $\zeta$ results in an approximation for $\zeta$ as a function of $\delta_i$ in terms of its real and imaginary parts:
\beq \label{zeta-expansion} \zeta = \delta_r(\delta_i)^2-\delta_i^2+\textrm{Re}(\zeta_c) + \left(2 \delta_r(\delta_i) \delta_i+\textrm{Im}(\zeta_c)\right) i.\eeq
We call (\ref{zeta-expansion}) an $n$th-order expansion where $n$ is the largest power of $\zeta_i$ from (\ref{delta-ri}) included. For instance, a third-order expansion for $\zeta$ is
\beq \label{zeta-expansion3} \zeta =  \left(\delta_1 \delta_i+ \delta_3 \delta_i^3\right)^2-\delta_i^2+\textrm{Re}(\zeta_c)+ \left(2 \left(\delta_1 \delta_i+ \delta_3 \delta_i^3\right) \delta_i+\textrm{Im}(\zeta_c)\right) i.\eeq
\end{enumerate}

\begin{figure}
\centering
\begin{tabular}{cc}
  \includegraphics[width=75mm]{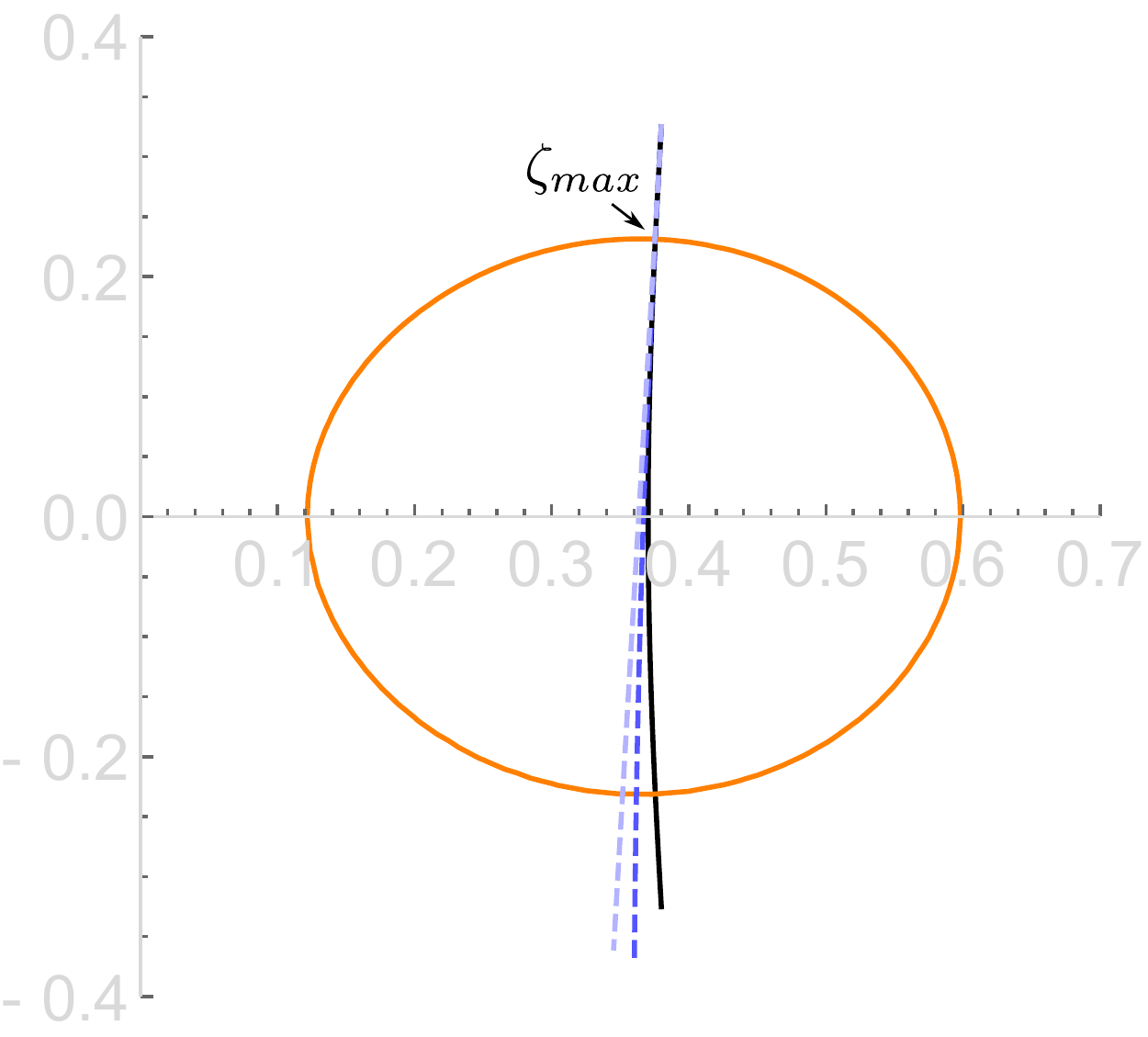} & \includegraphics[width=75mm]{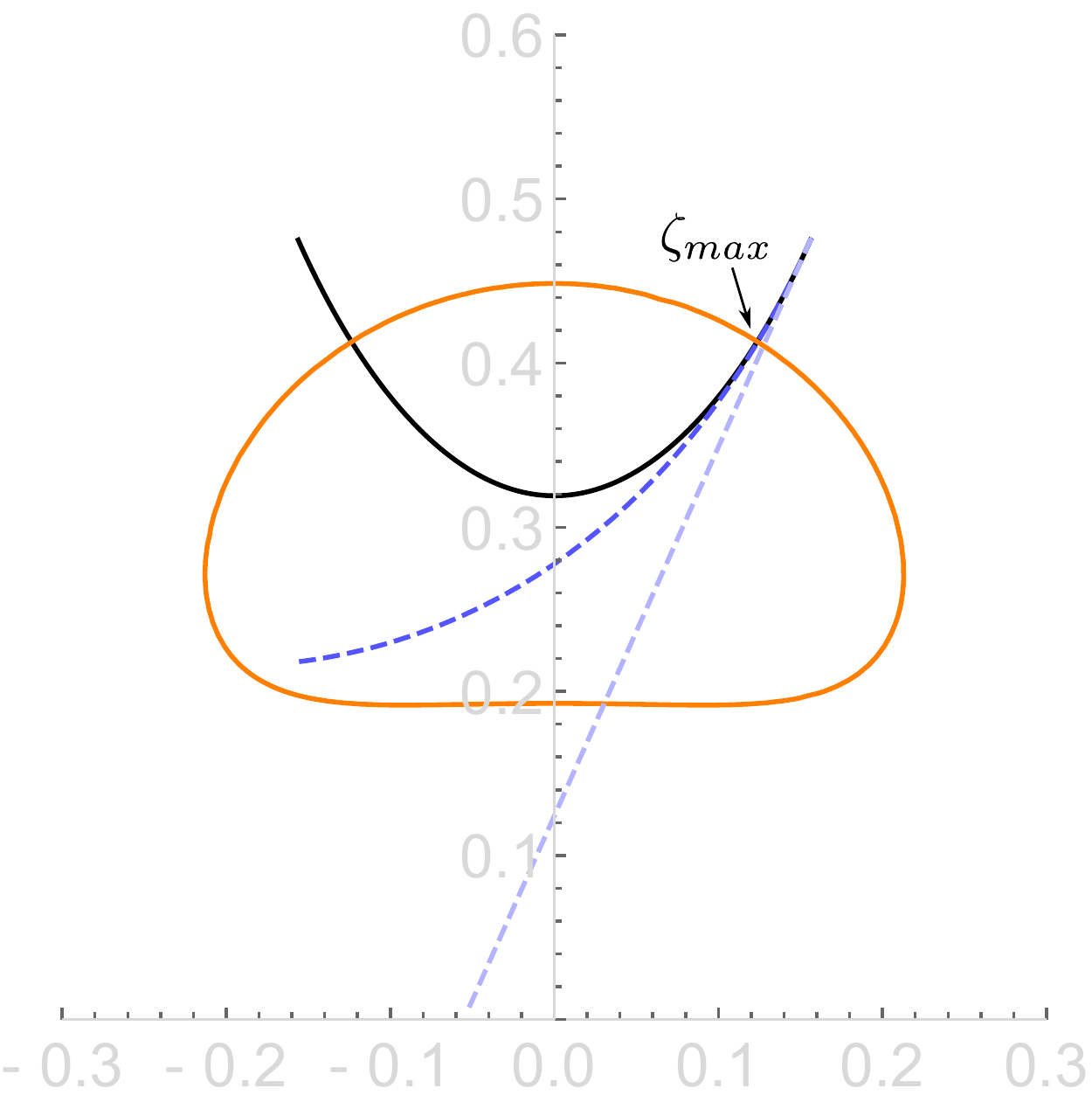} \\
(a) & (b)
\end{tabular}
\caption{Approximating the $\sigma_L$ spectrum for $\cn$ solutions. Shown are the $\sigma_L$ spectrum (black solid curve), the curve corresponding to greatest real part of the $\sigma_{\mathcal{L}}$ spectrum (orange solid curve), $\zeta_{max}$ at the intersection point of the black and orange curves, the first-order approximation to $\sigma_L$ around $\zeta_1$ (light-blue dotted curve), third-order approximation to $\sigma_L$ around $\zeta_1$ (dark-blue dotted curve). (a) A $\cn$ solution with piercing, $(k,b)=(0.65,0.4225);$ (b) $\cn$ solution without piercing, $(k,b)=(0.95,0.9025).$}
\label{grp-approxs}
\end{figure}

First- and third-order approximations to (\ref{intcond4}) around $\zeta_c$ are shown in Figure \ref{grp-approxs} for the two types of $\cn$ solutions. Although the expansion is only guaranteed to be valid around $\zeta_c,$ the first-order expansion approximates $\sigma_L$ well up to (and past) the point where $\sigma_{max}$ occurs. With this in mind, we present Figure \ref{grp-exactvapprox}, comparing the exact value of the greatest real part of the spectrum and the approximate value. From this figure, generally the approximation performs better in the piercing case ($k<k^*$) than in the non-piercing case ($k>k^*$). Also, with just the first-order approximation we get a maximum relative error of less than 18\%.  For third-order, the maximum relative error is less than 1\%, and for fifth-order this decreases to less than 0.1\%.

Using the approximations to $\sigma_{max}$ we can obtain an approximation to the eigenfunction profile with the largest growth rate. This is achieved by substituting $\zeta_{max}$ into (\ref{eigenfunctions}) using (\ref{varphidef}) and (\ref{gammacon}). The approximation for $\zeta_{max}$ does not exactly satisfy (\ref{intcond}) and in order to find a bounded eigenfunction we subtract the left-hand side of (\ref{intcond}) from the exponent in (\ref{gammacon}). Indeed, with $\zeta_{max}$ in Figure~\ref{grp-approxs} the left-hand side of (\ref{intcond}) is small in magnitude. For example, when $k=0.65$, the left-hand side of (\ref{intcond}) is $0.0014$ for the first-order approximation and $0.00034$ for the third-order approximation. These values should be compared with $0.22$ when $\zeta$ is chosen to correspond to a point in the middle of the figure 8.

In addition to expanding around $\zeta_c$, we can also expand (\ref{intcond4}) around $\zeta = \zeta_t$, corresponding to the top of the figure 8 or triple-figure 8. Note that we cannot do so if we are in the butterfly region or in the $\cn$ region without piercing, thus we require
\beq \label{figure8-cond} b+1-k^2-\frac{2 E(k)}{K(k)}<0. \eeq
Since we are expanding around a point where the expression inside the real part of (\ref{intcond4}) is analytic, we can use a Taylor series instead of a Puiseux series which vastly simplifies the analysis.

\begin{figure}
\centering
\begin{tabular}{cc}
  \includegraphics[width=75mm]{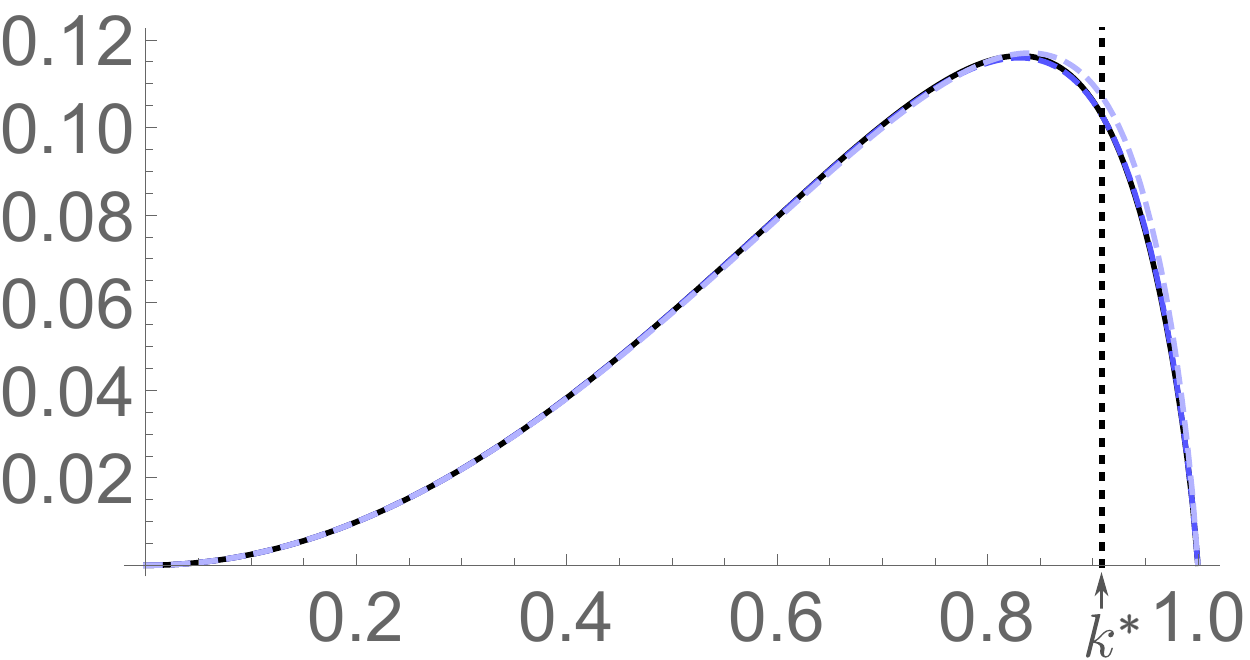} & \includegraphics[width=75mm]{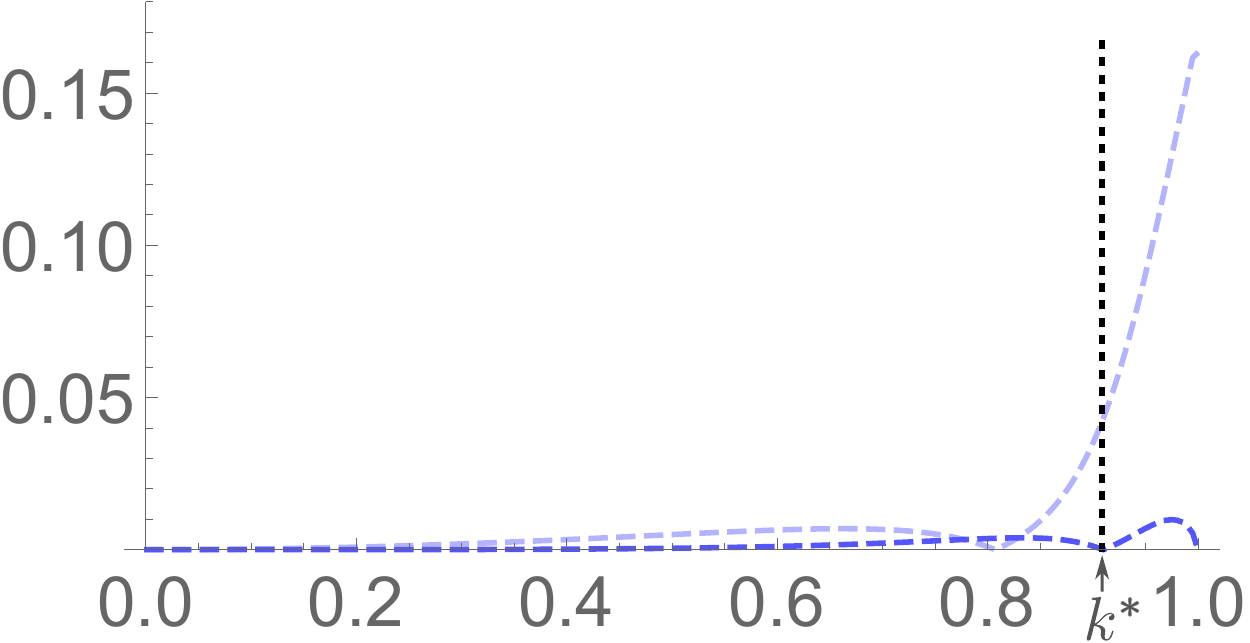} \\
(a) & (b)
\end{tabular}
\caption{(a) Comparison of the exact value for the greatest real part of $\sigma_{\mathcal{L}}$ for $\cn$ solutions (black solid curve) with a first-order approximation (dark blue dotted curve) and a third-order approximation (light blue dotted curve). (b) The relative error of the approximations: $|($approximation-exact)/exact$|$.  }
\label{grp-exactvapprox}
\end{figure}

{\it Procedure for finding an approximation to $\zeta$ satisfying (\ref{intcond4}) around $\zeta_t$:}
\begin{enumerate}
\item Expand the expression inside the real part of (\ref{intcond4}) around $\zeta_t$ in a Taylor series to give
\beq  \textrm{Re}\left[ (a_1+b_1 i)(\zeta- \zeta_t)+(a_2+b_2 i)(\zeta- \zeta_t)^2+(a_3+b_3 i)(\zeta- \zeta_t)^3+ \ldots\right] = 0, \eeq
where $a_i,b_i\in\mathbb{R}$ are the real and imaginary parts of the coefficients of the terms in the Taylor series. In fact, all $a_i$'s are identically zero, and $b_1=0$ so that
\beq \label{intcond-approx2} \textrm{Re}\left[b_2 i(\zeta- \zeta_t)^2+b_3 i(\zeta- \zeta_t)^3+ \ldots\right] = 0.\eeq
\item Let
\beq \label{delta-eqn2} \delta = \delta_r+i \delta_i = \zeta-\zeta_t,\eeq
for $\delta_r,\delta_i\in \mathbb{R}$. Then (\ref{intcond-approx2}) becomes
\beq \label{intcond-delta2} \textrm{Re}\left[b_2 i\delta^2+b_3 i\delta^3+ O(\delta^4)\right] = 0. \eeq
\item Near $\zeta=\zeta_c,$ $\delta$ is small. Let $\delta = \delta_r(\delta_i) + \delta_i,$ with
\beq \label{delta-ri2} \delta_r(\delta_i) = \delta_1 \delta_i+ \delta_2 \delta_i^2+\delta_3 \delta_i^3+O(\delta_i^4).\eeq
\item Substituting (\ref{delta-ri2}) into (\ref{intcond-delta2}) and simplifying the expression on the left-hand side, we equate powers of $\delta_i$ to solve for $\delta_1,\delta_2,\delta_3,\ldots$. We find that $\delta_i=0$ for $i$ odd and
\begin{align}
\delta_2 = & \frac{b_3}{2 b_2},\\
\delta_4 = & \frac{-3b_3^3+8b_2 b_3 b_4-4 b_2^2 b_5}{8 b_2^3}, \\
\delta_6 = & \frac{9 b_3^5-40 b_2 b_4 b_3^3+32 b_2^2 b_5 b_3^2+32 b_2^2 b_4^2 b_3-24 b_2^3 b_6 b_3-16 b_2^3 b_4 b_5+8 b_2r^4 b_7}{16 b_2^5},\\\nonumber
\cdots &
\end{align}
\item Solving (\ref{delta-eqn2}) for $\zeta$ we obtain an approximation for $\zeta$ as a function of $\delta_i$ in terms of its real and imaginary parts:
\beq \label{zeta-expansion2} \zeta = \left( \delta_r(\delta_i)+\zeta_t\right) +i \delta_i.\eeq
As before, call (\ref{zeta-expansion2}) an $n$th-order expansion where $n$ is the largest power of $\zeta_i$ from (\ref{delta-ri2}) included. For instance, a fourth-order approximation for $\zeta$ is
\beq \label{zeta-expansion4} \zeta = \left(\delta_2 \delta_i^2+\delta_4 \delta_i^4+\zeta_t\right) + i \delta_i.\eeq
\end{enumerate}

\begin{figure}
\centering
\begin{tabular}{cccc}
  \includegraphics[width=36mm]{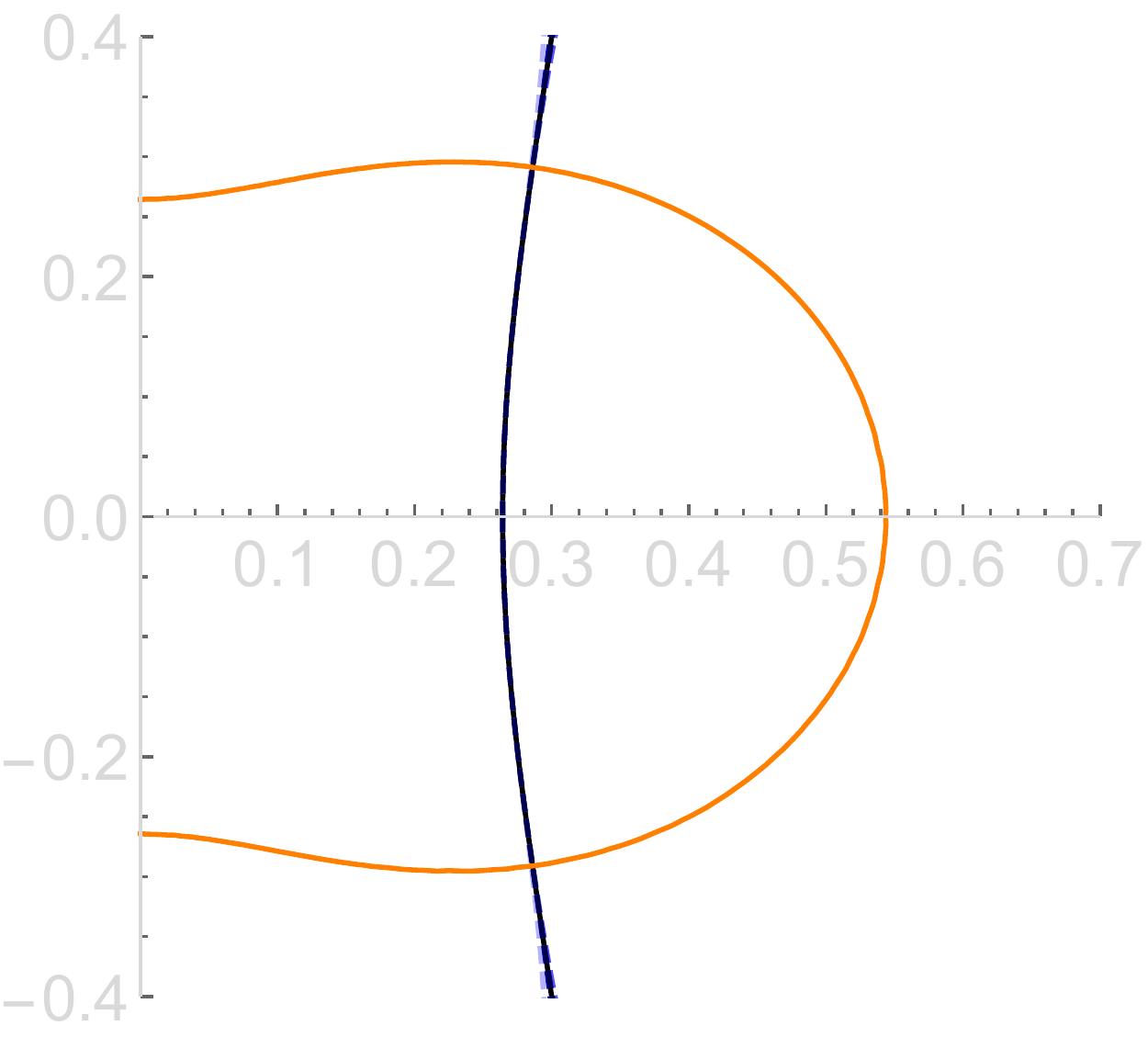} & \includegraphics[width=36mm]{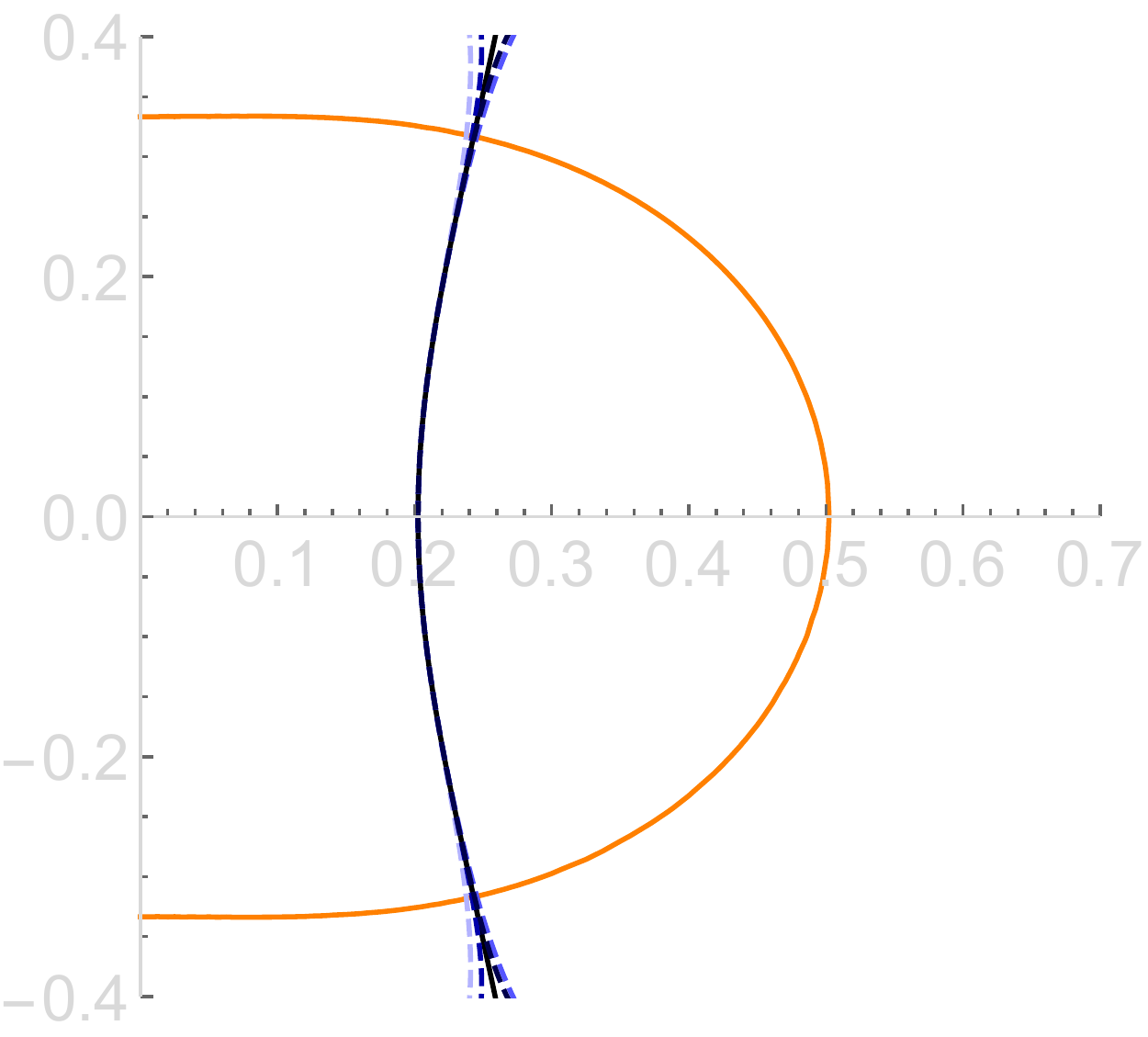} & \includegraphics[width=36mm]{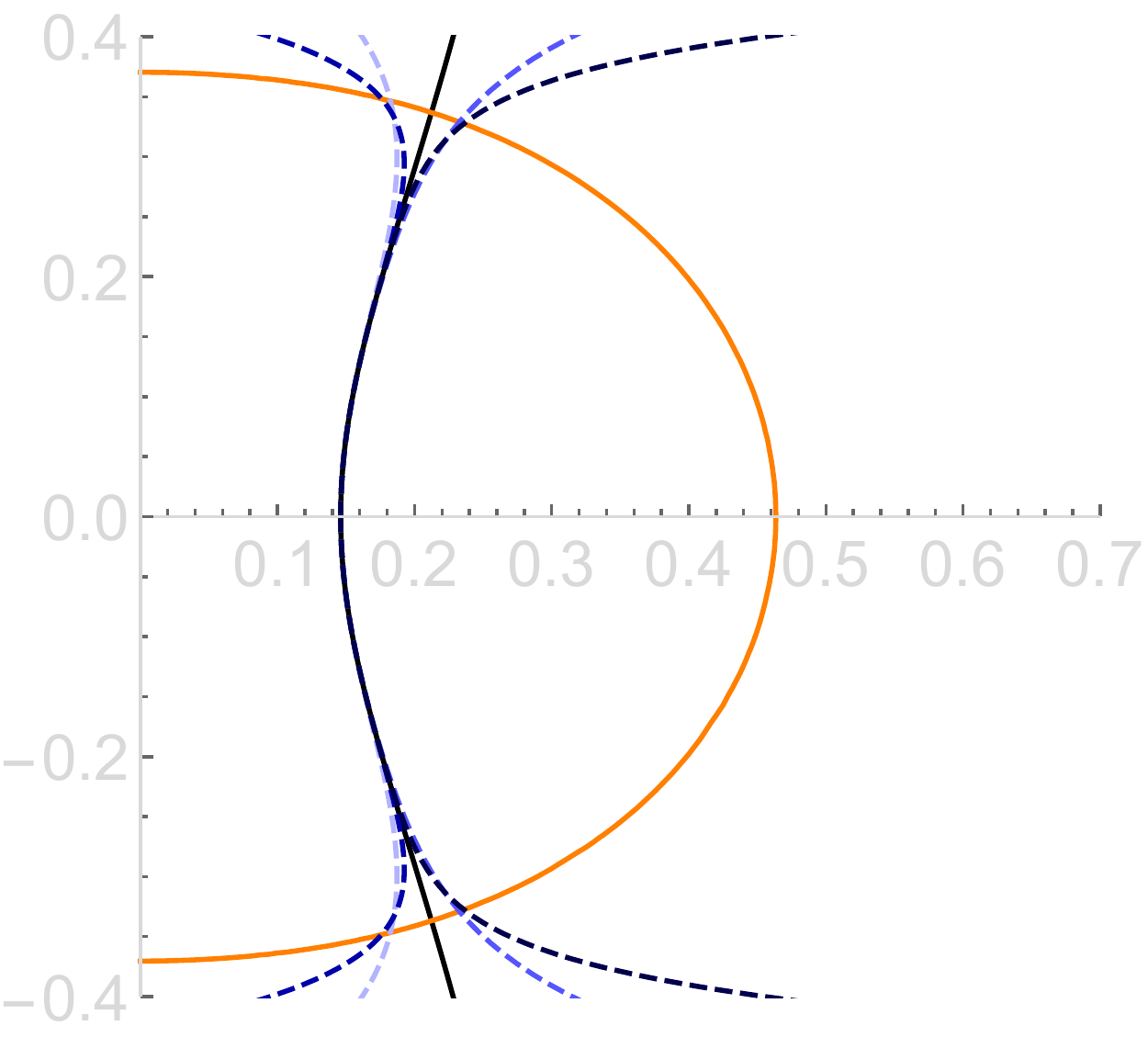} & \includegraphics[width=36mm]{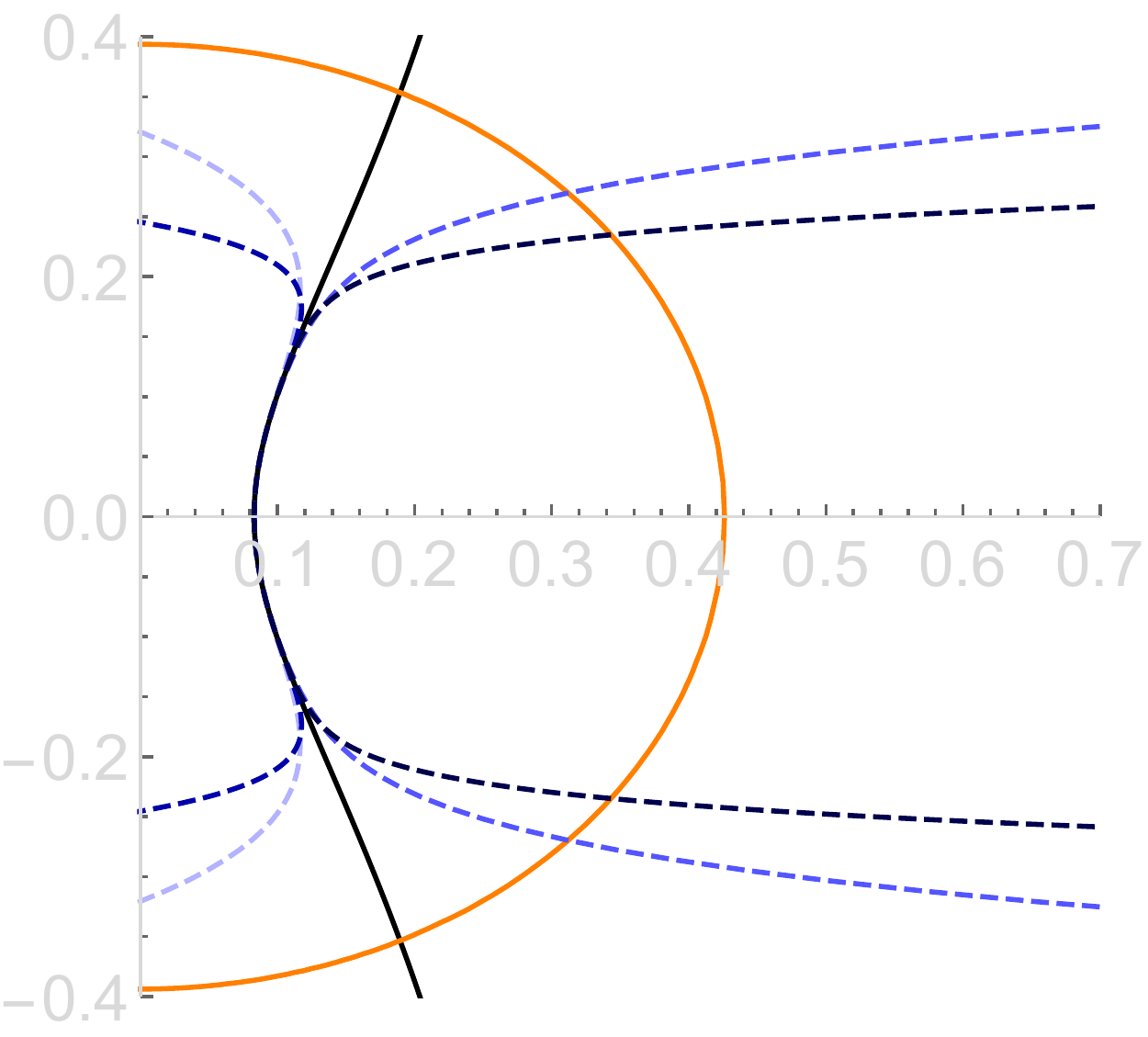} \\
(a) & (b) & (c) & (d)
\end{tabular}
\caption{Approximating $\sigma_L$ for $\cn$ solutions around the top of the figure 8. Shown are $\sigma_L$ (black solid curve), the curve corresponding to the greatest real part of $\sigma_{\mathcal{L}}$ (orange solid curve), the first-order approximation to $\sigma_L$ around $\zeta_1$ (lightest-blue dotted curve), third-order approximation to $\sigma_L$ around $\zeta_1$ (light-blue dotted curve, fifth-order approximation to $\sigma_L$ around $\zeta_1$ (dark-blue dotted curve, seventh-order approximation to $\sigma_L$ around $\zeta_1$ (darkest-blue dotted curve). (a) A $\cn$ solution, $(k,b)=(0.8,0.64);$ (b) $\cn$ solution, $(k,b)=(0.85,0.7225);$ (c) $\cn$ solution, $(k,b)=(0.88,0.7744);$ (d) $\cn$ solution, $(k,b)=(0.9,0.81).$}
\label{grp-approxs2}
\end{figure}

The fourth-, sixth-, eighth-, and tenth-order approximations to (\ref{intcond4}) are shown in Figure \ref{grp-approxs2} for piercing $\cn$ solutions as $k$ approaches $k^*$. We see that these approximations quickly diverge from the $\sigma_L$ spectrum as $k$ approaches $k^*$. The results here are shown for $\cn$ solutions but hold in the nontrivial-phase case as well. For small values of $k$ and $b$ satisfying (\ref{figure8-cond}) we are able to approximate $\sigma_{max}$ well using this Taylor series approach, but as the left-hand side of (\ref{figure8-cond}) approaches $0$ this approximation fails.
In general, the Puiseux expansions around $\zeta_c$ serve as more robust approximations than the Taylor expansions around $\zeta_t$.

\section{Conclusion}

In this paper, we have taken the next step in an ongoing research
program of analyzing the stability of periodic solutions of integrable
equations. Our methods rely on the squared eigenfunction connection
\cite{AKNS} and the existence of an infinite sequence of conserved
quantities, as described below. Thus far, the following results have
been obtained:

\begin{itemize}

\item {\bf The KdV equation.} In [5], the squared eigenfunction
connection was used to establish the spectral stability of the
periodic traveling waves of the KdV equation with respect to
perturbations that are bounded on the whole line (periodic,
quasi-periodic, or linear superpositions of such). This result was
built on in \cite{DK} to establish the orbital stability of
these solutions with respect to subharmonic perturbations of any
period, using an extra conserved quantity as an appropriate Lyapunov
function. This method, employing all conserved quantities, was
extended to establish the orbital stability of the periodic finite-gap
solutions of the equation in \cite{DN}, again with respect to
subharmonic perturbations.

\item {\bf The defocusing mKdV equation.} In \cite{DN}, the method of
\cite{BD} was adapted to the defocusing modified KdV equation to prove the
spectral stability of the periodic traveling waves with respect to
bounded perturbations.

\item {\bf The defocusing NLS equation.} In \cite{BDN}, the
squared eigenfunction connection was employed to show the spectral
stability of the stationary solutions of the defocusing NLS equation.
Orbital stability with respect to subharmonic perturbations is also
demonstrated in \cite{BDN}, again requires the use of an additional conserved
quantity.

\item {\bf The focusing NLS equation.} In this paper, the method of
\cite{BD} and \cite{BDN} is used to examine the stability spectrum of the
stationary solutions of the focusing NLS equation. Because the
underlying Lax pair is not self adjoint, the application of the method
does not simplify as it does for the above equations. Unbridled use of
elliptic function identities allows for the explicit determination of
the spectrum, demonstrating spectral instability for all stationary
(non-soliton) solutions. We demonstrate that the parameter space for
the stationary solution separates in different regions where the
topology of the spectrum is different. An additional subdivision of
this parameter space is found when considering the stability of the
solutions with respect to subharmonic perturbations of a specific
period, leading to the conclusion of spectral stability of some
solutions with respect to some smaller classes of physically relevant
perturbations.

\end{itemize}

Many directions for future research remain. We are currently applying
the same methods to the Sine-Gordon equation \cite{DS},
recovering and extending recent results by Jones, Marangell, Miller
and Plaza \cite{JMMP}. Building on the results found in this manuscript, we
are extending the spectral stability results of Section~\ref{subharmonic} to
orbital stability \cite{DSU}.

\section{Acknowledgments}
This work was supported by the National Science Foundation through grant NSF-DMS-100801 (BD). Benjamin L. Segal acknowledges funding from a Department of Applied Mathematics Boeing fellowship and the Achievement Rewards for College Scientists (ARCS) fellowship. Any opinions, findings, and conclusions or recommendations expressed in this material are those of the authors and do not necessarily reflect the views of the funding sources.

\bibliography{refs}
\bibliographystyle{acm}

\end{document}